%% file: main.tex
\documentclass[11pt]{article}
\input{header.tex}

\usepackage{tikz}
\usepackage{afterpage}

\begin{document}

\newgeometry{margin=0.9in}
\title{Symmetric Exponential Time Requires Near-Maximum Circuit Size}
\ifnum\Anonymity=0
\author{
	Lijie Chen\\ \small{UC Berkeley} \\ \small{\texttt{\href{mailto:lijiechen@berkeley.edu}{lijiechen@berkeley.edu}}}
	\and
	Shuichi Hirahara \\ \small{National Institute of Informatics} \\ \small{\texttt{\href{mailto:s\_hirahara@nii.ac.jp}{s\_hirahara@nii.ac.jp}}}
	\and
	Hanlin Ren\\ \small{University of Oxford} \\ \small{\texttt{\href{mailto:hanlin.ren@cs.ox.ac.uk}{hanlin.ren@cs.ox.ac.uk}}}
}
\fi

\maketitle

\pagenumbering{gobble}

\vspace{-0.75cm}

\input{abstract.tex}

\setcounter{tocdepth}{2}
{\small \tableofcontents}

\restoregeometry

\newpage

\pagenumbering{arabic}
\input{intro.tex}
\input{prelim.tex}
\input{korten.tex}

\input{sigma_2.tex}

\input{S2.tex}

\section*{Acknowledgments}
\ifnum\Anonymity=0
Part of the work was done when all authors were participating in the Meta-Complexity program at the Simons Institute. Lijie Chen is supported by a Miller Research Fellowship. Shuichi Hirahara is supported by  JST, PRESTO Grant Number JPMJPR2024, Japan. Hanlin Ren received support from DIMACS through grant number CCF-1836666 from the National Science Foundation. We thank Oliver Korten, Zhenjian Lu, Igor C.~Oliveira, Rahul Santhanam, Roei Tell, and Ryan Williams for helpful discussions. We also want to thank Jiatu Li, Igor C.~Oliveira, and Roei Tell for comments on an early draft of the paper.
\else
Anonymous acknowledgements.
\fi

{\small \bibliography{main}}

\newpage
\appendix

\newpage
\listoffixmes

\end{document}

%% file: header.tex
\def\Anonymity{0}%

\usepackage{geometry}
\def\PrintMode{0}%

\usepackage[T1]{fontenc}
\usepackage{graphicx}
\usepackage[linesnumbered,boxed,ruled,vlined,algo2e,resetcount,algosection]{algorithm2e}
\usepackage{xspace}
\usepackage[utf8]{inputenc}
\usepackage[noadjust]{cite}
\usepackage{amsmath, amssymb, amsfonts, amsthm,mathtools}
\usepackage{enumerate}
\usepackage{xcolor}
\usepackage[pagebackref=true,colorlinks,linkcolor=magenta,citecolor=blue,bookmarks,bookmarksopen,bookmarksnumbered]{hyperref}
\usepackage{bm}
\usepackage{paralist}
\usepackage{thmtools}
\usepackage{rotating}
\usepackage{mdframed}
\usepackage{multirow}
\usepackage{verbatim}
\usepackage{mathrsfs}
\usepackage{indentfirst}
\usepackage{mleftright}
\usepackage[nottoc]{tocbibind}%
\usepackage{complexity}
\usepackage{tablefootnote}
\usepackage{relsize}
\usepackage{exscale}

\ifnum\PrintMode=1

\usepackage{savetrees}
\else
\fi

\usepackage[capitalize]{cleveref}%

\newtheorem{theorem}{Theorem}[section]
\newtheorem{lemma}[theorem]{Lemma}
\newtheorem{corollary}[theorem]{Corollary}
\newtheorem{claim}[theorem]{Claim}

\newtheorem{fact}[theorem]{Fact}

\newtheorem{open}[theorem]{Open Problem}

\theoremstyle{definition}
\newtheorem{definition}[theorem]{Definition}

\renewcommand{\MAEXP}{\mathsf{MAEXP}}

\theoremstyle{remark}
\newtheorem{remark}[theorem]{Remark}

\newenvironment{claimproof}[1][\proofname]
{\proof[#1]}
{\endproof}

\newenvironment{reminder}[1]{\medskip
	\noindent {\bf Reminder of #1.}\em}{
	\smallskip}

\makeatletter
\def\moverlay{\mathpalette\mov@rlay}
\def\mov@rlay#1#2{\leavevmode\vtop{%
		\baselineskip\z@skip \lineskiplimit-\maxdimen
		\ialign{\hfil$\m@th#1##$\hfil\cr#2\crcr}}}
\newcommand{\charfusion}[3][\mathord]{
	#1{\ifx#1\mathop\vphantom{#2}\fi
		\mathpalette\mov@rlay{#2\cr#3}
	}
	\ifx#1\mathop\expandafter\displaylimits\fi}
\makeatother

\renewcommand{\poly}{\mathrm{poly}}

\renewcommand{\polylog}{\mathrm{polylog}}

\newcommand{\eps}{\varepsilon}

\newcommand{\WT}{\widetilde}

\def\BF{\mathsf{BF}}
\newcommand{\tl}{t^{(\ell)}}
\newcommand{\nell}{n^{(\ell)}}

\newcommand{\Enc}{\mathsf{Enc}}

\newlang{\MCSP}{MCSP}
\newlang{\MFSP}{MFSP}
\newlang{\MKtP}{MKtP}
\newlang{\MKTP}{MKTP}
\newlang{\itrMCSP}{itrMCSP}
\newlang{\itrMKTP}{itrMKTP}
\newlang{\itrMINKT}{itrMINKT}
\newlang{\MINKT}{MINKT}
\newlang{\MINK}{MINK}
\newlang{\MINcKT}{MINcKT}
\newlang{\CMD}{CMD}
\newlang{\DCMD}{DCMD}
\newlang{\CGL}{CGL}
\newlang{\PARITY}{PARITY}
\renewlang{\Gap}{Gap}
\newlang{\Empty}{\textsc{Empty}}
\newlang{\Avoid}{\textsc{Avoid}}
\newlang{\Sparsification}{\textsc{Sparsification}}
\newlang{\HamEst}{\mathsf{HammingEst}}
\newlang{\HamHit}{\mathsf{HammingHit}}
\newlang{\CktEval}{\textsc{Circuit-Eval}}
\newlang{\Hard}{\textsc{Hard}}
\newlang{\cHard}{\textsc{cHard}}
\newlang{\CAPP}{CAPP}
\newlang{\GapUNSAT}{GapUNSAT}
\newlang{\OV}{OV}
\newlang{\PRIMES}{PRIMES}
\renewlang{\PCP}{PCP}
\newlang{\PCPP}{PCPP}
\newclass{\FMA}{FMA}
\newclass{\Avg}{Avg}
\newclass{\ZPEXP}{ZPEXP}
\newclass{\DLOGTIME}{DLOGTIME}
\newclass{\ALOGTIME}{ALOGTIME}
\newclass{\ATIME}{ATIME}%
\newclass{\SZKA}{SZKA}
\newclass{\Laconic}{Laconic\text{-}}
\newclass{\APEPP}{APEPP}
\newclass{\SAPEPP}{SAPEPP}
\newclass{\TFSigma}{TF\Sigma}
\newclass{\NTIMEGUESS}{NTIMEGUESS}
\newclass{\FZPP}{FZPP}

\newlang{\Formula}{Formula}
\newlang{\THR}{THR}
\newlang{\MAJ}{MAJ}
\newlang{\DOR}{DOR}
\newlang{\ETHR}{ETHR}
\newlang{\Midbit}{Midbit}
\newlang{\LCS}{LCS}
\newlang{\TAUT}{TAUT}

\newcommand{\io}{\mathrm{i.o.}\text{-}}

\renewcommand{\K}{\mathrm{K}}

\newcommand{\bitset}{\ensuremath{\{0,1\}}}
\newcommand{\bs}[1]{\ensuremath{\bitset^{#1}}}

\newcommand{\calA}{\mathcal{A}}

\newcommand{\calC}{\mathcal{C}}

\newcommand{\calI}{\mathcal{I}}

\newcommand{\N}{\mathbb{N}}

\newcommand{\F}{\mathbb{F}}

\newcommand{\Range}{\mathrm{Range}}

\newcommand{\GGM}{\mathsf{GGM}}
\newcommand{\GGMEval}{\mathsf{GGMEval}}
\newcommand{\History}{\mathsf{History}}
\newcommand{\istar}{i_{\star}}
\newcommand{\jstar}{j_{\star}}

\newcommand{\Kor}{\mathsf{Korten}}
\newcommand{\In}{\mathsf{Input}}
\newcommand{\Out}{\mathsf{Output}}
\newcommand{\Select}{\mathsf{Select}}

\newcommand{\EncHisKor}{\widetilde{\mathsf{History}}}
\newcommand{\HisKor}{\mathsf{History}}

\newcommand{\bkts}[1]{\mleft[#1\mright]}

\newcommand{\TT}{\mathsf{TT}}

\newcommand{\zeroUpto}[1]{\{0,1,\dotsc,#1\}}

\renewcommand{\tt}{\mathtt{tt}}

\usepackage[inline,marginclue,index]{fixme}  %
\fxsetup{theme=color,mode=multiuser}

\definecolor{color1}{RGB}{46,134,193}
\FXRegisterAuthor{hanlin}{ahanlin}{\colorbox{color1}{\color{white}Hanlin}}

\definecolor{color2}{RGB}{20,60,100}
\FXRegisterAuthor{zhikun}{azhikun}{\colorbox{color2}{\color{white}Zhikun}}
\FXRegisterAuthor{clj}{aclj}{\colorbox{color2}{\color{white}Lijie}}

\definecolor{color3}{RGB}{0,0,255}
\FXRegisterAuthor{igor}{aigor}{\colorbox{color3}{\color{white}Igor}}
\definecolor{color4}{RGB}{255,0,0}
\FXRegisterAuthor{zhenjian}{azhenjian}{\colorbox{color4}{\color{white}Zhenjian}}

\definecolor{color5}{RGB}{0,255,0}
\FXRegisterAuthor{rahul}{arahul}{\colorbox{color5}{\color{white}Rahul}}

\usepackage{aliascnt} %

%% file: abstract.tex
\begin{abstract}
	We show that there is a language in $\S_2\E/_1$ (symmetric exponential time with one bit of advice) with circuit complexity at least $2^n/n$. In particular, the above also implies the same near-maximum circuit lower bounds for the classes $\Sigma_2\E$, $(\Sigma_2\E\cap\Pi_2\E)/_1$, and $\mathsf{ZPE}^\NP/_1$. Previously, only ``half-exponential'' circuit lower bounds for these complexity classes were known, and the smallest complexity class known to require exponential circuit complexity was $\Delta_3\E = \E^{\Sigma_2\P}$ (Miltersen, Vinodchandran, and Watanabe COCOON'99).
 
	Our circuit lower bounds are corollaries of an unconditional zero-error pseudodeterministic algorithm with an $\NP$ oracle and one bit of advice ($\FZPP^\NP/_1$) that solves the range avoidance problem infinitely often. This algorithm also implies unconditional infinitely-often pseudodeterministic $\FZPP^\NP/_1$ constructions for Ramsey graphs, rigid matrices, two-source extractors, linear codes, and $\K^\poly$-random strings with nearly optimal parameters. %

	Our proofs relativize. The two main technical ingredients are (1) Korten's $\P^\NP$ reduction from the range avoidance problem to constructing hard truth tables (FOCS'21), which was in turn inspired by a result of Je\v{r}\'{a}bek on provability in Bounded Arithmetic (Ann.~Pure Appl.~Log.~2004); and (2) the recent iterative win-win paradigm of Chen, Lu, Oliveira, Ren, and Santhanam (FOCS'23).
\end{abstract}

%% file: intro.tex
\section{Introduction}

Proving lower bounds against non-uniform computation (i.e., circuit lower bounds) is one of the most important challenges in theoretical computer science. From Shannon's counting argument \cite{Shannon49,FrandsenM05}, we know that almost all $n$-bit Boolean functions have \emph{near-maximum} ($2^n/n$) circuit complexity.\footnote{All $n$-input Boolean functions can be computed by a circuit of size $(1+\frac{3\log n}{n} + O(\frac{1}{n}))2^n/n$~\cite{lupanov1958synthesis, FrandsenM05}, while most Boolean functions require circuits of size $(1+\frac{\log n}{n}-O(\frac{1}{n}))2^n/n$~\cite{FrandsenM05}. Hence, in this paper, we say an $n$-bit Boolean function has \emph{near-maximum} circuit complexity if its circuit complexity is at least $2^n/n$.} Therefore, the task of proving circuit lower bounds is simply to \emph{pinpoint} one such hard function. More formally, one fundamental question is:

\begin{quote}
	What is the smallest complexity class that contains a language of exponential ($2^{\Omega(n)}$) circuit complexity?
\end{quote}

Compared with super-polynomial lower bounds, exponential lower bounds are interesting in their own right for the following reasons. First, an exponential lower bound would make Shannon's argument \emph{fully constructive}. Second, exponential lower bounds have more applications than super-polynomial lower bounds: For example, if one can show that $\E$ has no $2^{o(n)}$-size circuits, then we would have $\text{pr}\P = \text{pr}\BPP$~\cite{NisanW94,ImpagliazzoW97}, while super-polynomial lower bounds such as $\EXP \not\subset \P/_\poly$ only imply sub-exponential time derandomization of $\text{pr}\BPP$.\footnote{$\E = \DTIME[2^{O(n)}]$ denotes \emph{single-exponential} time and $\EXP = \DTIME[2^{n^{O(1)}}]$ denotes \emph{exponential} time; classes such as $\E^\NP$ and $\EXP^\NP$ are defined analogously. Exponential time and single-exponential time are basically interchangeable in the context of super-polynomial lower bounds (by a padding argument); the exponential lower bounds proven in this paper will be stated for single-exponential time classes since this makes our results stronger. Below, $\Sigma_3\E$ and $\Pi_3\E$ denote the exponential-time versions of $\Sigma_3\P=\NP^{\NP^{\NP}}$ and $\Pi_3\P = \coNP^{\NP^{\NP}}$, respectively.}

Unfortunately, despite its importance, our knowledge about exponential lower bounds is quite limited.~Kannan \cite{Kannan82} showed that there is a function in $\Sigma_3\E\cap\Pi_3\E$ that requires maximum circuit complexity; the complexity of the hard function was later improved to $\Delta_3\E = \E^{\Sigma_2\P}$ by Miltersen, Vinodchandran, and Watanabe \cite{MiltersenVW99}, via a simple binary search argument. This is \textbf{essentially all we know} regarding exponential circuit lower bounds.\footnote{We also mention that Hirahara, Lu, and Ren \cite{HiraharaLR23} recently proved that for every constant $\varepsilon > 0$, $\BPE^\MCSP/_{2^{\varepsilon n}}$ requires near-maximum circuit complexity, where $\MCSP$ is the Minimum Circuit Size Problem \cite{KabanetsC00}. However, the hard function they constructed requires subexponentially ($2^{\varepsilon n}$) many advice bits to describe.} 

We remark that Kannan \cite[Theorem 4]{Kannan82} claimed that $\Sigma_2\E\cap\Pi_2\E$ requires exponential circuit complexity, but \cite{MiltersenVW99} pointed out a gap in Kannan's proof, and suggested that exponential lower bounds for $\Sigma_2\E\cap\Pi_2\E$ were ``reopened and considered an open problem.'' Recently, Vyas and Williams \cite{VyasW23} emphasized our lack of knowledge regarding the circuit complexity of $\Sigma_2\EXP$, even with respect to \emph{relativizing} proof techniques.
In particular, the following question has been open for at least 20 years (indeed, if we count from~\cite{Kannan82}, it would be at least 40 years):

\begin{open}\label{open:Sigma-2}
	Can we prove that $\Sigma_2\EXP \not\subset \SIZE[2^{\eps n}]$ for some absolute constant $\eps > 0$, or at least show a relativization barrier for proving such a lower bound?
\end{open}

\paragraph{The half-exponential barrier.} There is a richer literature regarding super-polynomial lower bounds than exponential lower bounds. Kannan \cite{Kannan82} proved that the class $\Sigma_2\E\cap\Pi_2\E$ does not have polynomial-size circuits. Subsequent works proved super-polynomial circuit lower bounds for exponential-time complexity classes such as $\ZPEXP^\NP$ \cite{KoblerW98, BshoutyCGKT96}, $\S_2\EXP$ \cite{CaiCHO05, Cai01a}, $\PEXP$ \cite{Vinodchandran05, Aaronson06}, and $\MAEXP$ \cite{BuhrmanFT98, Santhanam09}. 

Unfortunately, all these works fail to prove exponential lower bounds. All of their proofs go through certain \emph{Karp--Lipton} collapses \cite{KarpL80}; such a proof strategy runs into a so-called ``half-exponential barrier'', preventing us from getting exponential lower bounds. See~\autoref{sec:rel-half-exp} for a detailed discussion.

\subsection{Our Results}

\subsubsection{New near-maximum circuit lower bounds} 

In this work, we \emph{overcome} the half-exponential barrier mentioned above and resolve~\autoref{open:Sigma-2} by showing that both $\Sigma_2\E$ and $(\Sigma_2\E \cap \Pi_2\E)/_1$ require near-maximum ($2^n/n$) circuit complexity. Moreover, our proof indeed \emph{relativizes}:
\begin{theorem}\label{thm: main Sigma2}
	$\Sigma_2\E\not\subset\SIZE[2^n/n]$ and $(\Sigma_2\E\cap\Pi_2\E)/_1\not\subset\SIZE[2^n/n]$. Moreover, they hold in every relativized world.
\end{theorem}

Up to one bit of advice, we finally provide a proof of Kannan's original claim in~\cite[Theorem~4]{Kannan82}. Moreover, with some more work, we extend our lower bounds to the smaller complexity class $\S_2\E/_1$ (see~\autoref{def: symmetric time} for a formal definition), again with a relativizing proof:

\begin{theorem}
	$\S_2\E/_1 \not\subset\SIZE[2^n/n]$. Moreover, this holds in every relativized world.
\end{theorem}

\paragraph*{The symmetric time class $\S_2\E$.} $\S_2\E$ can be seen as a ``randomized'' version of $\E^\NP$ since it is sandwiched between $\E^\NP$ and $\ZPE^\NP$: it is easy to show that $\E^\NP\subseteq\S_2\E$ \cite{RussellS98}, and it is also known that $\S_2\E \subseteq \ZPE^\NP$ \cite{Cai01a}. We also note that under plausible derandomization assumptions (e.g., $\E^\NP$ requires $2^{\Omega(n)}$-size $\SAT$-oracle circuits), all three classes simply collapse to $\E^\NP$~\cite{kvm98}.

Hence, our results also imply a near-maximum circuit lower bound for the class $\ZPE^\NP/_1 \subseteq (\Sigma_2\E\cap \Pi_2\E)/_1$. This vastly improves the previous lower bound for $\Delta_3\E = \E^{\Sigma_2 \P}$. 

\begin{corollary}
	$\ZPE^{\NP}/_1 \not\subset\SIZE[2^n/n]$. Moreover, this holds in every relativized world.
\end{corollary}

\subsubsection{New algorithms for the range avoidance problem}\label{sec:intro-new-algo-range-avoid}

\paragraph{Background on $\Avoid$.} Actually, our circuit lower bounds are implied by our new algorithms for solving the range avoidance problem ($\Avoid$)~\cite{KKMP21, Korten21, RenSW22}, which is defined as follows: given a circuit $C:\{0, 1\}^n \to \{0, 1\}^{n+1}$ as input, find a string outside the range of $C$ (we define $\Range(C) \coloneqq \{ C(z) : z \in \bs{n} \}$). That is, output any string $y\in\{0, 1\}^{n+1}$ such that for every $x\in\{0, 1\}^n$, $C(x) \ne y$.

There is a trivial $\FZPP^\NP$ algorithm solving $\Avoid$: randomly generate strings $y\in\{0, 1\}^{n+1}$ and output the first $y$ that is outside the range of $C$ (note that we need an $\NP$ oracle to verify if $y \notin \Range(C)$). The class $\APEPP$ (Abundant Polynomial Empty Pigeonhole Principle) \cite{KKMP21} is the class of total search problems reducible to $\Avoid$. 

As demonstrated by Korten \cite[Section 3]{Korten21}, $\APEPP$ captures the complexity of explicit construction problems whose solutions are guaranteed to exist by the probabilistic method (more precisely, the dual weak pigeonhole principle \cite{Krajicek01, Jerabek04}), in the sense that constructing such objects reduces to the range avoidance problem. This includes many important objects in mathematics and theoretical computer science, including Ramsey graphs \cite{Erdos59}, rigid matrices \cite{Valiant77, GuruswamiLW22, GGNS23}, two-source extractors \cite{CZ19, Li23}, linear codes \cite{GuruswamiLW22}, hard truth tables \cite{Korten21}, and strings with maximum time-bounded Kolmogorov complexity (i.e., $\K^\poly$-random strings) \cite{RenSW22}. Hence, derandomizing the trivial $\FZPP^\NP$ algorithm for $\Avoid$ would imply explicit constructions for all these important objects.

\paragraph*{Our results: new pseudodeterministic algorithms for $\Avoid$.} We show that, \emph{unconditionally}, the trivial $\FZPP^\NP$ algorithm for $\Avoid$ can be made \emph{pseudodeterministic} on infinitely many input lengths. A \emph{pseudodeterministic} algorithm \cite{GatG11} is a randomized algorithm that outputs the same \emph{canonical} answer on most computational paths. In particular, we have:

\begin{theorem}\label{cor:ZPP-P-stretch-1}
	For every constant $d\ge 1$, there is a randomized algorithm $\calA$ with an $\NP$ oracle such that the following holds for infinitely many integers $n$. For every circuit $C \colon \bs{n} \to \bs{n+1}$ of size at most $n^d$, there is a string $y_C \in \bs{n} \setminus \Range(C)$ such that $\calA(C)$ either outputs $y_C$ or $\bot$, and the probability (over the internal randomness of $\calA$) that $\calA(C)$ outputs $y_C$ is at least $2/3$. Moreover, this theorem holds in every relativized world.
\end{theorem}

As a corollary, for every problem in $\APEPP$, we obtain zero-error pseudodeterministic constructions with an $\NP$ oracle and one bit of advice ($\FZPP^\NP/_1$) that works infinitely often\footnote{The one-bit advice encodes whether our algorithm succeeds on a given input length; it is needed since on bad input lengths, our algorithm might not be pseudodeterministic (i.e., there may not be a canonical answer that is outputted with high probability).}:

\begin{corollary}[Informal]
	There are infinitely-often zero-error pseudodeterministic constructions for the following objects with an $\NP$ oracle and one-bit of advice: Ramsey graphs, rigid matrices, two-source extractors, linear codes, hard truth tables, and $\K^\poly$-random strings.
\end{corollary}

Actually, we obtain single-valued $\mathsf{F}\S_2\P/_1$ algorithms for the explicit construction problems above (see \autoref{def: single-valued algorithms}), and the pseudodeterministic $\FZPP^\NP/_1$ algorithms follow from Cai's theorem that $\S_2\P\subseteq \ZPP^\NP$ \cite{Cai01a}. We stated them as pseudodeterministic $\FZPP^\NP/_1$ algorithms since this notion is better known than the notion of single-valued $\mathsf{F}\S_2\P/_1$ algorithms.

\autoref{cor:ZPP-P-stretch-1} is tantalizingly close to an infinitely-often $\FP^\NP$ algorithm for $\Avoid$ (with the only caveat of being \emph{zero-error} instead of being completely \emph{deterministic}). However, since an $\FP^\NP$ algorithm for range avoidance would imply near-maximum circuit lower bounds for $\E^\NP$, we expect that it would require fundamentally new ideas to completely derandomize our algorithm. Previously, Hirahara, Lu, and Ren \cite[Theorem 36]{HiraharaLR23} presented an infinitely-often pseudodeterministic $\FZPP^\NP$ algorithm for the range avoidance problem using $n^\varepsilon$ bits of advice, for any small constant $\varepsilon > 0$. Our result improves the above in two aspects: first, we reduce the number of advice bits to $1$; second, our techniques relativize but their techniques do not.

\paragraph*{Lower bounds against non-uniform computation with maximum advice length.} Finally, our results also imply lower bounds against non-uniform computation with maximum advice length. We mention this corollary because it is a stronger statement than circuit lower bounds, and similar lower bounds appeared recently in the literature of super-fast derandomization \cite{ChenT21}. %

\begin{corollary}
	For every $\alpha(n) \ge \omega(1)$ and any constant $k \ge 1$, $\S_2\E/_1 \not\subset \TIME[2^{kn}]/_{2^n - \alpha(n)}$. The same holds for $\Sigma_2\E$, $(\Sigma_2\E\cap\Pi_2\E)/_1$, and $\ZPE^{\NP}/_1$ in place of $\S_2\E/_1$. Moreover, this holds in every relativized world.
\end{corollary}

\subsection{Intuitions}\label{sec:intuitions}

In the following, we present some high-level intuitions for our new circuit lower bounds. 

\subsubsection{Perspective: single-valued constructions}

A key perspective in this paper is to view circuit lower bounds (for exponential-time classes) as \emph{single-valued} constructions of hard truth tables. This perspective is folklore; it was also emphasized in recent papers on the range avoidance problem \cite{Korten21, RenSW22}.

\def\Pihard{\Pi_{\rm hard}}
Let $\Pi \subseteq \{0, 1\}^\star$ be an \emph{$\eps$-dense} property, i.e., for every integer $N\in\N$, $|\Pi_N| \ge \eps\cdot 2^N$. (In what follows, we use $\Pi_N := \Pi \cap \{0, 1\}^N$ to denote the length-$N$ slice of $\Pi$.) As a concrete example, let $\Pihard$ be the set of hard truth tables, i.e., a string $tt \in \Pihard$ if and only if it is the truth table of a function $f:\{0, 1\}^n \to \{0, 1\}$ whose circuit complexity is at least $2^n/n$, where $n := \log N$. (We assume that $n := \log N$ is an integer.) Shannon's argument \cite{Shannon49, FrandsenM05} shows that $\Pihard$ is a $1/2$-dense property. We are interested in the following question:
\begin{quote}
	What is the complexity of \emph{single-valued} constructions for any string in $\Pihard$?
\end{quote}

Here, informally speaking, a computation is \emph{single-valued} if each of its computational paths either fails or outputs the \emph{same} value. For example, an $\NP$ machine $M$ is a single-valued construction for $\Pi$ if there is a ``canonical'' string $y \in \Pi$ such that (1) $M$ outputs $y$ on every accepting computational path; (2) $M$ has at least one accepting computational path. (That is, it is an $\mathsf{NPSV}$ construction in the sense of \cite{BookLS85, FHOS93, Selman94, HNOS96}.) Similarly, a $\BPP$ machine $M$ is a single-valued construction for $\Pi$ if there is a ``canonical'' string $y\in\Pi$ such that $M$ outputs $y$ on most (say $\ge 2/3$ fraction of) computational paths. (In other words, single-valued $\ZPP$ and $\BPP$ constructions are another name for \emph{pseudodeterministic constructions} \cite{GatG11}.)\footnote{Note that the trivial construction algorithms are not single-valued in general. For example, a trivial $\Sigma_2\P = \NP^\NP$ construction algorithm for $\Pihard$ is to guess a hard truth table $tt$ and use the $\NP$ oracle to verify that $tt$ does not have size-$N/\log N$ circuits; however, different accepting computational paths of this computation would output different hard truth tables. Similarly, a trivial $\BPP$ construction algorithm for every dense property $\Pi$ is to output a random string, but there is no \emph{canonical} answer that is outputted with high probability. In other words, these construction algorithms do not \emph{define} anything; instead, a single-valued construction algorithm should \emph{define} some particular string in $\Pi$.}

Hence, the task of proving circuit lower bounds is equivalent to the task of \emph{defining}, i.e., single-value constructing, a hard function, in the smallest possible complexity class. For example, a single-valued $\BPP$ construction (i.e., pseudodeterministic construction) for $\Pihard$ is equivalent to the circuit lower bound $\BPE \not\subset\io\SIZE[2^n/n]$.\footnote{To see this, note that (1) $\BPE \not\subset\io\SIZE[2^n/n]$ implies a simple single-valued $\BPP$ construction for $\Pihard$: given $N = 2^n$, output the truth table of $L_n$ ($L$ restricted to $n$-bit inputs), where $L \in \BPE$ is the hard language not in $\SIZE[2^n/n]$; and (2) assuming a single-valued $\BPP$ construction $A$ for $\Pihard$, one can define a hard language $L$ such that the truth table of $L_n$ is the output of $A(1^{2^n})$, and observe that $L \in \BPE$.} In this regard, the previous near-maximum circuit lower bound for $\Delta_3\E := \E^{\Sigma_2\P}$ \cite{MiltersenVW99} can be summarized in one sentence: The lexicographically first string in $\Pihard$ can be constructed in $\Delta_3\P := \P^{\Sigma_2\P}$ (which is necessarily single-valued).

\paragraph*{Reduction to $\Avoid$.} It was observed in \cite{KKMP21, Korten21} that explicit construction of elements from $\Pi_{\rm hard}$ is a special case of range avoidance: Let $\TT \colon \bs{N-1} \to \bs{N}$ (here $N = 2^n$) be a circuit that maps the description of a $2^n/n$-size circuit into its $2^n$-length truth table (by~\cite{FrandsenM05}, this circuit can be encoded by $N-1$ bits). Hence, a single-valued algorithm solving $\Avoid$ for $\TT$ is equivalent to a single-valued construction for $\Pi_{\rm hard}$. This explains how our new range avoidance algorithms imply our new circuit lower bounds (as mentioned in~\autoref{sec:intro-new-algo-range-avoid}).

In the rest of~\autoref{sec:intuitions}, we will only consider the special case of $\Avoid$ where the input circuit for range avoidance is a $\P$-uniform circuit family. Specifically, let $\{C_n \colon \{0, 1\}^n \to \{0, 1\}^{2n}\}_{n \in \N}$ be a $\P$-uniform family of circuits, where $|C_n| \le \poly(n)$.\footnote{We assume that $C_n$ stretches $n$ bits to $2n$ bits instead of $n+1$ bits for simplicity; Korten \cite{Korten21} showed that there is a $\P^\NP$ reduction from the range avoidance problem with stretch $n+1$ to the range avoidance problem with stretch $2n$.} Our goal is to find an algorithm $A$ such that for infinitely many $n$, $A(1^n) \in \bs{2n} \setminus \Range(C_n)$; see~\autoref{sec:S2E:S2E-lb} and~\autoref{sec:S2E:range-avoid-algo} for how to turn this into an algorithm that works for arbitrary input circuit with a single bit of stretch. Also, since from now on we will not talk about truth tables anymore, we will use $n$ instead of $N$ to denote the input length of $\Avoid$ instances.

\subsubsection{The iterative win-win paradigm of~\texorpdfstring{\cite{ChenLORS23}}{CLORS23}}\label{sec:intro-iter-win-win}

In a recent work,  Chen, Lu, Oliveira, Ren, and Santhanam \cite{ChenLORS23} introduced the \emph{iterative win-win} paradigm for explicit constructions, and used that to obtain a polynomial-time pseudodeterministic construction of primes that works infinitely often. Since our construction algorithm closely follows their paradigm, it is instructive to take a detour and give a high-level overview of how the construction from~\cite{ChenLORS23} works.\footnote{Indeed, for every $1/\poly(n)$-dense property $\Pi \in \P$, they obtained a polynomial-time algorithm $A$ such that for infinitely many $n \in \N$, there exists $y_n \in \Pi_n$ such that $A(1^n)$ outputs $y_n$ with probability at least $2/3$. By~\cite{Agrawal02primesis} and the prime number theorem, the set of $n$-bit primes is such a property.}

In this paradigm, for a (starting) input length $n_0$ and some $t = O(\log n_0)$, we will consider an increasing sequence of input lengths $n_0,n_1,\dotsc,n_t$ (jumping ahead, we will set $n_{i+1} = n_i^\beta$ for a large constant $\beta$), and show that our construction algorithm succeeds on at least one of the input lengths. By varying $n_0$, we can construct infinitely many such sequences of input lengths that are pairwise disjoint, and therefore our algorithm succeeds on infinitely many input lengths. 

\newcommand{\ALG}{\mathsf{ALG}}

In more detail, fixing a sequence of input lengths $n_0,n_1,\dotsc,n_t$ and letting $\Pi$ be an $\eps$-dense property, for each $i \in \zeroUpto{t}$, we specify a (deterministic)  algorithm $\ALG_i$ that takes $1^{n_i}$ as input and aims to construct an explicit element from $\Pi_{n_i}$. We let $\ALG_0$ be the simple brute-force algorithm that enumerates all length-$n_0$ strings and finds the lexicographically first string in $\Pi_{n_0}$; it is easy to see that $\ALG_0$ runs in $T_0 := 2^{O(n_0)}$ time.

\paragraph*{The win-or-improve mechanism.} The core of~\cite{ChenLORS23} is a novel \emph{win-or-improve mechanism}, which is described by a (randomized) algorithm $R$. Roughly speaking, for input lengths $n_i$ and $n_{i+1}$, $R(1^{n_i})$ attempts to \emph{simulate $\ALG_i$ faster by using the oracle $\Pi_{n_{i+1}}$} (hence it runs in $\poly(n_{i+1})$ time). The crucial property is the following win-win argument:

\begin{enumerate}[(\textbf{Win})]
    \item Either $R(1^{n_i})$ outputs  $\ALG_{i}(1^{n_i})$ with probability at least $2/3$ over its internal randomness,\label{quote: win-win in CLORS}
    \item[(\textbf{Improve})] or, from the failure of $R(1^{n_i})$, we can construct an algorithm $\ALG_{i+1}$ that outputs an explicit element from $\Pi_{n_{i+1}}$ and runs in $T_{i+1} = \poly(T_i)$ time.
\end{enumerate}

\def\RefArg{\hyperref[quote: win-win in CLORS]{(Win-or-Improve)}\xspace}

We call the above \RefArg, since either we have a pseudodeterministic algorithm $R(1^{n_i})$ that constructs an explicit element from $\Pi_{n_i}$ in $\poly(n_{i+1}) \le \poly(n_i)$ time (since it simulates $\ALG_i$), or we have an \emph{improved} algorithm $\ALG_{i+1}$ at the input length $n_{i+1}$ (for example, on input length $n_1$, the running time of $\ALG_1$ is $2^{O\mleft(n_{1}^{1/\beta}\mright)} \ll 2^{O(n_1)}$). The \RefArg part in~\cite{ChenLORS23} is implemented via the Chen--Tell targeted hitting set generator \cite{ChenT21b} (we omit the details here). Jumping ahead, in this paper, we will implement a similar mechanism using Korten's $\P^\NP$ reduction from the range avoidance problem to constructing hard truth tables \cite{Korten21}.

\paragraph*{Getting polynomial time.} Now we briefly explain why \RefArg implies a \emph{polynomial-time} construction algorithm. Let $\alpha$ be an absolute constant such that we always have $T_{i+1} \le T_i^\alpha$; we now set $\beta := 2 \alpha$. Recall that $n_i = n_{i-1}^\beta$ for every $i$. The crucial observation is the following:
\begin{quote}
	Although $T_0$ is much larger than $n_0$, the sequence $\{T_i\}$ grows slower than $\{n_i\}$.
\end{quote}
Indeed, a simple calculation shows that when $t = O(\log n_0)$, we will have $T_t \le \poly(n_t)$; see \cite[Section 1.3.1]{ChenLORS23}.

For each $0\le i < t$, if $R(1^{n_i})$ successfully simulates $\ALG_i$, then we obtain an algorithm for input length $n_{i}$ running in $\poly(n_{i+1}) \le \poly(n_i)$ time. Otherwise, we have an algorithm $\ALG_{i+1}$ running in $T_{i+1}$ time on input length $n_{i+1}$. Eventually, we will hit $t$ such that $T_{t} \le \poly(n_{t})$, in which case $\ALG_t$ itself gives a polynomial-time construction on input length $n_t$. Therefore, we obtain a polynomial-time algorithm on at least one of the input lengths $n_0,n_1,\dotsc,n_t$.%

\subsubsection{Algorithms for range-avoidance via Korten's reduction}\label{sec:int:algo-range-avoid}

Now we are ready to describe our new algorithms for $\Avoid$. Roughly speaking, our new algorithm makes use of the iterative win-win argument introduced above, together with an easy-witness style argument~\cite{ImpagliazzoKW02} and Korten's reduction~\cite{Korten21}.\footnote{Korten's result was inspired by~\cite{Jerabek04}, which proved that the dual weak pigeonhole principle is equivalent to the statement asserting the existence of Boolean functions with exponential circuit complexity in a certain fragment of Bounded Arithmetic.} In the following, we introduce the latter two ingredients and show how to chain them together via the iterative win-win argument.

\def\hBF{h_{\sf BF}}

\paragraph*{An easy-witness style argument.} Let $\BF$ be the $2^{O(n)}$-time brute-force algorithm outputting the lexicographically first non-output of $C_{n}$. Our first idea is to consider its \emph{computational history}, a unique $2^{O(n)}$-length string $\hBF$ (that can be computed in $2^{O(n)}$ time), and \emph{branch on whether $\hBF$ has a small circuit or not}. Suppose $\hBF$ admits a, say, $n^{\alpha}$-size circuit for some large $\alpha$, then we apply an \emph{easy-witness-style} argument~\cite{ImpagliazzoKW02} to simulate $\BF$ by a single-valued $\mathsf{F}\Sigma_2\P$ algorithm running in $\poly(n^\alpha) = \poly(n)$ time (see \autoref{sec:int:comp-history}). Hence, we obtained the desired algorithm when $\hBF$ is easy.

However, it is less clear how to deal with the other case (when $\hBF$ is hard) directly. The crucial observation is that we have gained the following ability: we can generate a string $\hBF \in \bs{2^{O(n)}}$ that has circuit complexity at least $n^{\alpha}$, in only $2^{O(n)}$ time.

\paragraph*{Korten's reduction.} We will apply Korten's recent work~\cite{Korten21} to make use of the ``gain'' above. So it is worth taking a detour to review the main result of~\cite{Korten21}.
Roughly speaking,~\cite{Korten21} gives \textbf{an algorithm that uses a hard truth table $f$ to solve a derandomization task: finding a non-output of the given circuit (that has more output bits than input bits).}\footnote{This is very similar to the classical hardness-vs-randomness connection~\cite{NisanW94,ImpagliazzoW97}, which can be understood as an algorithm that uses a hard truth table $f$ (i.e., a truth table without small circuits) to solve another derandomization task: estimating the acceptance probability of the given circuit. This explains why one may want to use Korten's algorithm to replace the Chen--Tell targeted generator construction~\cite{ChenT21b} from~\cite{ChenLORS23}, as they are both hardness-vs-randomness connections.}

Formally,~\cite{Korten21} gives a $\P^\NP$-computable algorithm $\Kor(C, f)$ that takes as inputs a circuit $C \colon \bs{n} \to \bs{2n}$ and a string $f \in \bs{T}$ (think of $n \ll T$), and outputs a string $y \in \bs{2n}$. The guarantee is that if the circuit complexity of $f$ is sufficiently larger than the size of $C$, then the output $y$ is not in the range of $C$.

This fits perfectly with our ``gain'' above: for $\beta \ll \alpha$ and $m = n^{\beta}$, $\Kor(C_{m},\hBF)$ solves $\Avoid$ for $C_m$ since the circuit complexity of $\hBF$, $n^\alpha$, is sufficiently larger than the size of $C_m$. Moreover, $\Kor(C_{m},\hBF)$ runs in only $2^{O(n)}$ time, which is much less than the brute-force running time $2^{O(m)}$. Therefore, we obtain an improved algorithm for $\Avoid$ on input length $m$.

\paragraph*{The iterative win-win argument.} What we described above is essentially the first stage of an \emph{win-or-improve mechanism} similar to that from~\autoref{sec:intro-iter-win-win}. Therefore, we only need to iterate the argument above to obtain a polynomial-time algorithm.%

For this purpose, we need to consider the computational history of not only $\BF$, but also algorithms of the form $\Kor(C, f)$.\footnote{Actually, we need to consider all algorithms $\ALG_i$ defined below and prove the properties of computational history for these algorithms. It turns out that all of $\ALG_i$ are of the form $\Kor(C, f)$ (including $\ALG_0$), so in what follows we only consider the computational history of $\Kor(C, f)$.} For any circuit $C$ and ``hard'' truth table $f$, there is a \emph{unique} ``computational history'' $h$ of $\Kor(C, f)$, and the length of $h$ is upper bounded by $\poly(|f|)$. We are able to prove the following statement akin to the \emph{easy witness lemma} \cite{ImpagliazzoKW02}: if $h$ admits a size-$s$ circuit (think of $s \ll T$), then $\Kor(C, f)$ can be simulated by a single-valued $\mathsf{F}\Sigma_2\P$ algorithm in time $\poly(s)$; see~\autoref{sec:int:comp-history} for details on this argument.\footnote{With an ``encoded'' version of history and more effort, we are able to simulate $\Kor(C, f)$ by a single-valued $\mathsf{F}\S_2\P$ algorithm in time $\poly(s)$, and that is how our $\S_2\E$ lower bound is proved; see~\autoref{sec:int:S2P-algo-range-avoidance} for details.}

Now, following the iterative win-win paradigm of~\cite{ChenLORS23}, for a (starting) input length $n_0$ and some $t = O(\log n_0)$, we consider an increasing sequence of input lengths $n_0,n_1,\dotsc,n_t$, and show that our algorithm $A$ succeeds on at least one of the input lengths (i.e., $A(1^{n_i}) \in \bs{2 n_i} \setminus \Range(C_{n_i})$ for some $i \in \zeroUpto{t}$). For each $i \in \zeroUpto{t}$, we specify an algorithm $\ALG_i$ of the form $\Kor(C_{n_i}, -)$ that aims to solve $\Avoid$ for $C_{n_i}$; in other words, we specify a string $f_i \in \{0, 1\}^{T_i}$ for some $T_i$ and let $\ALG_i := \Kor(C_{n_i}, f_i)$. %

The algorithm $\ALG_0$ is simply the brute force algorithm $\BF$ at input length $n_0$. (A convenient observation is that we can specify an exponentially long string $f_0 \in \{0, 1\}^{2^{O(n_0)}}$ so that $\Kor(C_{n_0}, f_0)$ is equivalent to $\BF = \ALG_0$; see \autoref{fact: brute force for sigma2}.) For each $0\le i < t$, to specify $\ALG_{i+1}$, let $f_{i+1}$ denote the history of the algorithm $\ALG_i$, and consider the following win-or-improve mechanism.

\begin{enumerate}
	\item[(\textbf{Win})] If $f_{i+1}$ admits an $n_{i}^\alpha$-size circuit (for some large constant $\alpha$), by our easy-witness argument, we can simulate $\ALG_{i}$ by a $\poly(n_i)$-time single-valued $\mathsf{F}\Sigma_2\P$ algorithm. 
	
	\item[(\textbf{Improve})] Otherwise $f_{i+1}$ has circuit complexity at least $n_{i}^\alpha$, we plug it into Korten's reduction to solve $\Avoid$ for $C_{n_{i+1}}$. That is, we take $\ALG_{i+1} \coloneqq \Kor(C_{n_{i+1}},f_{i+1})$ as our new algorithm on input length $n_{i+1}$. %
\end{enumerate}

Let $T_i = |f_i|$, then $T_{i+1} \le \poly(T_i)$. By setting $n_{i+1} = n_i^\beta$ for a sufficiently large $\beta$, a similar analysis as~\cite{ChenLORS23} shows that for some $t = O(\log n_0)$ we would have $T_t \le \poly(n_t)$, meaning that $\ALG_{t}$ would be a $\poly(n_t)$-time $\FP^\NP$ algorithm (thus also a single-valued $\mathsf{F}\Sigma_2\P$ algorithm) solving $\Avoid$ for $C_{n_t}$. Putting everything together, we obtain a polynomial-time single-valued $\mathsf{F}\Sigma_2\P$ algorithm that solves $\Avoid$ for at least one of the $C_{n_i}$.

\paragraph*{The hardness condenser perspective.} Below we present another perspective on the construction above which may help the reader understand it better. In the following, we fix $C_n \colon \bs{n} \to \bs{2n}$ to be the truth table generator $\TT_{n,2n}$ that maps an $n$-bit description of a $\log(2n)$-input circuit into its length-$2n$ truth table. Hence, instead of solving $\Avoid$ in general, our goal here is simply \emph{constructing hard truth tables} (or equivalently, proving circuit lower bounds). 

We note that $\Kor(\TT_{n,2n},f)$ can then be interpreted as a \emph{hardness condenser}~\cite{Buresh-OppenheimS06}:\footnote{A hardness condenser takes a long truth table $f$ with certain hardness and outputs a shorter truth table with similar hardness.} Given a truth table $f \in \bs{T}$ whose circuit complexity is sufficiently larger than $n$, it outputs a length-$2n$ truth table that is maximally hard (i.e., without $n/\log n$-size circuits). The win-or-improve mechanism can be interpreted as an iterative application of this hardness condenser. 

At the stage $i$, we consider the algorithm $\ALG_i \coloneqq \Kor(\TT_{n_i,2n_i},f_i)$, which runs in $T_i \approx |f_i|$ time and creates (roughly) $n_i$ bits of hardness. (That is, the circuit complexity of the output of $\ALG_i$ is roughly $n_i$.) In the (\textbf{Win}) case above, $\ALG_i$ admits an $n_i^\alpha$-size history $f_{i+1}$ (with length approximately $|f_{i}|$) and can therefore be simulated in $\mathsf{F}\Sigma_2\P$. The magic is that in the (\textbf{Improve}) case, we actually have access to \emph{much more hardness than $n_i$}: the history string $f_{i+1}$ has $n_i^\alpha \gg n_i$ bits of hardness. So we can \emph{distill} these hardness by applying the condenser to $f_{i+1}$ to obtain a maximally hard truth tables of length $2n_{i+1} = 2n_i^\beta$, establish the next algorithm $\ALG_{i+1} \coloneqq \Kor(\TT_{n_{i+1},2n_{i+1}},f_{i+1})$, and keep iterating.

Observe that the string $f_{i+1}$ above has $n_{i}^\alpha > n_{i}^\beta = n_{i+1}$ bits of hardness. Since $|f_{i+1}| \approx |f_{i}|$ and $n_{i+1} = n_i^\beta$, the process above creates \emph{harder and harder} strings, until $|f_{i+1}| \le n_{i+1} \le n_i^\alpha$, so the (\textbf{Win}) case must happen at some point.

\subsection{Proof Overview}

In this section, we elaborate on the computational history of $\Kor$ and how the easy-witness-style argument gives us $\mathsf{F}\Sigma_2\P$ and $\mathsf{F}\S_2\P$ algorithms.

\subsubsection{Korten's reduction}

We first review the key concepts and results from~\cite{Korten21} that are needed for us. Given a circuit $C\colon\{0, 1\}^n \to \{0, 1\}^{2n}$ and a parameter $T \ge 2n$, Korten builds another circuit $\GGM_T[C]$ stretching $n$ bits to $T$ bits as follows:\footnote{We use the name $\GGM$ because the construction is similar to the pseudorandom function generator of Goldreich, Goldwasser, and Micali \cite{GoldreichGM86}.}
\begin{itemize}
	\item On input $x \in \bs{n}$, we set $v_{0,0} = x$. For simplicity, we assume that $T/n = 2^k$ for some $k \in \N$. We build a full binary tree with $k+1$ layers; see~\autoref{fig:GGM-tree} for an example with $k = 3$.
	
	\item For every $i \in \zeroUpto{k-1}$ and $j \in \zeroUpto{2^i - 1}$, we set $v_{i+1,2j}$ and $v_{i+1,2j+1}$ to be the first $n$ bits and the last $n$ bits of $C(v_{i,j})$, respectively.
	
	\item The output of $\GGM_T[C](x)$ is defined to be the concatenation of $v_{k,0},v_{k,1},\dotsc,v_{k,2^k-1}$.
\end{itemize}

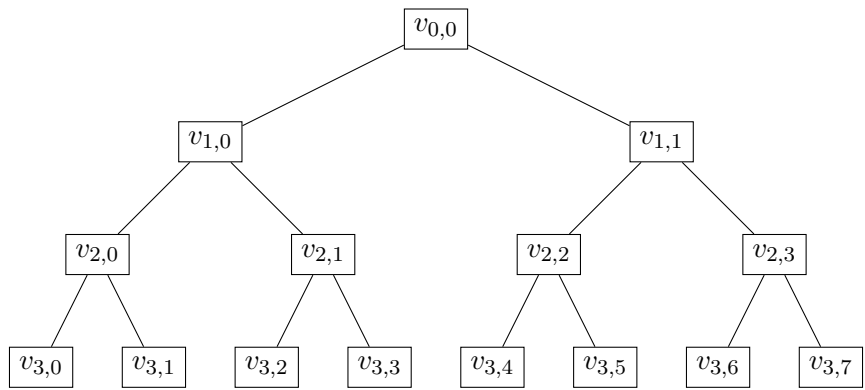
\begin{figure}[h]
	\begin{center}
			\begin{tikzpicture}[
					level distance=1.5cm,
					level 1/.style={sibling distance=6cm},
					level 2/.style={sibling distance=3cm},
					level 3/.style={sibling distance=1.5cm},
					every node/.style={rectangle,draw}
					]
					
					\node {$v_{0,0}$}
					child {node {$v_{1,0}$}
							child {node {$v_{2,0}$}
									child {node {$v_{3,0}$}}
									child {node {$v_{3,1}$}}
								}
							child {node {$v_{2,1}$}
									child {node {$v_{3,2}$}}
									child {node {$v_{3,3}$}}
								}
						}
					child {node {$v_{1,1}$}
							child {node {$v_{2,2}$}
									child {node {$v_{3,4}$}}
									child {node {$v_{3,5}$}}
								}
							child {node {$v_{2,3}$}
									child {node {$v_{3,6}$}}
									child {node {$v_{3,7}$}}
								}
						};
					
				\end{tikzpicture}
		\end{center}
	\caption{An illustration of the GGM Tree, in which, for instance, it holds that $(v_{3,4},v_{3,5}) = C(v_{2,2})$.}
	\label{fig:GGM-tree}
\end{figure}

The following two properties of $\GGM_T[C]$ are established in~\cite{Korten21}, which will be useful for us:
\begin{enumerate}
	\item Given $i \in [T], C$ and $x \in \bs{n}$, by traversing the tree from the root towards the leaf with the $i$-th bit, one can compute the $i$-th bit of $\GGM_T[C](x)$ in $\poly(\SIZE(C), \log T)$ time. Consequently, for every $x$, $\GGM_T[C](x)$ has circuit complexity at most $\poly(\SIZE(C), \log T)$.
	
	\item There is a $\P^\NP$ algorithm $\Kor(C,f)$ that takes an input $f \in \bs{T} \setminus \Range(\GGM_{T}[C])$ and outputs a string $u \in \bs{2n} \setminus \Range(C)$. Note that this is a reduction from solving $\Avoid$ for $C$ to solving $\Avoid$ for $\GGM_{T}[C]$.
\end{enumerate}

In particular, letting $f$ be a truth table whose circuit complexity is sufficiently larger than $\SIZE(C)$, by the first property above, it is not in $\Range(\GGM_{T}[C])$, and therefore $\Kor(C,f)$ solves $\Avoid$ for $C$. This confirms our description of $\Kor$ in~\autoref{sec:intro-new-algo-range-avoid}.

\subsubsection{Computational history of \texorpdfstring{$\Kor$}{Korten} and an easy-witness argument for \texorpdfstring{$\mathsf{F}\Sigma_2\P$}{FS2P} algorithms} \label{sec:int:comp-history}

The algorithm $\Kor(C,f)$ works as follows: we first view $f$ as the labels of the last layer of the binary tree, and try to reconstruct the whole binary tree, layer by layer (start from the bottom layer to the top layer, within each layer, start from the rightmost node to the leftmost one), by filling the labels of the intermediate nodes. To fill $v_{i,j}$, we use an $\NP$ oracle to find the lexicographically first string $u \in \bs{n}$ such that $C(u) = v_{i+1,2j} \circ v_{i+1,2j+1}$, and set $v_{i,j} = u$. If no such $u$ exists, the algorithm stops and report $v_{i+1,2j} \circ v_{i+1,2j+1}$ as the solution to $\Avoid$ for $C$. Observe that this reconstruction procedure must stop somewhere, since if it successfully reproduces all the labels in the binary tree, we would have $f = \GGM_T[C](v_{0,0}) \in \Range(\GGM_{T}[C])$, contradicting the assumption. See~\autoref{lemma: Korten's reduction} for details.

\paragraph*{The computational history of $\Kor$.} The algorithm described above induces a natural description of the computational history of $\Kor$, denoted as $\HisKor(C,f)$, as follows: the index $(\istar,\jstar)$ when the algorithm stops (i.e., the algorithm fails to fill in $v_{\istar,\jstar}$) concatenated with the labels of all the nodes generated by $\Kor(C,f)$ (for the intermediate nodes with no label assigned, we set their labels to a special symbol $\bot$); see~\autoref{fig:GGM-tree-history} for an illustration. This history has length at most $5T$, and for convenience, we pad additional zeros at the end of it so that its length is exactly $5T$.

\begin{figure}[h]
	\begin{center}
		\begin{tikzpicture}[
			level distance=1.5cm,
			level 1/.style={sibling distance=6cm},
			level 2/.style={sibling distance=3cm},
			level 3/.style={sibling distance=1.5cm},
			every node/.style={rectangle,draw}
			]
			
			\node {$\bot$}
			child {node {$\bot$}
				child {node {$\bot$}
					child {node {$v_{3,0}$}}
					child {node {$v_{3,1}$}}
				}
				child {node[line width=0.5mm] {$\bot$}
					child {node {$v_{3,2}$}}
					child {node {$v_{3,3}$}}
				}
			}
			child {node {$\bot$}
				child {node {$v_{2,2}$}
					child {node {$v_{3,4}$}}
					child {node {$v_{3,5}$}}
				}
				child {node {$v_{2,3}$}
					child {node {$v_{3,6}$}}
					child {node {$v_{3,7}$}}
				}
			};
			
			\node at (3.5,0) {$(\istar,\jstar) = (2,1)$}; %
			
		\end{tikzpicture}
	\end{center}
	\caption{An illustration of the history of $\Kor(C,f)$. Here we have $\HisKor(C,f) = (2,1) \circ \bot\bot\bot\bot\bot \circ v_{2,2} \circ v_{2,3} \circ v_{3,0} \circ \dotsc \circ v_{3,7}$ and $\Kor(C,f) = v_{3,2} \circ v_{3,3}$.}
	\label{fig:GGM-tree-history}
\end{figure}
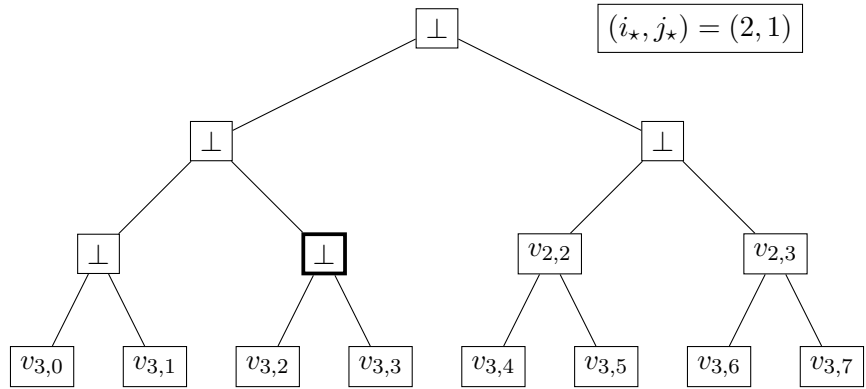

\paragraph*{A local characterization of $\HisKor(C,f)$.} The crucial observation we make on $\HisKor(C,f)$ is that it admits a local characterization in the following sense: there is a family of local constraints $\{\psi_{x}\}_{x \in \bs{\poly(n)}}$, where each $\psi_{x} \colon \bs{5T} \times \bs{T} \to \bs{}$ reads only $\poly(n)$ many bits of its input (we think about it as a local constraint since usually $n \ll T$), such that for fixed $f$, $\HisKor(C,f) \circ f$ is the unique string making all the $\psi_x$ outputting $1$. 

The constraints are follows: (1) for every leaf node $v_{k,i}$, its content is consistent with the corresponding block in $f$; (2) all labels at or before node $(\istar,\jstar)$ are $\bot$;\footnote{We say that $(i,j)$ is before (after) $(\istar,\jstar)$ if the pair $(i,j)$ is lexicographically smaller (greater) than $(\istar,\jstar)$.} (3) for every $z \in \bs{n}$, $C(z) \ne v_{\istar+1,2\jstar}\circ v_{\istar+1,2\jstar+1}$ (meaning the algorithm fails at $v_{\istar,\jstar}$); (4) for every $(i,j)$ after $(\istar,\jstar)$, $C(v_{i,j}) = v_{i+1,2j} \circ v_{i+1,2j+1}$ ($v_{i,j}$ is the correct label); (5) for every $(i,j)$ after $(\istar,\jstar)$ and for every $v' < v_{i,j}$, $C(v') \ne v_{i+1,2j} \circ v_{i+1,2j+1}$ ($v_{i,j}$ is the lexicographically first correct label). It is clear that each of these constraints above only reads $\poly(n)$ many bits from the input and a careful examination shows they precisely \textbf{define} the string $\HisKor(C,f)$.

A more intuitive way to look at these local constraints is to treat them as a $\poly(n)$-time oracle algorithm $V_{\HisKor}$ that takes a string $x \in \poly(n)$ as input and two strings $h \in \bs{5 T}$ and $f \in \bs{T}$ as oracles, and we simply let $V_{\HisKor}^{h,f}(x) = \psi_{x}(h \circ f)$. Since the constraints above are all very simple and only read $\poly(n)$ bits of $h \circ f$, $V_{\HisKor}$ runs in $\poly(n)$ time. In some sense, $V_{\HisKor}$ is a local $\Pi_1$ verifier: it is local in the sense that it only queries $\poly(n)$ bits from its oracles, and it is $\Pi_1$ since it needs a universal quantifier over $x \in \bs{\poly(n)}$ to perform all the checks.

\paragraph*{$\mathsf{F}\Sigma_2 \P$ algorithms.} Before we proceed, we give a formal definition of a single-valued $\mathsf{F}\Sigma_2 \P$ algorithm $A$. Here $A$ is implemented by an algorithm $V_A$ taking an input $x$ and two $\poly(|x|)$-length witnesses $\pi_1$ and $\pi_2$. We say $A(x)$ outputs a string $z \in \bs{\ell}$ (we assume $\ell = \ell(x)$ can be computed in polynomial time from $x$) if $z$ is the \emph{unique} length-$\ell$ string such that the following hold:
\begin{itemize}
	\item there exists $\pi_1$ such that for every $\pi_2$, $V_{\HisKor}(x,\pi_1,\pi_2,z) = 1$.\footnote{Note that our definition here is different from the formal definition we used in~\autoref{def: single-valued algorithms}. But from this definition, it is easier to see why $\mathsf{F}\Sigma_2\P$ algorithms for constructing hard truth tables imply circuit lower bounds for $\Sigma_2\E$.}
\end{itemize}

We can view $V_{\HisKor}$ as a verifier that checks whether $z$ is the desired output using another universal quantifier: given a proof $\pi_1$ and a string $z \in \bs{\ell}$. $A$ accepts $z$ if and only if \emph{for every} $\pi_2$, $V_{\HisKor}(x,\pi_1,\pi_2,z) = 1$. That is, $A$ can perform exponentially many checks on $\pi_1$ and $z$, each taking $\poly(|x|)$ time.

\paragraph*{The easy-witness argument.} Now we are ready to elaborate on the easy-witness argument mentioned in~\autoref{sec:intro-new-algo-range-avoid}. Recall that at stage $i$, we have $\ALG_i = \Kor(C_{n_i},f_i)$ and $f_{i+1} = \HisKor(C_{n_i},f_i)$ (the history of $\ALG_i$). Assuming that $f_{i+1}$ admits a $\poly(n_{i})$-size circuit, we want to show that $\Kor(C_{n_i},f_i)$ can be simulated by a $\poly(n_i)$-time single-valued $\mathsf{F}\Sigma_2 \P$ algorithm.

Observe that for every $t \in [i+1]$, $f_{t-1}$ is simply a substring of $f_t$ since $f_t = \HisKor(C_{n_{t-1}},f_{t-1})$. Therefore, $f_{i+1}$ admitting a $\poly(n_{i})$-size circuit implies that all $f_{t}$ admit $\poly(n_{i})$-size circuits for $t \in [i]$. We can then implement $A$ as follows: the proof $\pi_1$ is a $\poly(n_i)$-size circuit $C_{i+1}$ supposed to compute $f_{i+1}$, from which one can obtain in polynomial time a sequence of circuits $C_{1},\dotsc,C_{i}$ that are supposed to compute $f_1,\dotsc,f_{i}$, respectively. (Also, from \autoref{fact: brute force for sigma2}, one can easily construct a $\poly(n_0)$-size circuit $C_0$ computing $f_0$.) Next, for every $t \in \zeroUpto{i}$, $A$ checks whether $\tt(C_{t+1}) \circ \tt(C_t)$ satisfies all the local constraints $\psi_x$'s from the characterization of $\HisKor(C_{n_t},f_t)$. In other words, $A$ checks whether $V_{\HisKor}^{C_{t+1},C_{t}}(x) = 1$ for all $x \in \bs{\poly(n_t)}$.

The crucial observation is that since all the $C_t$ have size $\poly(n_i)$, each check above can be implemented in $\poly(n_i)$ time as they only read at most $\poly(n_i)$ bits from their input, despite that $\tt(C_{t+1}) \circ \tt(C_t)$ itself can be much longer than $\poly(n_i)$. Assuming that all the checks of $A$ above are passed, by induction we know that $f_{t+1} = \HisKor(C_{n_t},f_t)$ for every $t \in \zeroUpto{i}$. Finally, $A$ checks whether $z$ corresponds to the answer described in $\tt(C_{i+1}) = f_{i+1}$.

\subsubsection{Selectors and an easy-witness argument for \texorpdfstring{$\mathsf{F}\S_2\P$}{FS2P} algorithms}\label{sec:int:S2P-algo-range-avoidance}

Finally, we discuss how to implement the easy-witness argument above with a single-valued $\mathsf{F}\S_2\P$ algorithm. It is known that any single-valued $\mathsf{F}\S_2 \BPP$ algorithm can be converted into an equivalent single-valued $\mathsf{F}\S_2\P$ algorithm outputting the same string~\cite{Canetti96, RussellS98} (see also the proof of~\autoref{theo:S-2-algo-for-range-avoidance} for a self-contained argument). Therefore, in the following we aim to give a single-valued $\mathsf{F}\S_2\BPP$ algorithm for solving range avoidance, which is easier to achieve.

\paragraph*{$\mathsf{F}\S_2\BPP$ algorithms and randomized selectors.} Before we proceed, we give a formal definition of a single-valued $\mathsf{F}\S_2 \BPP$ algorithm $A$. We implement $A$ by a randomized algorithm $V_A$ that takes an input $x$ and two $\poly(|x|)$-length witnesses $\pi_1$ and $\pi_2$.\footnote{$\mathsf{F}\S_2\P$ algorithms are the special case of $\mathsf{F}\S_2\BPP$ algorithms where the algorithm $V_A$ is \emph{deterministic}.} We say that $A(x)$ outputs a string $z \in \bs{\ell}$ (we assume $\ell = \ell(x)$ can be computed in polynomial time from $x$) if the following hold:
\begin{itemize}
	\item there exists a string $h$ such that for every $\pi$, both $V_A(x,h,\pi)$ and $V_A(x,\pi,h)$ output $z$ with probability at least $2/3$. (Note that such $z$ must be unique if it exists.)
\end{itemize}

Actually, our algorithm $A$ will be implemented as a randomized \emph{selector}: given two potential proofs $\pi_1$ and $\pi_2$, it first selects the correct one and then outputs the string $z$ induced by the correct proof.\footnote{If both proofs are correct or neither proofs are correct, it can select an arbitrary one. The condition only applies when exactly one of the proofs is correct.}

\newcommand{\VS}{V_{\mathsf{select}}}

\paragraph{Recap.} Revising the algorithm in~\autoref{sec:int:algo-range-avoid}, our goal now is to give an $\mathsf{F}\S_2 \BPP$ simulation of $\Kor(C_{n_i},f_i)$, assuming that $\HisKor(C_{n_i},f_{i})$ admits a small circuit. Similar to the local $\Pi_1$ verifier used in the case of $\mathsf{F}\Sigma_2\P$ algorithms, now we consider a local randomized selector $\VS$ which takes oracles $\pi_1,\pi_2 \in \bs{5T}$ and $f \in \bs{T}$ such that if exactly one of the $\pi_1$ and $\pi_2$ is $\HisKor(C,f)$, $\VS$ outputs its index with high probability.

Assuming that $f_{i+1} = \HisKor(C_{n_i},f_{i})$ admits a small circuit, one can similarly turn $\VS$ into a single-valued $\mathsf{F}\S_2 \BPP$ algorithms $A$ computing $\Kor(C_{n_i},f_i)$: treat two proofs $\pi_1$ and $\pi_2$ as two small circuits $C$ and $D$ both supposed to compute $f_{i+1}$, from $C$ and $D$ we can obtain a sequence of circuits $\{C_t\}$ and $\{D_t\}$ supposed to compute the $f_t$ for $t \in [i]$. Then we can use the selector $\VS$ to decide for each $t \in [i+1]$ which of the $C_t$ and $D_t$ is the correct circuit for $f_t$. Finally, we output the answer encoded in the selected circuit for $f_{i+1}$; see the proof of~\autoref{theo:S-2-algo-for-range-avoidance} for details.\footnote{However, for the reasons to be explained below, we will actually work with the encoded history instead of the history, which entails a lot of technical challenges in the actual proof.}

\paragraph{Observation: it suffices to find the first differing node label.} Ignore the $(\istar,\jstar)$ part of the history for now. Let $\{v^{1}_{i,j}\}$ and $\{v^{2}_{i,j}\}$ be the node labels encoded in $\pi_1$ and $\pi_2$, respectively. We also assume that exactly one of them corresponds to the correct node labels in $\HisKor(C,f)$. The crucial observation here is that, since the correct node labels are generated by a deterministic procedure \emph{node by node} (from bottom to top and from rightmost to leftmost), it is possible to tell which of the $\{v^{1}_{i,j}\}$ and $\{v^{2}_{i,j}\}$ is correct given the largest $(i',j')$ such that $v^1_{i',j'} \ne v^2_{i',j'}$. (Note that since all $(i,j)$ are processed by $\Kor(C,f)$ in reverse lexicographic order, this $(i',j')$ corresponds to the first node label that the wrong process differs from the correct process, so we call this the first differing point.)

In more detail, assuming we know this $(i',j')$, we proceed by discussing several cases. First of all, if $(i',j')$ corresponds to a leaf, then one can query $f$ to figure out which of $v^{1}_{i',j'}$ and $v^{2}_{i',j'}$ is consistent with the corresponding block in $f$. Now we can assume $(i',j')$ corresponds to an intermediate node. Since $(i',j')$ is the first differing point, we know that $v^1_{i'+1,2j'} \circ v^1_{i'+1,2j'+1} = v^2_{i'+1,2j'} \circ v^2_{i'+1,2j'+1}$ (we let this string to be $\alpha$ for convenience). By the definition of $\HisKor(C,f)$, it follows that the correct $v_{i',j'}$ should be uniquely determined by $\alpha$, which means the selector only needs to read $\alpha$, $v^{1}_{i',j'}$, and $v^{2}_{i',j'}$, and can then be implemented by a somewhat tedious case analysis (so it is local). We refer readers to the proof of~\autoref{lemma:Kor-prop-s2} for the details and only highlight the most illuminating case here: if both of $v^{1}_{i',j'}$ and $v^{2}_{i',j'}$ are good (we say a string $\gamma$ is good, if $\gamma \ne \bot$ and $C(\gamma) = \alpha$), we select the lexicographically smaller one. To handle the $(\istar,\jstar)$ part, one needs some additional case analysis. We omit the details here and refer the reader to the proof of~\autoref{lemma:Kor-prop-s2}.

The takeaway here is that if we can find the first differing label $(i',j')$, then we can construct the selector $\VS$ and hence the desired single-valued $\mathsf{F}\S_2\BPP$ algorithm.

\paragraph{Encoded history.} However, the above assumes the knowledge of $(i',j')$. In general, if one is only given oracle access to $\{v^{1}_{i,j}\}$ and $\{v^{2}_{i,j}\}$, there is no $\poly(n)$-time oracle algorithm computing $(i',j')$ because there might be exponentially many nodes. To resolve this issue, we will encode $\{v^{1}_{i,j}\}$ and $\{v^{2}_{i,j}\}$ via Reed--Muller codes.

Formally, recall that $\HisKor(C,f)$ is the concatenation of $(\istar,\jstar)$ and the string $S$, where $S$ is the concatenation of all the labels on the binary tree. We now define the encoded history, denoted as $\EncHisKor(C,f)$, as the concatenation of $(\istar,\jstar)$ and \emph{a Reed--Muller encoding} of $S$. The new selector is given oracle access to two candidate encoded histories together with $f$. By applying low-degree tests and self-correction of polynomials, we can assume that the Reed--Muller parts of the two candidates are indeed low-degree polynomials. Then we can use a reduction to polynomial identity testing to compute the first differing point between $\{v^{1}_{i,j}\}$ and $\{v^{2}_{i,j}\}$ in randomized polynomial time. See the proof of~\autoref{lm:comp-RM} for the details. This part is similar to the selector construction from~\cite{Hirahara15}.

\subsection{Discussions}

We conclude the introduction by discussing some related works.

\subsubsection{Previous approach: Karp--Lipton collapses and the half-exponential barrier}\label{sec:rel-half-exp}
In the following, we elaborate on the half-exponential barrier mentioned earlier in the introduction.\footnote{A function $f\colon\N\to\N$ is \emph{sub-half-exponential} if $f(f(n)^c) = 2^{o(n)}$ for every constant $c\ge 1$, i.e., composing $f$ twice yields a sub-exponential function. For example, for constants $c\ge 1$ and $\varepsilon > 0$, the functions $f(n) = n^c$ and $f(n) = 2^{\log^c n}$ are sub-half-exponential, but the functions $f(n) = 2^{n^\varepsilon}$ and $f(n) = 2^{\varepsilon n}$ are not.} Let $\calC$ be a ``typical'' uniform complexity class containing $\P$, a \emph{Karp--Lipton collapse} to $\calC$ states that if a large class (say $\EXP$) has polynomial-size circuits, then this class collapses to $\calC$. For example, there is a Karp--Lipton collapse to $\calC = \Sigma_2\P$:
\begin{quote}
	Suppose $\EXP \subseteq \P/_\poly$, then $\EXP = \Sigma_2\P$. (\cite{KarpL80}, attributed to Albert Meyer)
\end{quote}

Now, assuming that $\EXP \subseteq \P/_\poly \implies \EXP = \calC$, the following win-win analysis implies that $\calC\text{-}\EXP$, the exponential-time version of $\calC$, is not in $\P/_\poly$: (1) if $\EXP\not\subset \P/_\poly$, then of course $\calC\text{-}\EXP\supseteq \EXP$ does not have polynomial-size circuits; (2) otherwise $\EXP\subseteq \P/_\poly$. We have $\EXP = \calC$ and by padding $\EEXP = \calC\text{-}\EXP$. Since $\EEXP$ contains a function of maximum circuit complexity by direct diagonalization, it follows that $\calC\text{-}\EXP$ does not have polynomial-size circuits.

Karp--Lipton collapses are known for the classes $\Sigma_2\P$ \cite{KarpL80}, $\ZPP^\NP$ \cite{BshoutyCGKT96}, $\S_2\P$ \cite{Cai01a} (attributed to Samik Sengupta), $\PP$, $\MA$ \cite{LundFKN92, BabaiFNW93}, and $\ZPP^\MCSP$ \cite{ImpagliazzoKV18}. All the aforementioned super-polynomial circuit lower bounds for $\Sigma_2\EXP$, $\ZPEXP^\NP$, $\S_2\EXP$, $\PEXP$, $\MAEXP$, and $\ZPEXP^\MCSP$ are proven in this way.\footnote{There are some evidences that Karp--Lipton collapses are essential for proving circuit lower bounds \cite{ChenMMW19}.}

\paragraph*{The half-exponential barrier.} The above argument is very successful at proving various super-polynomial lower bounds. However, a closer look shows that it is only capable of proving \emph{sub-half-exponential} circuit lower bounds. Indeed, suppose we want to show that $\calC\text{-}\EXP$ does not have circuits of size $f(n)$. We will have to perform the following win-win analysis:
\begin{itemize}
	\item if $\EXP\not\subset \SIZE[f(n)]$, then of course $\calC\text{-}\EXP\supseteq \EXP$ does not have circuits of size $f(n)$;
	\item if $\EXP\subseteq \SIZE[f(n)]$, then (a scaled-up version of) the Karp--Lipton collapse implies that $\EXP$ can be computed by a $\calC$ machine of $\poly(f(n))$ time. Note that $\TIME[2^{\poly(f(n))}]$ does not have circuits of size $f(n)$ by direct diagonalization. By padding, $\TIME[2^{\poly(f(n))}]$ can be computed by a $\calC$ machine of $\poly(f(\poly(f(n))))$ time. Therefore, if $f$ is sub-half-exponential (meaning $f(\poly(f(n))) = 2^{o(n)}$), then $\calC\text{-}\EXP$ does not have circuits of size $f(n)$.
\end{itemize}

Intuitively speaking, the two cases above are \emph{competing with each other}: we cannot get exponential lower bounds in both cases.

\subsubsection{Implications for the Missing-String problem?}
\def\MissingString{\textsc{Missing-String}}
In the $\MissingString$ problem, we are given a list of $m$ strings $x_1, x_2, \dots, x_m \in \{0, 1\}^n$ where $m < 2^n$, and the goal is to output any length-$n$ string $y$ that does not appear in $\{x_1, x_2, \dots, x_m\}$. Vyas and Williams \cite{VyasW23} connected the circuit complexity of $\MissingString$ with the (relativized) circuit complexity of $\Sigma_2\E$:
\begin{theorem}[{\cite[Theorem 32]{VyasW23}, Informal}]\label{thm: lower bounds vs missing strings}
	The following are equivalent:\begin{itemize}
		\item $\Sigma_2\E^A\not\subset \io\SIZE^A[2^{\Omega(n)}]$ for every oracle $A$;
		\item for $m = 2^{\Omega(n)}$, the $\MissingString$ problem can be solved by a uniform family of size-$2^{O(n)}$ depth-$3$ $\AC^0$ circuits. 
	\end{itemize}
\end{theorem}

The intuition behind \autoref{thm: lower bounds vs missing strings} is roughly as follows. For every oracle $A$, the set of truth tables with low $A$-oracle circuit complexity induces an instance for $\MissingString$, and solving this instance gives us a hard truth table relative to $A$. If the algorithm for $\MissingString$ is a uniform $\AC^0$ circuit of depth $3$, then the hard function is inside $\Sigma_2\E^A$.

However, despite our \autoref{thm: main Sigma2} being completely relativizing, it does not seem to imply any non-trivial depth-$3$ $\AC^0$ circuit for $\MissingString$. The reason is the heavy win-win analysis \emph{across multiple input lengths}: for each $0\le i < t$, we have a single-valued $\mathsf{F}\Sigma_2\P$ construction algorithm for hard truth tables relative to oracle $A$ on input length $n_i$, but this algorithm needs access to $A_{n_{i+1}}$, \emph{a higher input length of $A$}. Translating this into the language of $\MissingString$, we obtain a weird-looking depth-$3$ $\AC^0$ circuit that takes as input a \emph{sequence} of $\MissingString$ instances $\calI_{n_0}, \calI_{n_1}, \dots, \calI_{n_t}$ (where \emph{each} $\calI_{n_i} \subseteq \{0, 1\}^{n_i}$ is a set of strings), looks at all of the instances (or, at least $\calI_{n_i}$ and $\calI_{n_{i+1}}$), and outputs a  purportedly missing string of $\calI_{n_i}$. It is guaranteed that for at least one input length $i$, the output string is indeed a missing string of $\calI_{n_i}$. However, if our algorithm is only given one instance $\calI \subseteq \{0, 1\}^n$, without assistance from a larger input length, it does not know how to find any missing string of $\calI$.

It remains an intriguing open problem whether the bullets in \autoref{thm: lower bounds vs missing strings} are true or not. In other words, is there an oracle $A$ relative to which $\Sigma_2\E$ has small circuits \emph{on infinitely many input lengths}?

\section*{Organization}

In~\autoref{sec:prelim}, we introduce the necessary technical preliminaries for this paper. In~\autoref{sec:korten}, we review Korten's reduction from solving range avoidance to generating hard truth tables~\cite{Korten21}, together with some new properties required by our new results. In~\autoref{sec: lower bounds for sigma2}, we prove the near-maximum circuit lower bound for $\Sigma_2\E$; although this lower bound is superseded by the later $\S_2\E/_{1}$ lower bound, we nonetheless include it in the paper since its proof is much more elementary. In~\autoref{sec:S2-lowb}, we extend the near-maximum circuit lower bound to $\S_2\E/_{1}$, and also present our new algorithms for solving the range avoidance problem.

%% file: prelim.tex
\section{Preliminaries}\label{sec:prelim}

\paragraph*{Notation.} We use $[n]$ to denote $\{1, 2, \dots, n\}$. %
A search problem $\Pi$ maps every input $x \in \bs{*}$ into a solution set $\Pi_x \subseteq \bs{*}$. We say an algorithm $A$ solves the search problem $\Pi$ on input $x$ if $A(x) \in \Pi_x$. 

\subsection{Complexity Classes}
We assume basic familiarity with computation complexity theory (see, e.g.,~\cite{AB09-book,Goldreich-book} for references). Below we recall the definition of $\S_2\TIME[T(n)]$~\cite{RussellS98,Canetti96}.

\begin{definition}\label{def: symmetric time}
	Let $T\colon \N \to \N$. We say a language $L \in \S_2\TIME[T(n)]$, if there exists an $O(T(n))$-time verifier $V(x,\pi_1,\pi_2)$ that takes $x \in \bs{n}$ and $\pi_1,\pi_2 \in \bs{T(n)}$ as input, satisfying that
	\begin{itemize}
		\item if $x \in L$, then there exists $\pi_1$ such that for every $\pi_2$, $V(x,\pi_1,\pi_2) = 1$, and
		\item if $x \not\in L$, then there exists $\pi_2$ such that for every $\pi_1$, $V(x,\pi_1,\pi_2) = 0$.
	\end{itemize}

	Moreover, we say $L \in \S_2\E$ if $L \in \S_2\TIME[T(n)]$ for some $T(n) \le 2^{O(n)}$, and $L \in \S_2\P$ if $L \in \S_2\TIME[p(n)]$ for some polynomial $p$.
\end{definition}

It is known that $\S_2\P$ contains $\MA$ and $\P^\NP$~\cite{RussellS98}, and $\S_2\P$ is contained in $\ZPP^\NP$~\cite{Cai01a}. From its definition, it is also clear that $\S_2\P \subseteq \Sigma_2 \P \cap \Pi_2 \P$.

\subsection{Single-valued \texorpdfstring{$\mathsf{F}\Sigma_2\P$}{FSigma2P} and \texorpdfstring{$\mathsf{F}\S_2\P$}{FS2P} Algorithms}

We consider the following definitions of single-valued algorithms which correspond to circuit lower bounds for $\Sigma_2\E$ and $\S_2\E$.

\newcommand{\algo}{\mathbb{A}}

\begin{definition}[{Single-valued $\mathsf{F}\Sigma_2\P$ and $\mathsf{F}\S_2\P$ algorithms}]~\label{def: single-valued algorithms}
	A single-valued $\mathsf{F}\Sigma_2\P$ algorithm $A$ is specified by a polynomial $\ell(\cdot)$ together with a polynomial-time algorithm $V_A(x, \pi_1, \pi_2)$. On an input $x \in \bs{*}$, we say that $A$ outputs $y_x \in \bs{*}$, if the following hold:
	\begin{enumerate}
		\item[(a)] There is a $\pi_1\in\{0, 1\}^{\ell(|x|)}$ such that for every $\pi_2\in\{0, 1\}^{\ell(|x|)}$, $V_A(x, \pi_1, \pi_2)$ outputs $y_x$.
		\item[(b)] For every $\pi_1\in\{0, 1\}^{\ell(|x|)}$, there is a $\pi_2\in\{0, 1\}^{\ell(|x|)}$ such that the output of $V_A(x, \pi_1, \pi_2)$ is either $y_x$ or $\bot$ (where $\bot$ indicates ``I don't know'').
	\end{enumerate}

	A single-valued $\mathsf{F}\S_2\P$ algorithm $A$ is specified similarly, except that %
 we replace the second condition above with the following:
	\begin{enumerate}
			\item[(b')] There is a $\pi_2\in\{0, 1\}^{\ell(|x|)}$ such that for every $\pi_1\in\{0, 1\}^{\ell(|x|)}$, $V_A(x, \pi_1, \pi_2)$ outputs $y_x$.
	\end{enumerate}
	
	\label{defi:sv-algos}
\end{definition}

Now, we say that a single-valued $\mathsf{F}\Sigma_2\P$ ($\mathsf{F}\S_2\P$) algorithm $A$ solves a search problem $\Pi$ on input $x$ %
if it outputs a string $y_x$ and $y_x \in \Pi_x$. Note that from~\autoref{defi:sv-algos}, if $A$ outputs a string $y_x$, then $y_x$ is unique.

For convenience, we mostly only consider single-valued algorithms $A(x)$ with fixed output lengths, meaning that the output length $|A(x)|$ only depends on $|x|$ and can be computed in polynomial time given $1^{|x|}$.\footnote{If $A$ takes multiple inputs like $x,y,z$, then the output length $A(x,y,z)$ only depends on $|x|,|y|,|z|$ and can be computed in polynomial time given $1^{|x|}$, $1^{|y|}$, and $1^{|z|}$.}

\subsubsection{Single-Valued \texorpdfstring{$\mathsf{F}\S_2\P$}{FS2P} and \texorpdfstring{$\mathsf{F}\Sigma_2\P$}{FSigma2P} algorithms with \texorpdfstring{$\FP^\NP$}{FP NP} post-processing}
We also need the fact that single-valued $\mathsf{F}\S_2\P$ or $\mathsf{F}\Sigma_2\P$ algorithms with $\FP^\NP$ post-processing can still be implemented by single-valued $\mathsf{F}\S_2\P$ or $\mathsf{F}\Sigma_2\P$ algorithms, respectively. More specifically, we have:

\begin{theorem}\label{theo:S-2-and-PNP}
	Let $A(x)$ be a single-valued $\mathsf{F}\S_2\P$ \textup{(}resp.\ $\mathsf{F}\Sigma_2\P$\textup{)} algorithm and $B(x,y)$ be an $\FP^\NP$ algorithm, both with fixed output length. The function $f(x) \coloneqq B(x,A(x))$ also admits an $\mathsf{F}\S_2\P$ \textup{(}resp.\ $\mathsf{F}\Sigma_2\P$\textup{)} algorithm.
\end{theorem}
\begin{proof}
	We only provide a proof for the case of single-valued $\mathsf{F}\S_2\P$ algorithms. Recall that the Lexicographically Maximum Satisfying Assignment problem ($\mathsf{LMSAP}$) is defined as follows: given an $n$-variable formula $\phi$ together with an integer $k \in [n]$, one needs to decide whether $a_k = 1$, where $a_1,\dotsc,a_n \in \bs{n}$ is the lexicographically largest assignment satisfies $\phi$. By~\cite{Krentel88}, $\mathsf{LMSAP}$ is $\P^\NP$-complete.
	
	Let $V_A(x,\pi_1,\pi_2)$ be the corresponding verifier for the single-valued $\mathsf{F}\S_2\P$ algorithm $A$. Let $L(x,y,i)$ be the $\P^\NP$ language such that $L(x, y, i) = 1$ if and only if $B(x,y)_i = 1$. Let $\ell = |B(x, y)|$ be the output length of $B$. We now define a single-valued $\mathsf{F}\S_2\P$ algorithm $\WT{A}$ by defining the following verifier $V_{\WT{A}}$, and argue that $\WT{A}$ computes $f$.

 	The verifier $V_{\WT{A}}$ takes an input $x$ and two proofs $\vec{\pi}_1$ and $\vec{\pi}_2$, where $\vec{\pi}_1$ consists of $\omega_1$, acting as the second argument to $V_A$, and $\ell$ assignments $z^{1}_1,z^{1}_2,\dotsc,z^{1}_\ell \in \bs{m}$. Similarly, $\vec{\pi}_2$ consists of $\omega_2$ and $z^{2}_1,z^{2}_2,\dotsc,z^{2}_\ell \in \bs{m}$.
	
	First, $V_{\WT{A}}$ runs $V_A(x,\omega_1,\omega_2)$ to get $y \in \bs{|A(x)|}$. Then it runs the reduction from $L(x,y,i)$ to $\mathsf{LMSAP}$ for every $i \in [\ell]$ to obtain $\ell$ instances $\{(\phi_i,k_i)\}_{i \in [\ell]}$, where $\phi_i$ is an $m$-variable formula and $k_i \in [m]$. (Without loss of generality by padding dummy variables, we may assume that the number of variables in $\phi_i$ is the same for each $i$, i.e., $m$; and that $m$ only depends on $|x|$ and $|y|$.) Now, for every $\mu \in [2]$, we can define an answer $w_\mu \in \bs{\ell}$ by $(w_\mu)_i = (z^{\mu}_i)_{k_i}$ (i.e., the value of $B(x,y)$, assuming that $\vec{\pi}_\mu$ consists of the lexicographically largest assignments for all the $\mathsf{LMSAP}$ instances).
	
	In what follows, when we say that $V_{\WT{A}}$ \emph{selects} the proof $\mu \in [2]$, we mean that $V_{\WT{A}}$ outputs $w_{\mu}$ and terminates. Then, $V_{\WT{A}}$ works as follows:
	\begin{enumerate}
		\item For each $\mu \in [2]$, it first checks whether for every $i \in [\ell]$, $z^{\mu}_i$ satisfies $\phi_i$. If only one of the $\mu$ passes all the checks, $V_{\WT{A}}$ selects that $\mu$. If none of them passes all the checks, $V_{\WT{A}}$ selects $1$. Otherwise, it continues to the next step.
		
		\item Now, letting $Z^{\mu}= z^{\mu}_1 \circ z^{\mu}_2 \circ \dotsc \circ z^\mu_{\ell}$ for each $\mu \in [2]$. $V_{\WT{A}}$ selects the $\mu$ with the lexicographically larger $Z^\mu$. If $Z^1 = Z^2$, then $V_{\WT{A}}$ selects $1$.
	\end{enumerate}

	Now we claim that $\WT{A}$ computes $f(x)$, which can be established by setting $\vec{\pi}_1$ or $\vec{\pi}_2$  be the corresponding proof for $V_A$ concatenated with all lexicographically largest assignments for the $\{\phi_i\}_{i \in [\ell]}$.
\end{proof}

\subsection{The Range Avoidance Problem}

The \emph{range avoidance} problem \cite{KKMP21, Korten21, RenSW22} is the following problem: Given as input a circuit $C \colon \{0, 1\}^n \to \{0, 1\}^\ell$ where $\ell > n$, find any string $y\in \bs{\ell} \setminus \Range(C)$. Proving circuit lower bounds (for exponential-time classes) is equivalent to solving the range avoidance problem on the \emph{truth table generator} $\TT_{n, s}$, defined as follows. It was shown in~\cite{FrandsenM05} that for $n,s\in \N$, any $s$-size $n$-input circuit $C$ can be encoded as a \emph{stack program} with description size $L_{n,s} := (s+1) (7 + \log(n+s))$. The precise definition of stack programs does not matter (see~\cite{FrandsenM05} for a formal definition); the only property we need is that given $s$ and $n$ such that $n \le s \le 2^n$, in $\poly(2^n)$ time one can construct a circuit $\TT_{n,s} \colon \bs{L_{n,s}} \to \bs{2^n}$ mapping the description of a stack program into its truth table. By the equivalence between stack programs and circuits, it follows that any $f \in \bs{2^n} \setminus \Range(\TT_{n,s})$ satisfies $\SIZE(f) > s$.  Also, we note that for large enough $n \in \N$ and $s = 2^n/n$, we have $L_{n,s} < 2^n$.

\begin{fact}\label{fact: range avoidance for TT implies circuit lower bounds}
    Let $s(n) \colon \N \to \N$. Suppose that there is a single-valued $\mathsf{F}\S_2\P$ algorithm  $A$ such that for infinitely many $n \in \N$, $A(1^{2^n})$ takes $\alpha(n)$ bits of advice and outputs a string $f_n \in \bs{2^n} \setminus \Range(\TT_{n,s(n)})$. Then $\S_2\E/_{\alpha(n)}\not\subset \SIZE[s(n)]$.
\end{fact}
\begin{proof}[Proof sketch]
	We define a language $L$ such that the truth table of the characteristic function of $L \cap \{0,1\}^n$ is $A(1^{2^n})$. It is easy to see that $L \notin \SIZE[s(n)]$ and $L \in \S_2\E/_{\alpha(n)}$.
\end{proof}

%% file: korten.tex
\section{Korten's Reduction}\label{sec:korten}

Our results crucially rely on a reduction in \cite{Korten21} showing that proving circuit lower bounds is ``the hardest explicit construction'' under $\P^\NP$ reductions.

\paragraph{Notation.} Let $s$ be a string of length $n$. We will always use $0$-index (i.e., the first bit of $s$ is $s_0$ and the last bit of $s$ is $s_{n-1}$). Let $i<j$, we use $s_{[i, j]}$ to denote the substring of $s$ from the $i$-th bit to the $j$-th bit, and $s_{[i, j)}$ to denote the substring of $s$ from the $i$-th bit to the $(j-1)$-th bit. (Actually, we will use the notation $s_{[i, j)}$ more often than $s_{[i, j]}$ as it is convenient when we describe the GGM tree.) We also use $s_1\circ s_2 \circ \dots \circ s_k$ to denote the concatenation of $k$ strings.

\subsection{GGM Tree and the Reduction}

We first recall the GGM tree construction from~\cite{GoldreichGM86}, which is used in a crucial way by~\cite{Korten21}.

\begin{definition}[The GGM tree construction {\cite{GoldreichGM86}}]
	Let $C\colon\{0, 1\}^n \to \{0, 1\}^{2n}$ be a circuit. Let $n,T \in \N$ be such that $T \ge 4n$ and let $k$ be the smallest integer such that $2^k n \ge T$. The function $\GGM_T[C]\colon\{0, 1\}^n \to \{0, 1\}^T$ is defined as follows.
	
	Consider a perfect binary tree with $2^k$ leaves, where the root is on level $0$ and the leaves are on level $k$. Each node is assigned a binary string of length $n$, and for $0\le j < 2^i$, denote $v_{i, j} \in \{0, 1\}^n$ the value assigned to the $j$-th node on level $i$. Let $x\in\{0, 1\}^n$. We perform the following computation to obtain $\GGM_T[C](x)$: we set $v_{0, 0} := x$, and for each $0\le i < k$, $0 \le j < 2^i$, we set $v_{i+1, 2j} := C(v_{i, j})_{[0, n)}$ (i.e., the first half of $C(v_{i, j})$) and $v_{i+1, 2j+1} := C(v_{i, j})_{[n, 2n)}$ (i.e., the second half of $C(v_{i, j})$). (We say the nodes $(i+1, 2j)$ and $(i+1, 2j+1)$ are ``children'' of $(i, j)$.)
	
	Finally, we concatenate all values of the leaves and take the first $T$ bits as the output:
	\[\GGM_T[C](x) := (v_{k, 0} \circ v_{k, 1} \circ \dots \circ v_{k, 2^k-1})_{[0, T)}.\]
\end{definition}

\begin{lemma}[The output of GGM tree has a small circuit]\label{lemma: GGMEval}
	Let $\GGMEval(C, T, x, i)$ denote the $i$-th bit of $\GGM_T[C](x)$. There is an algorithm running in $\WT{O}\mleft(|C| \cdot \log T\mright)$ time that, given $C, T, x, i$, outputs $\GGMEval(C, T, x, i)$.
\end{lemma}
\begin{proof}[Proof Sketch]
	We first note that to compute the $i$-th bit of $\GGM_T[C](x) := (v_{k, 0} \circ v_{k, 1} \circ \dots \circ v_{k, 2^k-1})_{[0, T)}$, it suffices to compute $v_{k,\lfloor i/n\rfloor}$. Computing $v_{k,\lfloor i/n\rfloor}$ can be done by descending from the root of the GGM tree to the leave $(k,\lfloor i/n\rfloor)$, which takes $\WT{O}(|C| \cdot \log T)$ time.
\end{proof}

It is shown in \cite{Korten21} that the range avoidance problem for $C$ reduces to the range avoidance problem for $\GGM_T[C]$. In what follows, we review this proof, during which we also define the \emph{computational history} of ``solving range avoidance of $C$ from $\GGM_T[C]$'', which will be crucial in our main proof.

\begin{algorithm2e}
	\caption{$\Kor(C, f)$: Korten's reduction}\label{algo: Korten-algo}
	\SetKwInput{KwSetting}{Setting}
	\SetKwInput{KwPara}{Parameters}
	\SetKwInput{KwAssumption}{Assumption}
	\SetKwInput{KwAdvice}{Advice}
	\SetKw{KwAccept}{accept}
	\SetKw{KwReject}{reject}
	
	\KwIn{$C\colon \{0, 1\}^n \to \{0, 1\}^{2n}$ denotes the input circuit, and $f \in \bs{T}\setminus \Range(\GGM_T[C])$ denotes the input ``hard'' truth table}
	\KwOut{A non-output of $C$}
	\KwData{The computational history of $\Kor(C, f)$: a pair $(\istar,\jstar)$ and an array $\{v_{i,j}\}_{i,j}$ where $i \in \{0,1\dotsc,k\}$ and $j \in \{0,1,\dotsc,2^i\}$.} %
	Let $k \gets \left\lceil \log_2(T/n) \right\rceil$\;
	Append $f$ with $2^k n-|f|$ zeros at the end\;
	\For{$j\gets 0$ \emph{to} $2^k-1$}{
		$v_{k, j}\gets f_{[jn, (j+1)n)}$\;\tcc{the $j$-th ``block'' of $f$}
	}
	\For {$i\gets k-1$ \emph{downto} $0$}{
		\For {$j\gets 2^i-1$ \emph{downto} $0$}{
			Let $v_{i, j}$ be the lexicographically smallest string in $C^{-1}(v_{i+1, 2j}\circ v_{i+1, 2j+1})$\;
			\tcc{Note that this step needs to invoke the $\NP$ oracle}
			\If {$v_{i, j}$ \emph{does not exist}}{
				For every $(i', j')$ such that $v_{i', j'}$ is not set yet, set $v_{i', j'} \gets \bot$\;
				Set $\istar := i$, and $\jstar := j$\;
				\Return $v_{i+1, 2j}\circ v_{i+1, 2j+1}$\;
			}
		}
	}
	\Return{$\bot$}
\end{algorithm2e}

\begin{lemma}[{Reduction from solving range avoidance of $C$ to solving range avoidance of $\GGM_T[C]$}]\label{lemma: Korten's reduction}
	Let $C\colon\{0, 1\}^n \to \{0, 1\}^{2n}$ be a circuit. Let $f$ be a non-output of $\GGM_T[C]$, i.e., $f\in\{0, 1\}^T\setminus\Range(\GGM_T[C])$.  Then, $\Kor(C, f)$ (as defined in \autoref{algo: Korten-algo}) outputs a non-output of $C$ in deterministic $\poly(T, n)$ time with an $\NP$ oracle.
\end{lemma}
\begin{proof}[Proof Sketch]
	The running time of $\Kor(C,f)$ follows directly from its description. Also, note that whenever $\Kor(C,f)$ outputs a string $v_{i+1, 2j}\circ v_{i+1, 2j+1} \in \bs{2n}$, it holds that this string is not in the range of $C$. Therefore, it suffices to show that when $f \in \bs{T}\setminus \Range(\GGM_T[C])$, $\Kor(C,f)$ does not return $\bot$.

    Assume, towards a contradiction, that $\Kor(C,f)$ returns $\bot$.  This means that all the $\{v_{i,j}\}_{i,j}$ values are set.   It follows from the algorithm description that $f = \GGM_T[C](v_{0,0})$, which contradicts the assumption that $f \in \bs{T}\setminus \Range(\GGM_T[C])$.
\end{proof}

In addition, we observe the following trivial fact:
\begin{fact}\label{fact: brute force for sigma2}
	Let $C:\{0, 1\}^n \to \{0, 1\}^{2n}$ be a circuit, $T := 2^{2n}\cdot 2n$, and $f$ be the concatenation of all length-$2n$ strings (which has length $T$). Then $f\not\in\Range(\GGM_T[C])$.
\end{fact}

One can combine \autoref{fact: brute force for sigma2} with \autoref{lemma: Korten's reduction} to obtain a brute force algorithm that solves the range avoidance problem in $2^{O(n)}$ time with an $\NP$ oracle. Essentially, this brute force algorithm tests every possible length-$2n$ string against the range of the circuit. It will be the basis of our win-win analysis in \autoref{sec: lower bounds for sigma2}.

Finally, we give the following remark, showing that Korten's reduction relativizes.

\begin{remark}\label{rem:Korten-algo-rel}
	\autoref{algo: Korten-algo} and~\autoref{lemma: Korten's reduction} \emph{relativizes}, in the sense that if the input is actually an oracle circuit $C^{O}$ for some arbitrary oracle, the algorithm still works except now it needs to call an $\NP^{O}$ oracle to find the lexicographically smallest string in $C^{-1}(v_{i+1, 2j}\circ v_{i+1, 2j+1})$.
\end{remark}

\subsection{\texorpdfstring{$\Pi_1$}{Pi1} Verification of the History of \texorpdfstring{$\Kor(C,f)$}{Kor(C,f)}}

In what follows, we say that $(i, j) < (i', j')$ if either $i < i'$ or ($i = i'$ and $j < j'$) (that is, we consider the lexicographical order of pairs). Observe that \autoref{algo: Korten-algo} processes all the pairs $(i, j)$ in the reverse lexicographic order.

\begin{definition}[The computational history of $\Kor(C,f)$]\label{def: computation history of Korten}
	Let $n,T \in \N$ be such that $\log T \le n \le T$.	Let $C\colon\{0, 1\}^n \to \{0, 1\}^{2n}$ be a circuit, and $f\in\{0, 1\}^T$ be a ``hard truth table'' in the sense that $f\not\in\Range(\GGM_T[C])$. The \emph{computational history} of $\Kor(C, f)$, denoted as 
	\[
	\HisKor(C, f),
	\]
	consists of $(\istar, \jstar)$, as well as the concatenation of $v_{i, j}$ for every $0\le i < k$ and $0\le j < 2^i$, in the lexicographical order of $(i, j)$ ($(\istar, \jstar)$ and the $v_{i, j}$ are defined in~\autoref{algo: Korten-algo}). Each $v_{i, j}$ is encoded by $n+1$ bits ${\sf enc}(v_{i, j})$, where if $v_{i, j} \in \{0, 1\}^n$ then ${\sf enc}(v_{i, j}) = 0\circ v_{i, j}$, and if $v_{i, j} = \bot$ then ${\sf enc}(v_{i, j}) = 1^{n+1}$. The length of this history is at most $(2^{k+1}-1)(n+1) + 2\log T \le 5T$, and for convenience we always pad zeros at the end so that its length becomes exactly $5T$.
\end{definition}

The following lemma summarizes the properties of the computational history construction above required for the $\Sigma_2 \E$ lower bound in the next section.

\begin{lemma}\label{lemma:Kor-prop-sigma2}
	Let $n,T \in \N$ be such that $\log T \le n \le T$. Let $C\colon\{0, 1\}^n \to \{0, 1\}^{2n}$ be a circuit and $f\in\bs{T} \setminus \Range(\GGM_T[C])$. Let $h \coloneqq \HisKor(C, f)$ and $z \coloneqq \Kor(C, f)$.
	
	\begin{enumerate}
		\item {\bf (history contains input/output)} There is a $\poly(\log T)$-time one-query oracle algorithm $\In$ and an $O(n)$-time oracle algorithm $\Out$, both having input parameters $T, n$ and taking a string $\tilde{h} \in \bs{5T}$ as oracle, such that the following hold:\label{item: Input-Output}
		
		\begin{enumerate}
			\item When given $h$ as the oracle, $\In_{T,n}$ takes an additional input $i \in \zeroUpto{5T-1}$ and outputs $f_i$. \label{item: history-contain-input}
			
			\item When given $h$ as the oracle, $\Out_{T,n}$ outputs $z = \Kor(C, f)$. \label{item: history-contain-output}
		\end{enumerate}
		
		\item {\bf ($\Pi_1$ verification of the history)} There is an oracle algorithm $V$ with input parameters $T,n$ such that the following holds:\label{item: Pi1 verification of history}
		
		\begin{enumerate}
			\item $V$ takes $\tilde{f} \in \bs{T}, \tilde{h} \in \bs{5 T}$ as oracles and $C$ and $w \in \bs{5 \cdot (\log T + n)}$ as inputs. It runs in $\poly(n)$ time.
			\item $h = \HisKor(C,f)$ is the unique string from $\bs{5 T}$ satisfying the following: \label{item: uniquely verifiable history}
			\[
			V^{f,h}(C, w) = 1 \qquad\text{for every $w \in \bs{5 \cdot (\log T + n)}$.}
			\]
		\end{enumerate}
	\end{enumerate}
\end{lemma}
\begin{proof}
	\newcommand{\tdh}{\tilde{h}}
	From the definition of $\HisKor(C,f)$, the construction of $\In_{T,n}$ and $\Out_{T,n}$ are straightforward. Now we describe the verifier $V^{f,\tdh}$, where $f \in \bs{T}$ and $\tdh \in \bs{5T}$. Note that here we fix the first oracle of $V$ to be the input truth table $f$, while the second oracle $\tdh$ can be any string from $\bs{5T}$.
	
	First, $V$ reads $(\istar,\jstar)$ from $\tdh$. Note that the rest of $\tdh$ can be parsed as an array $\{v_{i,j}\}_{i,j}$ where $i \in \{0,1\dotsc,k\}$ and $j \in \{0,1,\dotsc,2^i\}$. We will think of $V$ as performing at most $2^{|w|}$ checks, each of which \emph{passes} or \emph{fails}. To show the second item of the lemma, we need to show that (1) if a string $\tdh$ passes all the checks, then it must be the case that $\tdh = h$; and (2) $h$ passes all the checks.
	
	Specifically, $V$ checks $\tdh$ as follows:
	
	\begin{itemize}
		\item The values written on the leaves of $\{v_{i,j}\}$ are indeed $f$. That is, for every $j \in \zeroUpto{2^k-1}$, check that $v_{k,j}$ is consistent with the corresponding block in $f$.
		
		\item For every $(i, j) > (\istar, \jstar)$ such that $i < k$, $C(v_{i, j}) = v_{i+1, 2j} \circ v_{i+1, 2j+1}$. (That is, the value $v_{i, j}$ is consistent with its two children.)
		\item For every $(i, j) > (\istar, \jstar)$ such that $i < k$, for every $x\in\{0, 1\}^n$ that is lexicographically smaller than $v_{i, j}$, $C(x) \ne v_{i+1, 2j} \circ v_{i+1, 2j+1}$. (That is, the value $v_{i, j}$ is the lexicographically first preimage of its two children.)
		\item For every $x\in\{0, 1\}^n$, $C(x) \ne v_{\istar + 1, 2\jstar} \circ v_{\istar + 1, 2\jstar + 1}$. (That is, the two children of $(\istar, \jstar)$ form a non-output of $C$; by the previous checks, $(\istar, \jstar)$ is the lexicographically largest such pair.)
		\item For every $(i, j) \le (\istar, \jstar)$, $v_{i, j} = \bot$.
	\end{itemize}
	
	Note that the above can be implemented with a universal ($\forall$) quantification over at most $5 \cdot (\log T + n)$ bits. First, one can see that by the definition of the correct history $h$ (\autoref{def: computation history of Korten}), $h$ passes all the checks above. Second, one can indeed see that all the conditions above \emph{uniquely determine} $h$, and therefore any $\tdh$ passing all the checks must equal $h$.%
\end{proof}

Again, it is easy to observe that~\autoref{def: computation history of Korten} and~\autoref{lemma:Kor-prop-sigma2} relativize.
\begin{remark}\label{rem:Kor-prop-sigma2-rel}
	\autoref{def: computation history of Korten} and~\autoref{lemma:Kor-prop-sigma2} \emph{relativize}, in the sense that if $C$ is an oracle circuit $C^{O}$ for some arbitrary oracle,~\autoref{def: computation history of Korten} needs no modification since~\autoref{algo: Korten-algo} relativizes, and~\autoref{lemma:Kor-prop-sigma2} holds with the only modification that $V$ now also need to take $O$ as an oracle (since it needs to evaluate $C$).
\end{remark}

%% file: sigma_2.tex
\section{Circuit Lower Bounds for \texorpdfstring{$\Sigma_2\E$}{Sigma2E}}\label{sec: lower bounds for sigma2}

In this section, we prove our near-maximum circuit lower bounds for $\Sigma_2 \E$ by providing a new single-valued $\mathsf{F}\Sigma_2\P$ algorithm for $\Avoid$.

Let $\{C_n \colon \{0, 1\}^n \to \{0, 1\}^{2n}\}_{n \in \N}$ be a $\P$-uniform family of circuits. We show that there is a single-valued $\mathsf{F}\Sigma_2\P$ algorithm $A$ that, on input $1^n$, outputs a canonical string that is outside the range of $C_n$ for infinitely many $n\in\N$.

\begin{theorem}\label{theo:Sigma-2-algo-for-range-avoidance}
	Let $\{C_n \colon \{0, 1\}^n \to \{0, 1\}^{2n}\}_{n \in \N}$ be a $\P$-uniform family of circuits. There is a single-valued $\mathsf{F}\Sigma_2\P$ algorithm $A$ with one bit of advice such that for infinitely many $n \in \N$, $A(1^n)$ outputs $y_n \in \bs{2n} \setminus \Range(C_n)$. 
\end{theorem}
\begin{proof} We begin with some notation.
	
	\paragraph*{Notation.}
	Let $n^{(1)}$ be a large enough power of $2$, $n^{(\ell)} = 2^{2^{n^{(\ell-1)}}}$ for each integer $\ell > 1$. Let $n_0^{(\ell)} = n^{(\ell)}$ and $t^{(\ell)} = O\mleft(\log n_0^{(\ell)}\mright)$ be parameters that we set later. For each $1\le i\le \tl$, let $\nell_i := \mleft(\nell_{i-1}\mright)^{10}$. To show our algorithm $A$ works on infinitely many input lengths, we will show that for every $\ell \in \N$, there is an input length $\nell_{i}$ for some $i \in \zeroUpto{\tl}$ such that $A$ works.
	
	Fix $\ell \in \N$. From now on, for convenience, we will use $n_i$ and $t$ to denote $\nell_i$ and $\tl$, respectively.
		
	\paragraph*{Specifying $T_i$ and $f_i$.} For each input length $n_i$, we will specify a parameter $T_i \in \N$ and a string $f_i \in \bs{T_i}$. Our win-win analysis is based on whether $f_i \in \Range(\GGM_{T_i}[C_{n_i}])$ for each $i \in \zeroUpto{t}$.
	
	Let $T_0 := 2^{2n_0}\cdot 2n_0$ and $f_0$ be the concatenation of all length-$2n_0$ strings (which has length $T_0$). From \autoref{fact: brute force for sigma2}, we have that $f_0 \not \in \Range(\GGM_{T_0}[C_{n_0}])$. For every $i \in [t]$, we define
	\[f_i := \HisKor(C_{n_{i-1}}, f_{i-1}).\]
	From \autoref{def: computation history of Korten}, this also means that we have set $T_i = 5 \cdot T_{i-1}$ for every $i \in [t]$.
	
	Let $t$ be the first integer such that $T_{t+1} \le 4n_{t+1}$. Note that we have $T_i = 5^i \cdot T_0 \le 2^{3 n_0 + i \cdot \log 5}$ and $n_i = (n_0)^{10^i} = 2^{\log n_0 \cdot 10^i}$. Hence, we have that $t \le O(\log n_0)$. (Also note that $n_t^{(\ell)} < n_0^{(\ell + 1)}$.)
	
	\paragraph*{Description of our $\mathsf{F}\Sigma_2\P$ algorithm $A$.} Now, let $k \in \{0,1,\dotsc,t\}$ be the largest integer such that $f_k\not\in \Range(\GGM_{T_k}[C_{n_k}])$. Since $f_0\not\in \Range(\GGM_{T_0}[C_{n_0}])$, such a $k$ must exist. Let $z := \Kor(C_{n_k}, f_k)$.  It follows from \autoref{lemma: Korten's reduction} that $z$ is not in the range of $C_{n_k}$. Our single-valued $\mathsf{F}\Sigma_2\P$ algorithm $A$ computes $z$ on input $1^{n_k}$ (see~\autoref{defi:sv-algos}). That is, for some $\ell_1, \ell_2 \le \poly(n_k)$:
	
	\begin{itemize}
		\item There exists $\pi_1\in\{0, 1\}^{\ell_1}$ such that for every $\pi_2\in\{0, 1\}^{\ell_2}$, $V_A(1^{n_k}, \pi_1, \pi_2)$ prints $z$, and
		\item For every $\pi_1\in\{0, 1\}^{\ell_1}$, there exists some $\pi_2\in\{0, 1\}^{\ell_2}$ such that $V_A(1^{n_k}, \pi_1, \pi_2)$ prints either $z$ or $\bot$.
	\end{itemize}
	
	In more details, if $k < t$, then $V_A$ treats $\pi_1$ as an input to the circuit $\GGM_{T_{k+1}}[C_{n_{k+1}}]$, and let 
	\[\hat{f}_{k+1} := \GGM_{T_{k+1}}[C_{n_{k+1}}](\pi_1).\]
	Here, the length of $\pi_1$ is $\ell_1 := n_{k+1} \le \poly(n_k)$. If $k = t$, then $V_A$ defines $\hat{f}_{k+1} := \pi_1$ and $\ell_1 := T_{t+1} \le \poly(n_k)$. It is intended that $\hat{f}_{k+1} = f_{k+1} = \History(C_{n_k}, f_k)$ (which $V_A$ needs to verify). Note that in the case where $k < t$, since $f_{k+1} \in \Range(\GGM_{T_{k+1}}[C_{n_{k+1}}])$, there indeed exists some $\pi_1$ such that $\hat{f}_{k+1} = f_{k+1}$.
	
	We note that \autoref{lemma: GGMEval} provides us ``random access'' to the (potentially very long) string $\hat{f}_{k+1}$: given $\pi_1$ and $j\in[T_{k+1}]$, one can compute the $j$-th bit of $\hat{f}_{k+1}$ in $\poly(n_k)$ time. Also recall from~\autoref{lemma:Kor-prop-sigma2} that for each $i$, $f_{i+1} = \History(C_{n_i}, f_i)$ contains the string $f_i$, which can be retrieved by the oracle algorithm $\In$ described in \autoref{item: Input-Output} of \autoref{lemma:Kor-prop-sigma2}. Therefore, for each $i$ from $k$ downto $1$, we can recursively define $\hat{f}_i$ such that $(\hat{f}_i)_j = \In_{T_i, n_i}^{\hat{f}_{i+1}}(j)$. We define $\hat{f}_0$ to be the concatenation of all length-$(2n_0)$ strings in the lexicographical order, so $\hat{f}_0 = f_0$. Applying the algorithm $\In$ recursively, we obtain an algorithm that given $i \in \zeroUpto{k}$ and $j \in \zeroUpto{T_i-1}$, outputs the $j$-th bit of $\hat{f}_i$. Since $\In$ only makes one oracle query, this algorithm runs in $\poly(n_k)$ time.\footnote{Note that the definition of $f_0$ is so simple that one can directly compute the $j$-th bit of $f_0$ in $\poly(n_0)$ time.}
	
	Then, $V_A$ parses the second proof $\pi_2$ into $\pi_2 = (i, w)$ where $i \in \zeroUpto{k}$ and $w\in\{0, 1\}^{5(\log T_i + n_i)}$. Clearly, the length of $\pi_2$ is at most $\ell_2 := \log (k+1) + 5(\log T_k + n_k) \le \poly(n_k)$. Now, let $V_{\HisKor}$ be the oracle algorithm in \autoref{item: Pi1 verification of history} of \autoref{lemma:Kor-prop-sigma2}, we let $V_A(1^{n_k}, \pi_1, \pi_2)$ check whether the following holds:
	\begin{equation}\label{eq: check in sigma2}
		V_{\HisKor}^{\hat{f}_i, \hat{f}_{i+1}}(C_{n_i}, w) = 1.\footnote{Here $V_{\History}$ also takes input parameters $T_i$ and $n_i$. We omit them in the subscript for notational convenience.}
	\end{equation}
	If this is true, then $V_A$ outputs the string $z := \Out^{\hat{f}_{k+1}}_{T_k, n_k}$, where $\Out$ is the output oracle algorithm defined in \autoref{item: Input-Output} of~\autoref{lemma:Kor-prop-sigma2}. %
 Otherwise, $V_A$ outputs $\bot$.
	
	\paragraph*{The correctness of $A$.} Before establishing the correctness of $A$, we need the following claim:
	\begin{claim}\label{claim:condition}
		$f_{k+1} = \hat{f}_{k+1}$ if and only if the following holds:
		\begin{itemize}
			\item  $V_{\History}^{\hat{f}_i, \hat{f}_{i+1}}(C_{n_i}, w) = 1$ for every $i \in \zeroUpto{k}$ and for every $w \in \{0, 1\}^{5(\log T_i + n_i)}$.
		\end{itemize}
	\end{claim}
	\begin{claimproof}
		First, assume that $f_{k+1} = \hat{f}_{k+1}$. By \autoref{item: history-contain-input} of \autoref{lemma:Kor-prop-sigma2}, we have that $\hat{f}_i = f_i$ for every $i \in \zeroUpto{k+1}$. Recall that by definition, $f_{i+1} = \HisKor(C_{n_i}, f_i)$ for every $i \in \zeroUpto{k}$. Hence, by \autoref{item: uniquely verifiable history} of \autoref{lemma:Kor-prop-sigma2}, we have that for every $i \in \zeroUpto{k}$, and for every $w \in \{0, 1\}^{5(\log T_i + n_i)}$, $V_{\History}^{\hat{f}_i,\hat{f}_{i+1}}(C_{n_i}, w) = 1$ holds.
		
		For the other direction, suppose that for every $i \in \zeroUpto{k}$ and $w \in \{0, 1\}^{5(\log T_i + n_i)}$, we have that $V_{\History}^{\hat{f}_i,\hat{f}_{i+1}}(C_{n_i}, w) = 1$ holds. First recall that $f_0 = \hat{f}_0$ by definition. By an induction on $i \in [k+1]$ and (the uniqueness part of) \autoref{item: uniquely verifiable history} of \autoref{lemma:Kor-prop-sigma2}, it follows that $f_i = \hat{f}_i$ for every $i \in \zeroUpto{k+1}$. In particular, $f_{k+1} = \hat{f}_{k+1}$.
	\end{claimproof}
	
	Now we are ready to establish that $A$ is a single-valued $\mathsf{F}\Sigma_2\P$ algorithm computing $z$ on input $1^{n_k}$. We first prove the completeness of $A$; i.e., there is a proof $\pi_1$ such that for every $\pi_2$, $V_A(1^{n_k}, \pi_1, \pi_2)$ outputs $z = \Kor(C_{n_k}, f_k)$. We set $\pi_1$ to be the following proof: If $k < t$, then $f_{k+1} \in \Range(\GGM_{T_{k+1}}[C_{n_{k+1}}])$, and we can set $\pi_1 \in \{0, 1\}^{n_{k+1}}$ to be the input such that $f_{k+1} = \GGM_{T_{k+1}}[C_{n_{k+1}}](\pi_1)$; if $k = t$, then we simply set $\pi_1 = f_{k+1}$. Then, we have $f_{k+1} = \hat{f}_{k+1}$, and by \autoref{claim:condition}, we know that $V_A$ will output $z = \Kor(C_{n_k}, f_k)$ on every proof $\pi_2$.

	Next, we show that for every $\pi_1$, there is some $\pi_2$ that makes $V_A$ output either $z$ or $\bot$. It suffices to consider $\pi_1$ such that for every $\pi_2$, $V_A(1^{n_k}, \pi_1, \pi_2)\ne \bot$. In this case, every invocation of \autoref{eq: check in sigma2} holds, and thus by \autoref{claim:condition} we know that $f_{k+1} = \hat{f}_{k+1}$. It follows that $\Kor(C_{n_k}, f_k) = z$ and $V_A$ will output $z$ regardless of $\pi_2$.

 	Finally, we generalize $A$ and $V_A$ to work on all inputs $1^n$. On input $1^n$, $V_A$ calculates the largest $\ell$ such that $n^{(\ell)} \le n$, and also calculates the largest ${k'}$ such that $n_{k'}^{(\ell)} \le n$. If $n_{k'}^{(\ell)}\ne n$, then $V_A$ immediately outputs $\bot$ and halts. Otherwise, $V_A$ receives an advice bit indicating whether $k' = k^{(\ell)}$ where $k^{(\ell)}$ is the largest integer such that $f_{k^{(\ell)}}^{(\ell)} \not \in \Range(\GGM_{T_k^{(\ell)}}[C_{n_k^{(\ell)}}])$. If this is the case, then $V_A$ runs the verification procedure above; otherwise, it immediately outputs $\bot$ and halts. It is easy to see that $V_A$ runs in $\poly(n)$ time, and is an infinitely-often single-valued $\mathsf{F}\Sigma_2\P$ algorithm solving the range avoidance problem of $\{C_n\}_{n \in \N}$.
\end{proof}

From~\autoref{rem:Korten-algo-rel} and~\autoref{rem:Kor-prop-sigma2-rel}, one can obverse that the proof above also relativizes. Hence we have the following as well.

\newcommand{\oracle}{\mathcal{O}}

\begin{theorem}[Relativized version of \autoref{theo:Sigma-2-algo-for-range-avoidance}]
	Let $\oracle \colon \bs{*} \to \bs{}$ be any oracle. Let $\{C_n^\oracle \colon \{0, 1\}^n \to \{0, 1\}^{2n}\}_{n \in \N}$ be a $\P$-uniform family of $\oracle$-oracle circuits. There is a single-valued $\mathsf{F}\Sigma_2\P^\oracle$ algorithm $A^\oracle$ with one bit of advice such that for infinitely many $n \in \N$, $A^\oracle(1^n)$ outputs $y_n \in \bs{2n} \setminus \Range(C_n^\oracle)$.
\end{theorem}

We omit the proof of the following corollary since it is superseded by the results in the next section.

\begin{corollary}
	$\Sigma_2\E\not\subseteq \SIZE[2^n/n]$ and $(\Sigma_2\E \cap \Pi_2\E)/_1\not\subseteq \SIZE[2^n/n]$.
	Moreover, these results relativize: for every oracle $\oracle$, $\Sigma_2\E^\oracle\not\subseteq \SIZE^\oracle[2^n/n]$ and $(\Sigma_2\E^\oracle \cap \Pi_2\E^\oracle)/_1\not\subseteq \SIZE^\oracle[2^n/n]$.
\end{corollary}

%% file: S2.tex
\section{Circuit Lower Bounds for \texorpdfstring{$\S_2\E$}{S2E}}\label{sec:S2-lowb}

In this section, we prove our near-maximum circuit lower bounds for $\S_2\E/_1$ by giving a new single-valued $\mathsf{F}\S_2\P$ algorithm for $\Avoid$.

\subsection{Reed--Muller Codes}

\newcommand{\RM}{\mathsf{RM}}
\newcommand{\BRM}{\mathsf{BRM}}
\newcommand{\degmax}{\deg_{\rm max}}

To prove maximum circuit lower bounds for $\S_2\E/_1$, we will need several standard tools for manipulating Reed--Muller (RM) codes (i.e., low-degree multi-variate polynomials).

For a polynomial $P \colon \F_p^m \to \F_p$, where $\F_p$ is the finite field of $p$ elements, we use $\degmax(P)$ to denote the maximum individual degree of variables in $P$. Let $p$ be a prime, $\Delta, m \in \N$. For a string $S \in \bs{\Delta^m}$, we use $\RM_{\F_p,\Delta,m}(S)$ to denote its Reed--Muller encoding by extension: letting $H = \{0,1,\dotsc,\Delta-1\}$ and $w_1,\dotsc,w_{\Delta^m} \in H^m$ be the enumeration of all elements in $H^m$ in the lexicographical order, $\RM_{\F_p,\Delta,m}(S)$ is the unique polynomial $P\colon \F_p^m \to \F_p$ such that (1) $P(w_i) = S_i$ for every $i \in [\Delta^m]$ and (2) $\degmax(P) \le \Delta - 1$.\footnote{To see the uniqueness of $P$, note that for every $P \colon \F_p^m \to \F_p$ with $\degmax(P) \le \Delta-1$, the restriction of $P$ to $H^m$ uniquely determines the polynomial $P$. Also, such $P$ can be constructed by standard interpolation.}

We also fix a Boolean encoding of $\F_p$ denoted as $\Enc_{\F_p} \colon \F_p \to \bs{\lceil \log p \rceil}$. For simplicity, we can just map $z \in \zeroUpto{p-1}$ to its binary encoding. In particular, $\Enc_{\F_p}(0) = 0^{\lceil \log p \rceil}$ and $\Enc_{\F_p}(1) = 0^{\lceil \log p \rceil -1} \circ 1$.\footnote{This fact is useful because if we know a string $m \in \bs{\lceil \log p \rceil}$ encodes either $0$ or $1$, then we can decode it by only querying the last bit of $m$.} Now we further define $\BRM_{\F_p,\Delta,m}(S)$ by concatenating $\RM_{\F_p,\Delta,m}(S)$ with $\Enc_{\F_p}$, thus obtaining a Boolean encoding again. Formally, letting $P = \RM_{\F_p,\Delta,m}(S)$ and $w_1,\dotsc,w_{p^m} \in \F_p^m$ be the enumeration of all elements from $\F_p^m$ in the lexicographic order, we define $\BRM_{\F_p,\Delta,m}(S) = \Enc_{\F_p}(P(w_1)) \circ \Enc_{\F_p}(P(w_2)) \circ \dotsc \circ \Enc_{\F_p}(P(w_{p^m}))$.
We remark that for every $i \in \bkts{\Delta^m}$, in $\poly(m,\log p)$ time one can compute an index $i' \in [p^m \cdot \lceil \log p \rceil]$ such that $\BRM_{\F_p,\Delta,m}(S)_{i'} = S_i$.

We need three properties of Reed--Muller codes, which we explain below.

\paragraph*{Self-correction for polynomials.} We first need the following self-corrector for polynomials, which efficiently computes the value of $P$ on any input given an oracle that is close to a low-degree polynomial $P$.  %
(In other words, it is a \emph{local decoder} for the Reed--Muller code.)

\begin{lemma}[A self-corrector for polynomials, cf.~\cite{DBLP:journals/ipl/GemmellS92, Sud95}]\label{lm:poly-SC}
	There is a probabilistic oracle algorithm $\mathsf{PCorr}$ such that the following holds. Let $p$ be a prime, $m,\Delta\in\mathbb{N}$ such that $\Delta<p/3$. Let $g\colon \F_{p}^m\to \F_p$ be a function such that for some polynomial $P$ of total degree at most $\Delta$,
	\[
	\Pr_{\vec{x}\gets{\F_p^m}}[g(\vec{x}) \ne P(\vec{x})] \le 1/4.
	\]
	Then for all $\vec{x}\in\F_p^m$, $\mathsf{PCorr}^{g}(p,m,\Delta,\vec{x})$ runs in time $\poly(\Delta,\log p,m)$ and outputs $P(\vec{x})$ with probability at least $2/3$.
\end{lemma}

\paragraph*{Low-max-degree test.} We also need the following efficient tester, which checks whether a given polynomial has maximum individual degree at most $\Delta$ or is far from such polynomials.\footnote{To obtain the theorem below, we set the parameter $\delta$ and $\eps$ from~\cite[Remark~5.15]{BabaiFL91} to be $\min\mleft(\frac{1}{200n^2(\Delta+1)},1/2p\mright)$ and $\min\mleft( \frac{1}{400n^3(\Delta+1)},1/2p\mright)$, respectively.}

\begin{lemma}[{Low-max-degree tester~\cite[Remark~5.15]{BabaiFL91}}]\label{lm:poly-LDT}
	Let $n,\Delta,p \in \N$ be such that $p \ge 20 \cdot (\Delta+1)^2 \cdot n^2$ and $p$ is a prime. There is a probabilistic non-adaptive oracle machine $\mathsf{LDT}$ such that the following holds. Let $g\colon \F_p^n\to \F_p$. Then for $\delta = 3 n^2 \cdot (\Delta+1) / p$, it holds that
	\begin{enumerate}
		\item if $\degmax(g)\leq \Delta$, then $\mathsf{LDT}^{g}(p,n,\Delta)$ accepts with probability $1$, 
		\item if $g$ is at least $\delta$-far from every polynomial with maximum individual degree at most $\Delta$, then $\mathsf{LDT}^{g}(p,n,\Delta)$ rejects with probability at least $2/3$, and
		\item $\mathsf{LDT}$ runs in $\poly(p)$ time.
	\end{enumerate}
	
\end{lemma}

\newcommand{\Comp}{\mathsf{Comp}}

\paragraph*{Comparing two RM codewords.} Lastly, we show an efficient algorithm that, given oracle access to two codewords of $\RM_{\F_p,\Delta,m}$, computes the lexicographically first differing point between the respective messages of the two codewords.

\begin{lemma}[Comparing two RM codewords]\label{lm:comp-RM}
	Let $p$ be a prime. Let $m,\Delta \in \N$ be such that $m \cdot \Delta < p/2$. There is a probabilistic oracle algorithm $\Comp$ that takes two polynomials $f,g \colon \F_p^m \to \F_p$ as oracles, such that if both $\degmax(f)$ and $\degmax(g)$ are at most $\Delta$, then the following holds with probability at least $9/10$:
	\begin{itemize}
		\item If $f \ne g$, then $\Comp^{f,g}(p,m,\Delta)$ outputs the lexicographically smallest element $w$ in $H^m$ such that $f(w) \ne g(w)$, where $H = \{0,1,\dotsc,\Delta-1\}$.\footnote{Since both $f$ and $g$ have max degree at most $\Delta$, their values are completely determined by their restrictions on $H^m$. Hence, if $f\ne g$, then such $w$ must exist.}
		\item If $f = g$, then $\Comp^{f,g}(p,m,\Delta)$ outputs $\bot$.
		\item $\Comp$ makes at most $\poly(m \cdot \Delta)$ queries to both $f$ and $g$, and runs in $\poly(m \cdot \Delta \cdot \log p)$ time.
	\end{itemize}
\end{lemma}
\begin{proof}
	Our proof is similar to the proof from~\cite{Hirahara15}, which only considers multi-linear polynomials. Our algorithm $\Comp^{f,g}(p,m,\Delta)$ works as follows:
	\begin{enumerate}
		\item The algorithm has $m$ stages, where the $i$-th stage aims to find the $i$-th entry of $w$. At the end of the $i$-th stage, the algorithm obtains a length-$i$ prefix of $w$.
		
		\item For every $i \in [m]$:
		\begin{enumerate}
			\item Let $w_{<i} \in H^{i-1}$ be the current prefix. For every $h \in \{0,1,\dotsc,\Delta-1\}$, we run a randomized polynomial identity test to check whether the restricted polynomial $f(w_{<i},h,\cdot)$ and $g(w_{<i},h,\cdot)$ are the same, with error at most $\frac{1}{10 m|H|}$.\footnote{Note that these two polynomials have total degree at most $m \cdot \Delta < p/2$. Hence if they are different, their values on a random element from $\F_p^{m-i}$ are different with probability at least $1/2$. Hence the desired error level can be achieved by sampling $O(\log m + \log \Delta)$ random points from $\F^{m-i}$ and checking whether $f(w_{<i},h,\cdot)$ and $g(w_{<i},h,\cdot)$ have the same values.}
			\item We set $w_i$ to be the smallest $h$ such that our test above reports that $f(w_{<i},h,\cdot)$ and $g(w_{<i},h,\cdot)$ are distinct. If there is no such $h$, we immediately return $\bot$.
		\end{enumerate}
	\end{enumerate}

	By a union bound, all $m H$ polynomial identity testings are correct with probability at least $9/10$. In this case, if $f = g$, then the algorithm outputs $\bot$ in the first stage. If $f \ne g$, by induction on $i$, we can show that for every $i \in [m]$, $w_{\le i}$ is the lexicographically smallest element from $H^{m}$ such that $f(w_{\le i},\cdot)$ and $g(w_{\le i},\cdot)$ are distinct, which implies that the output $w$ is also the lexicographically smallest element $w$ in $H^m$ such that $f(w) \ne g(w)$. 
\end{proof}

\subsection{Encoded History and \texorpdfstring{$\S_2\BPP$}{S2BPP} Verification}

Next, we define the following encoded history.

\begin{definition}\label{defi:encoded-history-Korten}
	Let $C\colon\{0, 1\}^n \to \{0, 1\}^{2n}$ be a circuit, and $f\in\{0, 1\}^T$ be a ``hard truth table'' in the sense that $f\not\in\Range(\GGM_T[C])$. Let $k$, $(\istar, \jstar)$, and $\{v_{i,j}\}_{i,j}$ be defined as in~\autoref{algo: Korten-algo}. Let $S$ be the concatenation of ${\sf enc}(v_{i, j})$ for every $i \in \zeroUpto{k-1}$, $j \in \zeroUpto{2^i-1}$, in the reserve lexicographical order of $(i, j)$, padded with zeros at the end to length exactly $5 T$. (Recall that ${\sf enc}(v_{i, j})$ was defined in \autoref{def: computation history of Korten}.) Let $p$ be the smallest prime that is at least $20 \cdot \log^5 T$, and $m$ be the smallest integer such that $(\log T)^{m} \ge 5 \cdot T$.
	
	The \emph{encoded computational history} of $\Kor(C,f)$, denoted as 
	\[
	\EncHisKor(C,f),
	\]
	consists of $(\istar, \jstar)$, concatenated with $\BRM_{\F_p,\log T,m}(S)$.
	
	The length of the encoded history is at most 
	\[
	\mleft\lceil \log(40 \cdot \log^5 T) \mright\rceil \cdot (40 \cdot \log^5 T)^{\log (5T) / \log\log T + 1} + 2 \log T \le T^6
	\]
	for all sufficiently large $T$, and for convenience we always pad zeros at the end so that its length becomes exactly $T^6$.\footnote{For simplicity even for $T$ such that the length of the encoded history is longer than $T^6$, we will pretend its length is exactly $T^6$ throughout this section. This does not affect the analysis in our main theorem~\autoref{theo:S-2-algo-for-range-avoidance} since there we only care about sufficiently large $T$.}
\end{definition}

Recall that the original computational history $\HisKor(C,f)$ is simply the concatenation of $(\istar,\jstar)$ and $S$. In the encoded version, we encode its $S$ part by the Reed--Muller code instead. In the rest of this section, when we say history, we always mean the encoded history $\EncHisKor(C,f)$ instead of the vanilla history $\HisKor(C,f)$.

We need the following lemma. 

\begin{lemma}\label{lemma:Kor-prop-s2}
	Let $n,T \in \N$ be such that $\log T \le n \le T$. Let $C\colon\{0, 1\}^n \to \{0, 1\}^{2n}$ be a circuit and $f\in\bs{T} \setminus \Range(\GGM_T[C])$. Let $h \coloneqq \EncHisKor(C, f)$ and $z \coloneqq \Kor(C, f)$.
	
	\begin{enumerate}
		\item {\bf (history contains input/output)} There is a $\poly(\log T)$-time oracle algorithm $\In$ and an $O(n)$-time oracle algorithm $\Out$, both of which have input parameters $T, n$ and take a string $\tilde{h} \in \bs{T^6}$ as oracle, such that the following hold:\label{item: Input-Output-S2}
		
		\begin{enumerate}
			\item $\In_{T,n}$ makes a single query to its oracle; when given $h$ as the oracle, $\In_{T,n}$ takes an additional input $i \in \zeroUpto{T^6-1}$ and outputs $f_i$. \label{item: history-contain-input-S2}
			
			\item $\Out_{T,n}$ makes at most $4n$ queries to its oracle; when given $h$ as the oracle, $\Out_{T,n}$ outputs $z = \Kor(C, f)$. \label{item: history-contain-output-S2}
		\end{enumerate}
			
		\item {\bf ($\S_2\BPP$ verification of the history)} There is a randomized oracle algorithm $V$ with input parameters $T,n$ such that the following hold: \label{item: S2BPP-V}
		
		\begin{enumerate}
			\item $V$ takes strings $\tilde{f} \in \bs{T}, \pi_1,\pi_2 \in \bs{T^6}$ as oracles, the circuit $C$, an integer $i \in \bkts{T^6}$, and $\eps \in (0,1)$ as input, and runs in $\poly(n,\log \eps^{-1})$ time. %
			
			\item For every $\pi \in \bs{T^6}$ and every $i \in \zeroUpto{T^6 - 1}$, we have that
			\[
			\Pr\bkts{V^{f,\pi,h}_{T, n}(C,i,\eps) = h_i} \ge 1-\eps \quad\text{and}\quad \Pr\bkts{V^{f,h,\pi}_{T, n}(C,i,\eps) = h_i} \ge 1-\eps.
			\]
		\end{enumerate}
	\end{enumerate}
\end{lemma}
\begin{proof}
	\newcommand{\wtg}{\WT{g}}
	Again, the algorithms $\In_{T,n}$ and $\Out_{T,n}$ can be constructed in a straightforward way.\footnote{To see that $\Out_{T,n}$ makes at most $4n$ queries: Note that $\Out$ first reads the pair $(\istar,\jstar)$ from $h$, and then reads two corresponding blocks from $v_{i,j}$ encoded in $h$. In total, it reads at most $2\log T + 2n \le 4n$ bits from $h$.} %
	So we focus on the construction of $V$. Let $p,m,k \in \N$ be as in~\autoref{defi:encoded-history-Korten}. We also set $\F = \F_p$ and $\Delta = \log T$ in the rest of the proof.
	
	Our $V$ always first \emph{selects} one of the oracles $\pi_1$ and $\pi_2$ (say $\pi_\mu$ for $\mu \in \{1,2\}$), and then outputs $\pi_\mu(i)$. Hence, in the following, we say that $V$ selects $\pi_\mu$ to mean that $V$ outputs $\pi_\mu(i)$ and terminates. Given $\pi_1$ and $\pi_2$, let $g_1,g_2 \colon \F^{m} \to \F$ be the (potential) RM codewords encoded in $\pi_1$ and $\pi_2$, respectively.\footnote{Technically $\pi_1$ and $\pi_2$ are supposed to contain the RM codewords concatenated with $\Enc_{\F_p} \colon \F_p \to \bs{\lceil \log p \rceil}$.} From now on, we will assume that $i$ points to an entry in the encoded history $g_1$ or $g_2$ instead of the encoded pair of integers $(\istar,\jstar)$. We will discuss the other case at the end of the proof.
	
	\paragraph*{Low-max-degree test and self-correction.}  We first run $\mathsf{LDT}^{g_1}(p,m,\Delta - 1)$ and $\mathsf{LDT}^{g_2}(p,m,\Delta - 1)$ for $c_1$ times, where $c_1$ is a sufficiently large constant. Recall that $p \ge 20 \cdot \log^5 T$, $m = \mleft\lceil \log \mleft( 5T \mright) /  \log\log T \mright\rceil$, and $\Delta = \log T$. It follows that $p \ge 20 \cdot ( (\Delta - 1) + 1 )^2 \cdot m^2$, which satisfies the condition of~\autoref{lm:poly-LDT}. We also note that $3 m^2 \cdot ((\Delta - 1) + 1) / p < 1/4$. Hence, by~\autoref{lm:poly-LDT}, if $g_1$ is $1/4$-far from all polynomials with maximum individual degree at most $\Delta-1$, then $\mathsf{LDT}^{g_1}(p,m,\Delta - 1)$ rejects with probability $2/3$, and similarly for $g_2$.
	
	Now, if any of the runs on $\mathsf{LDT}^{g_1}(p,m,\Delta - 1)$ rejects, $V$ selects $\pi_2$, and if any of the runs on $\mathsf{LDT}^{g_2}(p,m,\Delta - 1)$ rejects, $V$ selects $\pi_1$.\footnote{As a minor detail, if both $g_1$ and $g_2$ are rejected by some runs, $V$ selects $\pi_2$.} In other words, $V$ first \emph{disqualifies} the oracles that do not pass the low-max-degree test. We set $c_1$ to be large enough so that conditioning on the event that $V$ does not terminate yet, with probability at least $0.99$, both $g_1$ and $g_2$ are $1/4$-close to polynomials $\wtg_1\colon \F_p^m \to \F$ and $\wtg_2\colon \F_p^m \to \F$, respectively, where $\degmax(\wtg_1)$ and $\degmax(\wtg_2)$ are at most $\Delta - 1$.
	
	We can then use $\mathsf{PCorr}^{g_1}(p,m,m \cdot (\Delta - 1),\cdot)$ and $\mathsf{PCorr}^{g_2}(p,m,m \cdot (\Delta - 1),\cdot)$ to access the polynomials $\wtg_1$ and $\wtg_2$. (Note that $m \cdot (\Delta-1) < p/3$, which satisfies the condition of~\autoref{lm:poly-SC}). We repeat them each $O(\log T + \log m)$ times to make sure that on a single invocation, they return the correct values of $\wtg_1$ and $\wtg_2$ respectively with probability at least $1 - 1/(mT)^{c_2}$ for a sufficiently large constant $c_2$. By~\autoref{lm:poly-SC}, each call to $\mathsf{PCorr}^{g_1}(p,m,m \cdot (\Delta - 1),\cdot)$ or $\mathsf{PCorr}^{g_2}(p,m,m \cdot (\Delta - 1),\cdot)$ takes $\polylog(T)$ time.
	
	\paragraph*{Selecting the better polynomial.}

	From now on, we \textbf{refine} what it means when $V$ selects $\pi_\mu$: now it means that $V$ outputs the bit corresponding to $i$ in $\WT{g}_\mu$ (recall that we are assuming that $i$ points to an entry in the encoded history $g_1$ or $g_2$).
	
	Let $\{v^{1}_{i,j}\}$ and $\{v^{2}_{i,j}\}$ be the encoded histories in $\WT{g}_1$ and $\WT{g}_2$. Then $V$ uses $\Comp^{\WT{g}_1,\WT{g}_2}(p,m,\Delta-1)$ to find the lexicographically largest $(i',j')$ such that $v^{1}_{i',j'} \ne v^{2}_{i',j'}$.\footnote{Recall that the $\{v_{i,j}\}$ is encoded in the reverse lexicographic order (\autoref{defi:encoded-history-Korten}).} Note that $\Comp^{\WT{g}_1,\WT{g}_2}(p,m,\Delta-1)$ makes at most $\poly(m \cdot \Delta)$ queries to both $\WT{g}_1$ and $\WT{g}_2$.  By making $c_2$ large enough, we know that $\Comp$ operates correctly with probability at least $0.8$. By operating correctly, we mean that (1) if $\wtg_1 \ne \wtg_2$, $\Comp$ finds the correct $(i',j')$ and (2) If $\wtg_1 = \wtg_2$, $\Comp$ returns $\bot$.\footnote{From~\autoref{lm:comp-RM}, $\Comp^{\WT{g}_1,\WT{g}_2}(p,m,\Delta-1)$ itself operates correctly with probability at least $0.9$. But the access to $\WT{g}_1$ (similarly to $\WT{g}_2$) is provided by $\mathsf{PCorr}^{g_1}(p,m,m \cdot (\Delta - 1),\cdot)$, which may err with probability at most $1/(mT)^{c_2}$. So we also need to take a union bound over all the bad events that a query from $\Comp$ to $\WT{g}_1$ or $\WT{g}_2$ is incorrectly answered.}
	
	In what follows, we assume that $\Comp$ operates correctly. If $\Comp$ returns $\bot$, then $V$ simply selects $\pi_1$. Otherwise, there are several cases:

	\def\ipred{i_{\circ}}
	\def\jpred{j_{\circ}}
	
	\begin{enumerate}
		\item $i' = k$. In this case, $\wtg_1$ and $\wtg_2$ disagree on their leaf values, which intend to encode $f$. $V$ queries $f$ to figure out which one has the correct value, and selects the corresponding oracle. (Note that at most one of them can be consistent with $f$.  If none of them are consistent, then $V$ selects $\pi_1$.)
		
		From now on, assume $i' < k$ and set $\alpha = v^{1}_{i'+1,2j'}\circ v^{1}_{i'+1,2j'+1}$. Note that by the definition of $(i',j')$, it holds that $\alpha = v^{2}_{i'+1,2j'}\circ v^{2}_{i'+1,2j'+1}$ as well.
		
		\item $i' < k$ and both $v^{1}_{i',j'}$ and $v^{2}_{i',j'}$ are not $\bot$. In this case, $V$ first checks whether both of them are in $C^{-1}(\alpha)$ (it can be checked by testing whether $C(v^{1}_{i',j'}) = \alpha$ and $C(v^{2}_{i',j'}) = \alpha$). If only one of them is contained in $C^{-1}(\alpha)$, $V$ selects the corresponding oracle. If none of them are contained, $V$ selects $\pi_1$. Finally, if both are contained in $C^{-1}(\alpha)$, $V$ checks which one is lexicographically smaller, and selects the corresponding oracle.
		
		\item $i' < k$, and one of the $v^{1}_{i',j'}$ and $v^{2}_{i',j'}$ is $\bot$. Say that $v^b_{i',j'} = \bot$ for some $b\in\{1, 2\}$, and denote $\bar{b} := 3-b$ as the index of the other proof. In this case, let $(\ipred, \jpred)$ denote the predecessor of $(i', j')$ in terms of the reverse lexicographical order (that is, the smallest pair that is lexicographically greater than $(i', j')$). Since $\Comp$ operates correctly, we have that $v^1_{\ipred, \jpred} = v^2_{\ipred, \jpred}$. If $v^1_{\ipred, \jpred} = \bot$, then $\pi_{\bar{b}}$ has to be incorrect (since by \autoref{def: computation history of Korten}, $\bot$'s form a contiguous suffix of the history), and $V$ selects $\pi_b$. Otherwise, if $v^{\bar{b}}_{i', j'} \in C^{-1}(\alpha)$, then $\pi_b$ is incorrect (as it claims that $C^{-1}(\alpha) = \varnothing$), and $V$ selects $\pi_{\bar{b}}$. Otherwise, $V$ selects $\pi_b$.
	\end{enumerate}

	\paragraph*{Analysis.} Now we show that $\Pr\bkts{V_{T,n}^{f,h,\pi}(i) = h(i)} \ge 2/3$. (The proof for $\Pr\bkts{V_{T,n}^{f,\pi,h}(i) = h(i)} \ge 2/3$ is symmetric.) To get the desired $\eps$ error probability, one can simply repeat the above procedure $O(\log 1/\eps)$ times and output the majority.
	
	First, by~\autoref{lm:poly-LDT}, $\mathsf{LDT}^{g_1}(p,m,\Delta - 1)$ passes with probability $1$. If some of the runs of $\mathsf{LDT}^{g_2}(p,m,\Delta - 1)$ rejects, then $V$ selects $h$. Otherwise, we know that with probability at least $0.99$, $\mathsf{PCorr}^{g_1}(p,m,m \cdot (\Delta - 1),\cdot)$ and $\mathsf{PCorr}^{g_2}(p,m,m \cdot (\Delta - 1),\cdot)$ provide access to polynomials $\wtg_1$ and $\wtg_2$ with maximum individual degree at most $\Delta - 1$, where $\wtg_1$ encodes the correct history values $\{v_{i,j}\}_{i,j}$ of $\Kor(C,f)$.
	
	Then, assuming $\Comp$ operates correctly (which happens with probability at least $0.8$), if $\wtg_1 = \wtg_2$, then the selection of $V$ does not matter. Now we assume $\wtg_1 \ne \wtg_2$.
	
	We will verify that in all three cases above, $h$ (as the first oracle) is selected by $V$. In the first case, by definition, all leaf values in $h$ are consistent with $f$, and hence $h$ is selected. In the second case, since $h$ contains the correct history values, we know that $v^1_{i',j'}$ must be the smallest element from $C^{-1}(\alpha)$, so again $h$ is selected. In the last case: (1) if $v^1_{\ipred, \jpred} = \bot$, then $v^1_{i', j'}$ has to be $\bot$ as well, thus $h$ is selected; (2) if $v^1_{\ipred, \jpred} \ne \bot$ and $v^1_{i', j'} = \bot$, then $C^{-1}(\alpha) = \varnothing$, and since the other proof $\pi$ claims some element $v^2_{i', j'}\in C^{-1}(\alpha)$, $h$ is selected; and (3) if $v^1_{\ipred, \jpred} \ne\bot$ and $v^1_{i', j'} \ne \bot$, then $\pi$ claims that $C^{-1}(\alpha) = \varnothing$ and we can check that $v^1_{i', j'} \in C^{-1}(\alpha)$, therefore $h$ is selected as well.

	\paragraph*{The remaining case: $i$ points to the location of $(\istar,\jstar)$.} In this case, $V$ still runs the algorithm described above to make a selection. Indeed, if $\Comp$ does not return $\bot$, $V$ operates exactly the same. But when $\Comp$ returns $\bot$, $V$ cannot simply select $\pi_1$ since we need to make sure that $V$ selects the oracle corresponding to $h$  (it can be either $\pi_1$ or $\pi_2$). Hence, in this case, $V$ first reads $(\istar^1,\jstar^1)$ and $(\istar^2,\jstar^2)$ from $\pi_1$ and $\pi_2$. If they are the same, $V$ simply selects $\pi_1$. Otherwise, for $b \in [2]$, $V$ checks whether $v^b_{\istar^b,\jstar^b} = \bot$, and select the one that satisfies this condition. (If none of the $v^b_{\istar^b,\jstar^b}$ are, then $V$ selects $\pi_1$). If both of $v^b_{\istar^b,\jstar^b}$ are $\bot$, $V$ selects the $\mu \in [2]$ such that $(\istar^\mu,\jstar^\mu)$ is larger.

	Now, we can verify that $V_{T,n}^{f,h,\pi}$ selects $h$ with high probability as well. (To see this, note that in the correct history, $(\istar,\jstar)$ points to the lexicographically largest all-zero block.)
 
 	Finally, the running time bound follows directly from the description of $V$.
\end{proof}

\subsubsection{A remark on relativization}
Perhaps surprisingly, although~\autoref{lemma:Kor-prop-s2} heavily relies on arithmetization tools such as Reed--Muller encoding and low-degree tests, it in fact also relativizes. To see this, the crucial observation is that, similarly to~\autoref{lemma:Kor-prop-sigma2}, the verifier $V$ from~\autoref{lemma:Kor-prop-s2} only needs \emph{black-box access} to the input circuit $C$, meaning that it only needs to evaluate $C$ on certain chosen input. Hence, when $C$ is actually an oracle circuit $C^\oracle$ for some arbitrary oracle $\oracle$, the only modification we need is that $V$ now also takes $\oracle$ as an oracle.

\begin{remark}\label{rem:Kor-prop-s2-rel}
	\autoref{defi:encoded-history-Korten} and~\autoref{lemma:Kor-prop-s2} \emph{relativize}, in the sense that if $C$ is an oracle circuit $C^{\oracle}$ for some arbitrary oracle,~\autoref{defi:encoded-history-Korten} needs no modification since~\autoref{def: computation history of Korten} relativizes, and~\autoref{lemma:Kor-prop-s2} holds with the only modification that $V$ now also needs to take $\oracle$ as an oracle (since it needs to evaluate $C$).
\end{remark}

Indeed, the remark above might sound strange at first glance: arguments that involve PCPs often do not \emph{relativize}, and the encoded history $\EncHisKor(C, f)$ looks similar to a PCP since it enables $V$ to perform a probabilistic local verification. However, a closer inspection reveals a key difference: the circuit $C$ is always treated as a black box\textemdash both in the construction of history (\autoref{def: computation history of Korten}) and in the construction of the encoded history (\autoref{defi:encoded-history-Korten}). That is, the arithmetization in the encoded history \emph{does not arithmetize} the circuit $C$ itself.

\subsection{Lower Bounds for \texorpdfstring{$\S_2\E$}{S2E}}\label{sec:S2E:S2E-lb}

Let $\{C_n:\{0, 1\}^n \to \{0, 1\}^{2n}\}$ be a $\P$-uniform family of circuits. We show that there is a single-valued $\mathsf{F}\S_2\P$ algorithm $\calA$ such that for infinitely many $n\in\N$, on input $1^n$, $\calA(1^n)$ outputs a canonical string that is outside the range of $C_n$.

\begin{theorem}\label{theo:S-2-algo-for-range-avoidance}
	Let $\{C_n \colon \{0, 1\}^n \to \{0, 1\}^{2n}\}_{n \in \N}$ be a $\P$-uniform family of circuits. There is a sequence of valid outputs $\{y_n\in\bs{2n} \setminus \Range(C_n)\}_{n \in \N}$ and a single-valued $\mathsf{F}\S_2\P$ algorithm $A$ with one bit of advice, such that for infinitely many $n \in \N$, $A(1^n)$ outputs $y_n$.
\end{theorem}
\begin{proof}
	Our proof proceeds similarly to the proof of the previous~\autoref{theo:Sigma-2-algo-for-range-avoidance}. We will follow the same notation.
	
	\paragraph*{Notation.}
	Let $n^{(1)}$ be a large enough power of $2$, $n^{(\ell)} = 2^{2^{n^{(\ell-1)}}}$ for each integer $\ell > 1$. Let $n_0^{(\ell)} = n^{(\ell)}$ and $t^{(\ell)} = O\mleft(\log n_0^{(\ell)}\mright)$ be parameters that we set later. For each $1\le i\le \tl$, let $\nell_i := \mleft(\nell_{i-1}\mright)^{10}$. To show our algorithm $A$ works on infinitely many input lengths, we will show that for every $\ell \in \N$, there is an input length $\nell_{i}$ for some $i \in \bkts{\tl}$ such that $A$ works.
	
	Fix $\ell \in \N$. From now on, for convenience, we will use $n_i$ and $t$ to denote $\nell_i$ and $\tl$, respectively.
	
	\paragraph*{Specifying $T_i$ and $f_i$.} For each input length $n_i$, we will specify a parameter $T_i \in \N$ and a string $f_i \in \bs{T_i}$. Our win-win analysis is based on whether $f_i \in \Range(\GGM_{T_i}[C_{n_i}])$ for each $i \in \zeroUpto{t}$.
	
	Let $T_0 := 2^{2n_0}\cdot 2n_0$ and $f_0$ be the concatenation of all length-$2n_0$ strings (which has length $T_0$). From \autoref{fact: brute force for sigma2}, we have that $f_0 \not \in \Range(\GGM_{T_0}[C_{n_0}])$. For every $i \in [t]$, we define $$f_i = \EncHisKor(C_{n_{i-1}}, f_{i-1}).$$ From \autoref{defi:encoded-history-Korten}, this also means that we have set $T_i = T_{i-1}^6$ for every $i \in [t]$.
	
	Let $t$ be the first integer such that $T_{t+1} \le n_{t+1}$. Note that we have $T_i = (T_0)^{6^i} \le 2^{3 n_0 \cdot 6^i}$ and $n_i = (n_0)^{10^i} = 2^{\log n_0 \cdot 10^i}$. Hence, we have that $t \le O(\log n_0)$. (Also note that $n_t^{(\ell)} < n_0^{(\ell + 1)}$.)
	
	\paragraph*{Description of our $\mathsf{F}\S_2\P$ algorithm $A$.} Now, let $k \in \{0,1,\dotsc,t\}$ be the largest integer such that $f_k\not\in \Range(\GGM_{T_k}[C_{n_k}])$. Since $f_0\not\in \Range(\GGM_{T_0}[C_{n_0}])$, such a $k$ must exist. Let $z := \Kor(C_{n_k}, f_k)$, it follows from \autoref{lemma: Korten's reduction} that $z$ is not in the range of $C_{n_k}$ (i.e., $z \in \bs{2n_k} \setminus \Range(C_{n_k})$). Our single-valued $\mathsf{F}\S_2\P$ algorithm $A$ computes $z$ on input $1^{n_k}$ (see~\autoref{defi:sv-algos}). %
	
	We will first construct an $\S_2\BPP$ verifier $V$ that computes $z$ in polynomial time on input $1^{n_k}$, and then use the fact that all $\S_2\BPP$ verifiers can be turned into equivalent $\S_2 \P$ verifiers with a polynomial-time blow-up~\cite{Canetti96, RussellS98}, from which we can obtain the desired verifier $V_A$ for $A$.
	
	\paragraph*{Description of an $\S_2\BPP$ verifier $V$ computing $z$.} 
	Formally, $V$ is a randomized polynomial-time algorithm that takes $1^{n_k}$ and two witnesses $\pi_1,\pi_2 \in \bs{n_{k+1}}$ as input, and we aim to establish the following:
	 \begin{mdframed}
	 	\small
	 	There exists $\omega \in \bs{n_{k+1}}$ such that for every $\pi \in \bs{n_{k+1}}$, we have
	 	\[
	 	\Pr[V(1^{n_k},\omega,\pi) = z] \ge 2/3\qquad\text{and}\qquad \Pr[V(1^{n_k},\pi,\omega) = z] \ge 2/3,
	 	\]
   where the probabilities are over the internal randomness of $V$.
	 \end{mdframed}
 
 \newcommand{\hfa}{\hat{f}}
 \newcommand{\hfb}{\hat{g}}
 
	 In more detail, if $k < t$, then $V$ treats $\pi_1$ and $\pi_2$ as inputs to the circuit $\GGM_{T_{k+1}}[C_{n_{k+1}}]$, and let 
 \[
 \hfa_{k+1} := \GGM_{T_{k+1}}[C_{n_{k+1}}](\pi_1)\quad\text{and}\quad\hfb_{k+1} := \GGM_{T_{k+1}}[C_{n_{k+1}}](\pi_2).
 \]
 Here, the lengths of $\pi_1$ and $\pi_2$ are $\ell \coloneqq n_{k+1} \le \poly(n_k)$. If $k = t$, then $V$ defines $\hfa_{k+1} := \pi_1$, $\hfb_{k+1} := \pi_2$, and their lengths are $\ell \coloneqq T_{t+1} \le n_{k+1} \le \poly(n_k)$. It is intended that one of the $\hfa_{k+1}$ and $\hfb_{k+1}$ is $f_{k+1} = \EncHisKor(C_{n_k}, f_k)$ ($V$ needs to figure out which one). 
 
 We now specify the intended proof $\omega \in \bs{n_{k+1}}$. When $k < t$, since $f_{k+1} \in \Range(\GGM_{T_{k+1}}[C_{n_{k+1}}])$, we can set $\omega$ so that $\GGM_{T_{k+1}}[C_{n_{k+1}}](\omega) = f_{k+1}$. When $k = t$, we simply set $\omega = f_{k+1}$.
 
 Note that \autoref{lemma: GGMEval} provides us ``random access'' to the (potentially very long) strings $\hfa_{k+1}$ and $\hfb_{k+1}$: (take $\hfa_{k+1}$ as an example) given $\pi_1$ and $j\in \zeroUpto{T_{k+1}-1}$, one can compute the $j$-th bit of $\hat{f}_{k+1}$ in $\poly(n_k)$ time. Also recall from~\autoref{lemma:Kor-prop-s2} that for each $i$, $f_{i+1} = \EncHisKor(C_{n_i}, f_i)$ contains the string $f_i$, which can be retrieved by the oracle algorithm $\In$ described in \autoref{item: Input-Output-S2} of \autoref{lemma:Kor-prop-s2}. Therefore, for each $i$ from $k$ downto $1$, we can recursively define $\hfa_i$ such that $(\hfa_i)_j = \In_{T_i, n_i}^{\hfa_{i+1}}(j)$ (similarly for $\hfb_i$). We also define $\hfa_0$ and $\hfb_0$ to be the concatenation of all length-$(2n_0)$ strings in the lexicographical order, so $\hfa_0 = \hfb_0 = f_0$. 
 
 Applying the algorithm $\In$ recursively, we obtain two algorithms $F$ and $G$ (depending on $\pi_1$ and $\pi_2$, respectively) that given $i \in \zeroUpto{k+1}$ and $j \in \zeroUpto{T_i-1}$, output the $j$-th bit of $\hfa_i$ or $\hfb_i$, respectively. Since $\In$ only makes one oracle query, these algorithms run in $\poly(n_k)$ time.%
	
 We are now ready to formally construct $V$. We first recursively define a series of procedures $V_0,\dotsc, V_{k+1}$, where each $V_i$ takes an input $j$ and outputs (with high probability) the $j$-th bit of $f_i$. Let $V_0$ be the simple algorithm that, on input $j$, computes the $j$-th bit of $f_0$. For every $i \in [k+1]$, we define 
 \[
 V_i(\alpha) \coloneqq \Select_{T_{i-1},n_{i-1}}^{V_{i-1},\hfa_{i},\hfb_{i}}(C_{n_{i-1}}, \alpha,\eps_i)
 \]
 
 for some $\eps_i \in [0,1)$ to be specified later,
 where $\Select$ is the algorithm in \autoref{item: S2BPP-V} of \autoref{lemma:Kor-prop-s2}.
 We note that since $V_{i-1}$ is a randomized algorithm, when $V_i$ calls $V_{i-1}$, it also draws \emph{independent} random coins used by the execution of $V_{i-1}$. Moreover, all calls to $\hfa_{i}$ and $\hfb_{i}$ in $V_i$ can be simulated by calling our algorithms $F$ and $G$. Jumping ahead, we remark that $V_i$ is supposed to compute $f_i$ when at least one of $\hfa_{i}$ or $\hfb_{i}$ is $f_{i}$. We then set 
 \[
 V(1^{n_k}, \pi_1, \pi_2) \coloneqq \Out_{T_{k},n_{k}}^{V_{k+1}}
 \]
 (note that $V_{k+1}$ is defined from $\hfa_{k+1}$ and $\hfb_{k+1}$, which are in turn constructed from $\pi_1$ and $\pi_2$), where $\Out_{T_{k},n_{k}}$ is the algorithm from~\autoref{item: Input-Output-S2} of~\autoref{lemma:Kor-prop-s2}.
 
 \paragraph*{Correctness of $V$.} Let $\tau \in \N$ be a large constant such that $\Select_{T,n}$  runs in $(n \cdot \log 1/\eps)^{\tau}$ time. In particular, on any input $\alpha$, $\Select_{T_{i-1},n_{i-1}}^{V_{i-1},\hfa_{i},\hfb_{i}}(C_{n_{i-1}}, \alpha,\eps_i)$ makes at most $(n_{i-1} \cdot \log 1/\eps_i)^{\tau}$ many queries to $V_{i-1}$. 
 
 We say $\Select^{f,\pi_1,\pi_2}_{T, n}(C, \alpha, \eps_i)$ makes an error if the following statements hold ($h = \EncHisKor(C,f)$ from~\autoref{lemma:Kor-prop-s2}):\footnote{The condition below only applies when at least one of $\pi_1$ and $\pi_2$ is $h$. If neither of them are $h$, then $\Select$ by definition never errs.}
 \[
 \mleft[ \pi_1 = h\quad \mathsf{OR}\quad \pi_2 = h \mright] \quad\mathsf{AND}\quad \mleft[ \Select^{f,\pi_1,\pi_2}_{T, n}(C_{n_{i-1}}, \alpha, \eps_i) \neq h_\alpha \mright].
 \]
 
 Similarly, we say that $\Select_{T_{i-1},n_{i-1}}^{V_{i-1},\hfa_{i},\hfb_{i}}(C_{n_{i-1}}, \alpha,\eps_i)$ makes an error if either (1) one of the queries to $V_{i-1}$ are incorrectly answered (i.e., the answer is not consistent with $f_{i-1}$) or (2) all queries are correctly answered but $\Select_{T_{i-1},n_{i-1}}^{f_{i-1},\hfa_{i},\hfb_{i}}(C_{n_{i-1}}, \alpha,\eps_i)$ makes an error. Note that (2) happens with probability at most $\eps_i$ from~\autoref{item: S2BPP-V} of~\autoref{lemma:Kor-prop-s2}.
 
 Now we are ready to specify the parameter $\eps_i$. We set $\eps_{k+1} =1 / (100 \cdot n_{k+1})$, and for every $i \in \{0,1,\dotsc,k\}$, we set 
 \[
 \eps_i = \frac{\eps_{i+1}}{4 \cdot (n_{i} \cdot \log 1/\eps_{i+1})^{\tau}}.
 \]
 
 To show the correctness of $V$, we prove the following claim by induction.
 
 \begin{claim}\label{claim:V_i-correct}
 	Assume either $\hfa_{k+1} = f_{k+1}$ or $\hfb_{k+1} = f_{k+1}$. For every $i \in \{0,1,\dotsc,k+1\}$ and $\alpha \in \bkts{|f_i|}$, $V_i(\alpha)$ outputs $f_i(\alpha)$ with probability at least $1 - 2 \eps_i$.
 \end{claim}
 \begin{claimproof}
 	The claim certainly holds for $V_0$. Now, for $i \in [k+1]$, assuming it holds for $V_{i-1}$, it follows that $\Select_{T_{i-1},n_{i-1}}^{V_{i-1},\hfa_{i},\hfb_{i}}(C_{n_{i-1}}, \alpha,\eps_i)$ makes an error with probability at most 
 	\[
 	\eps_i + (n_{i-1} \cdot \log 1/\eps_i)^{\tau} \cdot 2\eps_{i-1} \le 2\eps_i.
 	\]
 	By the definition of making an error and our assumption that either $\hfa_{k+1} = f_{k+1}$ or $\hfb_{k+1} = f_{k+1}$ (from which we know either $\hfa_{i} = f_{i}$ or $\hfb_{i} = f_{i}$), it follows that $V_i(\alpha)$ outputs $f_i(\alpha)$ with probability at least $1 - 2 \eps_i$.
 \end{claimproof}

Note that $\Out_{T_{k},n_{k}}^{V_{k+1}}$ makes at most $4n_k$ queries to $V_{k+1}$. It follows from~\autoref{claim:V_i-correct} that when either $\hfa_{k+1} = f_{k+1}$ or $\hfb_{k+1} = f_{k+1}$, we have that $V(1^{n_k}, \pi_1, \pi_2)$ outputs $z$ with probability at least $1 - (4n_k) \cdot 1/(100 n_{k+1}) \ge 2/3$. The correctness of $V$ then follows from our choice of $\omega$.

\paragraph*{Running time of $V$.} Finally, we analyze the running time of $V$, for which we first need to bound $\log \eps_i^{-1}$. First, we have
\[
\log \eps_{k+1}^{-1} = \log n_{k+1} + \log 100.
\]
By our definition of $\eps_i$ and the fact that $\tau$ is a constant, we have
\begin{align*}
\log \eps_{i}^{-1} &= \log \eps_{i+1}^{-1} + \log  4 + \tau \cdot \mleft(\log n_{i} + \log\log \eps_{i+1}^{-1} \mright)\\
						   &\le 2 \log \eps_{i+1}^{-1} + O(\log n_{i}).
\end{align*}

Expanding the above and noting that $k \le t \le O(\log n_0)$, for every $i \in [k+1]$ we have that
\[
\log \eps_{i}^{-1} \le 2^{k} \cdot O\mleft(\sum_{\ell=0}^{k} \log n_{\ell}\mright) \le \poly(n_0) \cdot \log n_k.
\]

Now we are ready to bound the running time of the $V_i$. First $V_0$ runs in $T_0 = \poly(n_0)$ time. For every $i \in [k+1]$, by the definition of $V_i$, we know that $V_i$ runs in time
\[
T_i = O\Big( (n_{i-1} \cdot \log 1/\eps_i)^\tau \Big) \cdot (T_{i-1} + n_k^{\beta} + 1),
\]
where $\beta$ is a sufficiently large constant and $n_k^\beta$ bounds the running time of answering each query $\Select_{T_{i-1},n_{i-1}}^{V_{i-1},\hfa_{i},\hfb_{i}}(C_{n_{i-1}}, \alpha,\eps_i)$ makes to $\hfa_{i}$ or $\hfb_{i}$, by running $F$ or $G$, respectively.

Expanding out the bound for $T_k$, we know that $V_{k+1}$ runs in time
\[
2^{O(k)} \cdot \mleft( \poly(n_0) \cdot \log n_k \mright)^{O(k \cdot \tau)} \cdot n_k^\beta \cdot \prod_{i=1}^{k+1} n_{i-1}^\tau.
\]
Since $n_k = n_{0}^{10^k}$ and $k \le O(\log n_0)$, the above can be bounded by $\poly(n_k)$. This also implies that $V$ runs in $\poly(n_k)$ time as well, which completes the analysis of the $\S_2\BPP$ verifier $V$.

\paragraph*{Derandomization of the $\S_2\BPP$ verifier $V$ into the desired $\S_2\P$ verifier $V_A$.} Finally, we use the underlying proof technique of $\S_2\BPP = \S_2\P$~\cite{Canetti96, RussellS98} to derandomize $V$ into a deterministic $\S_2\P$ verifier $V_A$ that outputs $z$.

By repeating $V$ $\poly(n_k)$ times and outputs the majority among all the outputs, we can obtain a new $\S_2\BPP$ verifier $\WT{V}$ such that
\begin{itemize}
	\item There exists $\omega \in \bs{n_{k+1}}$ such that for every $\pi \in \bs{n_{k+1}}$, we have
	\begin{equation}\label{eq:prob-WT-V}
		\Pr[\WT{V}(1^{n_k},\omega,\pi) = z] \ge  1 - 2^{-n_k} \qquad\text{and}\qquad \Pr[\WT{V}(1^{n_k},\pi,\omega) = z] \ge 1 - 2^{-n_k}.
	\end{equation}
\end{itemize}

Let $\ell = \poly(n_k)$ be an upper bound on the number of random coins used by $\WT{V}$. We also let $m := \poly(\ell, n_{k+1}) \le \poly(n_k)$ and use $\WT{V}(1^{n_k},\pi_1,\pi_2;r)$ to denote the output of $\WT{V}$ given randomness $r$. Now, we define $V_A$ as follows: It takes two vectors  $\vec{\pi}_1, \vec{\pi}_2 \in \bs{n_{k+1}} \times \mleft( \bs{\ell} \mright)^{m}$ as proofs. For $\vec{\pi}_1 = (\alpha,u_1,u_2,\dotsc,u_m)$ and $\vec{\pi}_2 = (\beta,v_1,v_2,\dotsc,v_m)$, $V_A$ outputs the majority of the multi-set
\[
\{ \WT{V}(1^{n_k},\alpha,\beta;u_i \oplus v_j) \}_{(i,j) \in [m]^2},
\]
where $u_i \oplus v_j$ denotes the bit-wise XOR of $u_i$ and $v_j$ (if no strings occur more than $m^2/2$ times in the set above, then $V_A$ simply outputs $\bot$).

We will show there exists $\vec{\omega} = (\gamma,r_1,\dotsc,r_m)$ such that for every $\vec{\pi} \in \bs{n_{k+1}} \times \mleft( \bs{\ell} \mright)^{m}$,
\[
\Pr[V_A(1^{n_k},\vec{\omega},\vec{\pi}) = z]\quad\text{and}\quad \Pr[V_A(1^{n_k},\vec{\pi},\vec{\omega}) = z].
\]

We first claim that there exist $r_1,\dotsc,r_m \in \bs{\ell}$ such that for every $u \in \bs{\ell}$ and for every $\pi \in \bs{n_{k+1}}$, it holds that (1) for at least a $2/3$ fraction of $i \in [m]$, we have $\WT{V}(1^{n_k},\omega,\pi;r_i \oplus u) = z$ and (2) for at least a $2/3$ fraction of $i \in [m]$, we have $\WT{V}(1^{n_k},\pi,\omega;r_i \oplus u) = z$.

To see this, for every fixed $u \in \bs{\ell}$  and $\pi \in \bs{n_{k+1}}$, by a simple Chernoff bound, the probability, over $m$ independently uniformly drawn $r_1,\dotsc,r_m$, that more than a $1/3$ fraction of $i \in [m]$ satisfies $\WT{V}(1^{n_k},\omega,\pi;r_i \oplus u) \ne z$ is at most $2^{-\Omega(m)}$, and the same probability upper bound holds for the corresponding case of $\WT{V}(1^{n_k},\pi,\omega;r_i \oplus u) \ne z$ as well. Our claim then just follows from a simple union bound over all $u \in \bs{\ell}$ and $\pi \in \bs{n_{k+1}}$.

Now, let $\gamma$ be the proof $\omega$ such that the condition~\eqref{eq:prob-WT-V} holds. We simply set $\vec{\omega} = (\gamma,r_1,\dotsc,r_m)$. From our choice of $\gamma$ and $r_1,\dotsc,r_m$, it then follows that for every $v_1,\dotsc,v_m \in \bs{\ell}$ and $\pi \in \bs{n_{k+1}}$, at least a $2/3$ fraction of $ \WT{V}(1^{n_k},\gamma,\pi;r_i \oplus v_j)$ equals $z$, and similarly for $ \WT{V}(1^{n_k},\pi,\gamma;r_i \oplus v_j)$. This completes the proof.

\paragraph*{Wrapping up.} Finally, we generalize $A$ and $V_A$ to work on all inputs $1^n$. On input $1^n$, $V_A$ calculates the largest $\ell$ such that $n^{(\ell)} \le n$, and also calculates the largest ${k'}$ such that $n_{k'}^{(\ell)} \le n$. If $n_{k'}^{(\ell)}\ne n$, then $V_A$ immediately outputs $\bot$ and halts. Otherwise, $V_A$ receives an advice bit indicating whether $k' = k^{(\ell)}$, where $k^{(\ell)}$ is the largest integer such that $f_{k^{(\ell)}}^{(\ell)} \not \in \Range(\GGM_{T_k^{(\ell)}}[C_{n_k^{(\ell)}}])$. If this is the case, then $V_A$ runs the verification procedure above; otherwise, it immediately outputs $\bot$ and halts. It is easy to see that $V_A$ runs in $\poly(n)$ time, and is an infinitely-often single-valued $\mathsf{F}\S_2\P$ algorithm solving the range avoidance problem of $\{C_n\}$.
\end{proof}

Moreover, observe that in the proof of~\autoref{lemma:Kor-prop-s2}, all considered input lengths (the $n^{(\ell)}_i$) are indeed powers of $2$. So we indeed have the following slightly stronger result.

\begin{corollary}\label{cor:S-2-algo-input-power-2}
	Let $\{C_n \colon \{0, 1\}^n \to \{0, 1\}^{2n}\}_{n \in \N}$ be a $\P$-uniform family of circuits. There is a single-valued $\mathsf{F}\S_2\P$ algorithm $A$ with one bit of advice such that for infinitely many $r \in \N$, letting $n = 2^r$, $A(1^n)$ outputs $y_n \in \bs{2n} \setminus \Range(C_n)$.
\end{corollary}

We need the following reduction from Korten which reduces solving range avoidance with one-bit stretch to solving range avoidance with doubling stretch.

\begin{lemma}[{\cite[Lemma~3]{Korten21}}]\label{lemma:stretch-1-Korten}
	Let $n \in \N$. There is a polynomial time algorithm $A$ and an $\FP^\NP$ algorithm $B$ such that the following hold:
	\begin{enumerate}
		\item Given a circuit $C \colon \bs{n} \to \bs{n+1}$, $A(C)$ outputs a circuit $D\colon \bs{n} \to \bs{2n}$.
		\item Given any $y \in \bs{2n} \setminus \Range(D)$, $B(C,y)$ outputs a string $z \in \bs{n+1} \setminus \Range(C)$.
	\end{enumerate}
\end{lemma}

The following corollary then follows by combining~\autoref{lemma:stretch-1-Korten} and~\autoref{theo:S-2-and-PNP}.

\begin{corollary}\label{cor:S-2-algo-input-power-2-stretch-1}
	Let $\{C_n \colon \{0, 1\}^n \to \{0, 1\}^{n+1}\}_{n \in \N}$ be a $\P$-uniform family of circuits. There is a single-valued $\mathsf{F}\S_2\P$ algorithm $A$ with one bit of advice such that for infinitely many $r \in \N$, letting $n = 2^r$, $A(1^n)$ outputs $y_n \in \bs{n+1} \setminus \Range(C_n)$.
\end{corollary}

The following corollary follows from~\autoref{fact: range avoidance for TT implies circuit lower bounds} and~\autoref{cor:S-2-algo-input-power-2-stretch-1}. %

\begin{corollary}
	$\S_2\E/_1 \not\subset \SIZE[2^n/n]$.  %
\end{corollary}

Finally, we also note that by letting $C_n$ be a universal Turing machine mapping $n$ bits to $n+1$ bits in $\poly(n)$ time, we have the following strong lower bounds for $\S_2\E/_1$ against non-uniform time complexity classes with maximum advice.

\begin{corollary}
	For every $\alpha(n) \ge \omega(1)$ and any constant $k \ge 1$, $\S_2\E/_1 \not\subset \mathsf{TIME}[2^{kn}]/_{2^n - \alpha(n)}$.
\end{corollary}

From~\autoref{rem:Kor-prop-s2-rel} and noting that the derandomization of $\S_2\BPP$ verifier $V$ to $\S_2\P$ verifier $A_V$ also relativizes, we can see that all the results above relativize as well.

\subsection{Infinitely Often Single-Valued \texorpdfstring{$\mathsf{F}\S_2\P$}{FS2P} Algorithm for Arbitrary Input Range Avoidance}\label{sec:S2E:range-avoid-algo}

\autoref{theo:S-2-algo-for-range-avoidance} and~\autoref{cor:S-2-algo-input-power-2-stretch-1} only give single-valued $\mathsf{F}\S_2\P$ algorithms for solving range avoidance for $\P$-uniform families of circuits. Applying Korten's reduction~\cite{Korten21}, we show that it can be strengthened into a single-valued infinitely-often $\mathsf{F}\S_2\P$ algorithm solving range avoidance given an arbitrary input circuit.

We need the following reduction from~\cite{Korten21}.

\begin{lemma}[{\cite[Theorem~7]{Korten21}}]\label{lemma:Korten-PNP}
	There is an $\FP^\NP$ algorithm $A_{\Kor}$ satisfying the following:
	\begin{enumerate}
		\item $A_{\Kor}$ takes an $s$-size circuit $C \colon \bs{n} \to \bs{n+1}$ and a truth table $f\in \bs{2^{m}}$ such that $2^m \ge s^3$ and $n \le s$ as input.
		
		\item If the circuit complexity of $f$ is at least $c_1 \cdot m \cdot s$ for a sufficiently large universal constant $c_1 \in \N$, then $A_{\Kor}(C,f)$ outputs a string $y \in \bs{n+1} \setminus \Range(C)$.
	\end{enumerate}
\end{lemma}

\begin{theorem}
	There is a single-valued $\mathsf{F}\S_2\P$ algorithm $A$ with one bit of advice such that for infinitely many $s \in \N$, for all $s$-size circuits $C \colon \bs{n} \to \bs{n+1}$ where $n \le s$, $A(C)$ outputs $y_C \in \bs{n+1} \setminus \Range(C)$.
\end{theorem}
\begin{proof}[Proof Sketch]
	By~\autoref{cor:S-2-algo-input-power-2-stretch-1}, there is a single-valued $\mathsf{F}\S_2\P$ algorithm $W$ with one bit of advice such that for infinitely many $n \in \N$, $W(1^{2^n})$ outputs a string $f_n \in \bs{2^n}$ with $\SIZE(f_n) \ge 2^n/n$.
	
	Now we construct our single-valued $\mathsf{F}\S_2\P$ algorithm $A$ with one bit of advice as follows: given an $s$-size circuit $C \colon \bs{n} \to \bs{n+1}$ with $n \le s$ as input; let $m = \lceil \log s^3 \rceil$ and $f_m = W(1^{2^m})$; output $A_{\Kor}(C,f_m)$. It follows from~\autoref{theo:S-2-and-PNP} that $A$ is a single-valued $\mathsf{F}\S_2\P$ algorithm with one bit of advice (the advice of $A$ is given to $W$).
\end{proof}

Finally, $\S_2\P \subseteq \ZPP^\NP$~\cite{Cai01a} implies that every single-valued $\mathsf{F}\S_2\P$ algorithm can also be implemented as a single-valued $\mathsf{F}\ZPP^\NP$ algorithm with polynomial overhead. Therefore, the above theorem also implies an infinitely often $\mathsf{F}\ZPP^\NP$ algorithm for range avoidance.

\begin{reminder}{\autoref{cor:ZPP-P-stretch-1}}
	There is a single-valued $\mathsf{F}\ZPP^{\NP}$ algorithm $A$ with one bit of advice such that for infinitely many $s \in \N$, for all $s$-size circuits $C \colon \bs{n} \to \bs{n+1}$ where $n \le s$, $A(C)$ outputs $y_C \in \bs{n+1} \setminus \Range(C)$. That is, for all those $s$, there is a string $y_C \in \bs{n+1} \setminus \Range(C)$ such that $A(C)$ either outputs $y_C$ or $\bot$, and the probability (over the inner randomness of $A$) that $A(C)$ outputs $y_C$ is at least $2/3$.
\end{reminder}

%% file: main.bbl
\newcommand{\etalchar}[1]{$^{#1}$}
 \newcommand{\proc}{Proc. } \newcommand{\stoc}[1]{\proc #1 Annual ACM Symposium
  on Theory of Computing \emph{({STOC})}} \newcommand{\focs}[1]{\proc #1 Annual
  IEEE Symposium on Foundations of Computer Science \emph{({FOCS})}}
  \newcommand{\cccieee}[1]{\proc #1 Annual IEEE Conference on Computational
  Complexity ({CCC})} \newcommand{\ccc}[1]{\proc #1 Computational Complexity
  Conference ({CCC})} \newcommand{\sct}[1]{\proc #1 Annual Structure in
  Complexity Theory Conference} \newcommand{\icalp}[1]{\proc #1 International
  Colloquium on Automata, Languages and Programming ({ICALP})}
  \newcommand{\soda}[1]{\proc #1 Annual ACM-SIAM Symposium on Discrete
  Algorithms ({SODA})} \newcommand{\crypto}[1]{\proc #1 Annual International
  Cryptology Conference ({CRYPTO})} \newcommand{\podc}[1]{\proc #1 ACM
  Symposium on Principles of Distributed Computing ({PODC})}
  \newcommand{\spaa}[1]{\proc #1 ACM Symposium on Parallelism in Algorithms and
  Architectures} \newcommand{\apprx}[1]{\proc #1 International Workshop on
  Approximation Algorithms for Combinatorial Optimization Problems ({APPROX})}
  \newcommand{\rnd}[1]{\proc #1 International Workshop on Randomization and
  Approximation Techniques in Computer Science ({RANDOM})}
  \newcommand{\fsttcs}[1]{\proc #1 Annual Conference on Foundations of Software
  Technology and Theoretical Computer Science ({FSTTCS})}
  \newcommand{\mfcs}[1]{\proc #1 International Symposium on Mathematical
  Foundations of Computer Science ({MFCS})} \newcommand{\itcs}[1]{\proc #1
  {C}onference on {I}nnovations in {T}heoretical {C}omputer {S}cience ({ITCS})}
  \newcommand{\sosa}[1]{\proc #1 Symposium on Simplicity in Algorithms
  ({SOSA})} \newcommand{\issac}[1]{\proc #1 International Symposium on Symbolic
  and Algebraic Computation ({ISSAC})} \newcommand{\esa}[1]{\proc #1 European
  Symposium on Algorithms ({ESA})} \newcommand{\tccconf}[1]{\proc #1 Theory of
  Cryptography Conference ({TCC})} \newcommand{\csr}[1]{\proc #1 International
  Computer Science Symposium in Russia ({CSR})} \newcommand{\cocoon}[1]{\proc
  #1 International Computing and Combinatorics Conference ({COCOON})}
  \newcommand{\stacs}[1]{\proc #1 Symposium on Theoretical Aspects of Computer
  Science ({STACS})} \newcommand{\wads}[1]{\proc #1 International Symposium on
  Algorithms and Data Structures ({WADS})} \newcommand{\swat}[1]{\proc #1
  Scandinavian Symposium and Workshop on Algorithm Theory ({SWAT})}
  \newcommand{\latin}[1]{\proc #1 Latin American Theoretical Informatics
  Symposium ({LATIN})} \newcommand{\isaac}[1]{\proc #1 International Symposium
  on Algorithms and Computation ({ISAAC})} \newcommand{\eccc}{Electronic
  Colloquium on Computational Complexity {\emph{(ECCC)}}}
  \newcommand{\jacm}{Journal of the ACM} \newcommand{\coco}{Computational
  Complexity} \newcommand{\jcss}{Journal of Computer and System Sciences}
  \newcommand{\siamj}{{SIAM} Journal of Computing}
  \newcommand{\ipl}{Information Processing Letters} \newcommand{\tocj}{Theory
  of Computing} \newcommand{\tcs}{Theoretical Computer Science}
  \newcommand{\toct}{ACM Transactions on Computation Theory}
  \newcommand{\toit}{IEEE Transactions on Information Theory}
  \newcommand{\talg}{ACM Transactions on Algorithms}
\begin{thebibliography}{CMMW19}

\bibitem[Aar06]{Aaronson06}
Scott Aaronson.
\newblock Oracles are subtle but not malicious.
\newblock In {\em {CCC}}, pages 340--354. {IEEE} Computer Society, 2006.
\newblock \href {https://doi.org/10.1109/CCC.2006.32}
  {\path{doi:10.1109/CCC.2006.32}}.

\bibitem[AB09]{AB09-book}
Sanjeev Arora and Boaz Barak.
\newblock {\em Computational Complexity: {A} Modern Approach}.
\newblock Cambridge University Press, 2009.

\bibitem[AKS04]{Agrawal02primesis}
Manindra Agrawal, Neeraj Kayal, and Nitin Saxena.
\newblock {PRIMES} is in {$\P$}.
\newblock {\em Annals of Mathematics}, 160(2):781--793, 2004.
\newblock \href {https://doi.org/10.4007/annals.2004.160.781}
  {\path{doi:10.4007/annals.2004.160.781}}.

\bibitem[BCG{\etalchar{+}}96]{BshoutyCGKT96}
Nader~H. Bshouty, Richard Cleve, Ricard Gavald{\`{a}}, Sampath Kannan, and
  Christino Tamon.
\newblock Oracles and queries that are sufficient for exact learning.
\newblock {\em J. Comput. Syst. Sci.}, 52(3):421--433, 1996.
\newblock \href {https://doi.org/10.1006/jcss.1996.0032}
  {\path{doi:10.1006/jcss.1996.0032}}.

\bibitem[BFL91]{BabaiFL91}
L{\'{a}}szl{\'{o}} Babai, Lance Fortnow, and Carsten Lund.
\newblock Non-deterministic exponential time has two-prover interactive
  protocols.
\newblock {\em \coco}, 1:3--40, 1991.
\newblock \href {https://doi.org/10.1007/BF01200056}
  {\path{doi:10.1007/BF01200056}}.

\bibitem[BFNW93]{BabaiFNW93}
L{\'{a}}szl{\'{o}} Babai, Lance Fortnow, Noam Nisan, and Avi Wigderson.
\newblock {$\BPP$} has subexponential time simulations unless {${\sf EXPTIME}$}
  has publishable proofs.
\newblock {\em Computational Complexity}, 3:307--318, 1993.
\newblock \href {https://doi.org/10.1007/BF01275486}
  {\path{doi:10.1007/BF01275486}}.

\bibitem[BFT98]{BuhrmanFT98}
Harry Buhrman, Lance Fortnow, and Thomas Thierauf.
\newblock Nonrelativizing separations.
\newblock In {\em CCC}, pages 8--12, 1998.
\newblock \href {https://doi.org/10.1109/CCC.1998.694585}
  {\path{doi:10.1109/CCC.1998.694585}}.

\bibitem[BLS85]{BookLS85}
Ronald~V. Book, Timothy~J. Long, and Alan~L. Selman.
\newblock Qualitative relativizations of complexity classes.
\newblock {\em J. Comput. Syst. Sci.}, 30(3):395--413, 1985.
\newblock \href {https://doi.org/10.1016/0022-0000(85)90053-4}
  {\path{doi:10.1016/0022-0000(85)90053-4}}.

\bibitem[BS06]{Buresh-OppenheimS06}
Joshua Buresh{-}Oppenheim and Rahul Santhanam.
\newblock Making hard problems harder.
\newblock In {\em CCC}, pages 73--87. {IEEE} Computer Society, 2006.
\newblock \href {https://doi.org/10.1109/CCC.2006.26}
  {\path{doi:10.1109/CCC.2006.26}}.

\bibitem[Cai07]{Cai01a}
Jin{-}yi Cai.
\newblock {$\S_2^p\subseteq\ZPP^\NP$}.
\newblock {\em J. Comput. Syst. Sci.}, 73(1):25--35, 2007.
\newblock \href {https://doi.org/10.1016/j.jcss.2003.07.015}
  {\path{doi:10.1016/j.jcss.2003.07.015}}.

\bibitem[Can96]{Canetti96}
Ran Canetti.
\newblock More on {$\BPP$} and the polynomial-time hierarchy.
\newblock {\em Inf. Process. Lett.}, 57(5):237--241, 1996.
\newblock \href {https://doi.org/10.1016/0020-0190(96)00016-6}
  {\path{doi:10.1016/0020-0190(96)00016-6}}.

\bibitem[CCHO05]{CaiCHO05}
Jin{-}yi Cai, Venkatesan~T. Chakaravarthy, Lane~A. Hemaspaandra, and Mitsunori
  Ogihara.
\newblock Competing provers yield improved {K}arp--{L}ipton collapse results.
\newblock {\em Inf. Comput.}, 198(1):1--23, 2005.
\newblock \href {https://doi.org/10.1016/j.ic.2005.01.002}
  {\path{doi:10.1016/j.ic.2005.01.002}}.

\bibitem[CLO{\etalchar{+}}23]{ChenLORS23}
Lijie Chen, Zhenjian Lu, Igor~C. Oliveira, Hanlin Ren, and Rahul Santhanam.
\newblock Polynomial-time pseudodeterministic construction of primes.
\newblock In {\em {FOCS}}, 2023.
\newblock To appear.
\newblock \href {https://doi.org/10.48550/arXiv.2305.15140}
  {\path{doi:10.48550/arXiv.2305.15140}}.

\bibitem[CMMW19]{ChenMMW19}
Lijie Chen, Dylan~M. McKay, Cody~D. Murray, and R.~Ryan Williams.
\newblock Relations and equivalences between circuit lower bounds and
  {K}arp--{L}ipton theorems.
\newblock In {\em {CCC}}, volume 137 of {\em LIPIcs}, pages 30:1--30:21.
  Schloss Dagstuhl - Leibniz-Zentrum f{\"{u}}r Informatik, 2019.
\newblock \href {https://doi.org/10.4230/LIPIcs.CCC.2019.30}
  {\path{doi:10.4230/LIPIcs.CCC.2019.30}}.

\bibitem[CT21a]{ChenT21b}
Lijie Chen and Roei Tell.
\newblock Hardness vs randomness, revised: Uniform, non-black-box, and
  instance-wise.
\newblock In {\em FOCS}, pages 125--136, 2021.
\newblock \href {https://doi.org/10.1109/FOCS52979.2021.00021}
  {\path{doi:10.1109/FOCS52979.2021.00021}}.

\bibitem[CT21b]{ChenT21}
Lijie Chen and Roei Tell.
\newblock Simple and fast derandomization from very hard functions: eliminating
  randomness at almost no cost.
\newblock In {\em STOC}, pages 283--291, 2021.
\newblock \href {https://doi.org/10.1145/3406325.3451059}
  {\path{doi:10.1145/3406325.3451059}}.

\bibitem[CZ19]{CZ19}
Eshan Chattopadhyay and David Zuckerman.
\newblock {Explicit two-source extractors and resilient functions}.
\newblock {\em Annals of Mathematics}, 189(3):653--705, 2019.
\newblock \href {https://doi.org/10.4007/annals.2019.189.3.1}
  {\path{doi:10.4007/annals.2019.189.3.1}}.

\bibitem[Erd59]{Erdos59}
Paul Erd{\H{o}}s.
\newblock Graph theory and probability.
\newblock {\em Canadian Journal of Mathematics}, 11:34--38, 1959.
\newblock \href {https://doi.org/10.4153/CJM-1959-003-9}
  {\path{doi:10.4153/CJM-1959-003-9}}.

\bibitem[FHOS93]{FHOS93}
Stephen~A. Fenner, Steven Homer, Mitsunori Ogiwara, and Alan~L. Selman.
\newblock On using oracles that compute values.
\newblock In {\em {STACS}}, volume 665 of {\em Lecture Notes in Computer
  Science}, pages 398--407. Springer, 1993.
\newblock \href {https://doi.org/10.1007/3-540-56503-5\_40}
  {\path{doi:10.1007/3-540-56503-5\_40}}.

\bibitem[FM05]{FrandsenM05}
Gudmund~Skovbjerg Frandsen and Peter~Bro Miltersen.
\newblock Reviewing bounds on the circuit size of the hardest functions.
\newblock {\em \ipl}, 95(2):354--357, 2005.
\newblock \href {https://doi.org/10.1016/j.ipl.2005.03.009}
  {\path{doi:10.1016/j.ipl.2005.03.009}}.

\bibitem[GG11]{GatG11}
Eran Gat and Shafi Goldwasser.
\newblock Probabilistic search algorithms with unique answers and their
  cryptographic applications.
\newblock {\em Electron. Colloquium Comput. Complex.}, {TR11-136}, 2011.
\newblock URL: \url{https://eccc.weizmann.ac.il/report/2011/136}.

\bibitem[GGM86]{GoldreichGM86}
Oded Goldreich, Shafi Goldwasser, and Silvio Micali.
\newblock How to construct random functions.
\newblock {\em \jacm}, 33(4):792--807, 1986.
\newblock \href {https://doi.org/10.1145/6490.6503}
  {\path{doi:10.1145/6490.6503}}.

\bibitem[GGNS23]{GGNS23}
Karthik Gajulapalli, Alexander Golovnev, Satyajeet Nagargoje, and Sidhant
  Saraogi.
\newblock Range avoidance for constant depth circuits: Hardness and algorithms.
\newblock In {\em {APPROX/RANDOM}}, volume 275 of {\em LIPIcs}, pages
  65:1--65:18. Schloss Dagstuhl - Leibniz-Zentrum f{\"{u}}r Informatik, 2023.
\newblock \href {https://doi.org/10.4230/LIPIcs.APPROX/RANDOM.2023.65}
  {\path{doi:10.4230/LIPIcs.APPROX/RANDOM.2023.65}}.

\bibitem[GLW22]{GuruswamiLW22}
Venkatesan Guruswami, Xin Lyu, and Xiuhan Wang.
\newblock Range avoidance for low-depth circuits and connections to
  pseudorandomness.
\newblock In {\em {APPROX/RANDOM}}, volume 245 of {\em LIPIcs}, pages
  20:1--20:21. Schloss Dagstuhl - Leibniz-Zentrum f{\"{u}}r Informatik, 2022.
\newblock \href {https://doi.org/10.4230/LIPIcs.APPROX/RANDOM.2022.20}
  {\path{doi:10.4230/LIPIcs.APPROX/RANDOM.2022.20}}.

\bibitem[Gol08]{Goldreich-book}
Oded Goldreich.
\newblock {\em Computational complexity: a conceptual perspective}.
\newblock Cambridge University Press, 2008.
\newblock \href {https://doi.org/10.1017/CBO9780511804106}
  {\path{doi:10.1017/CBO9780511804106}}.

\bibitem[GS92]{DBLP:journals/ipl/GemmellS92}
Peter Gemmell and Madhu Sudan.
\newblock Highly resilient correctors for polynomials.
\newblock {\em Inf. Process. Lett.}, 43(4):169--174, 1992.
\newblock \href {https://doi.org/10.1016/0020-0190(92)90195-2}
  {\path{doi:10.1016/0020-0190(92)90195-2}}.

\bibitem[Hir15]{Hirahara15}
Shuichi Hirahara.
\newblock Identifying an honest {$\EXP^{\NP}$} oracle among many.
\newblock In {\em CCC}, volume~33 of {\em LIPIcs}, pages 244--263. Schloss
  Dagstuhl - Leibniz-Zentrum f{\"{u}}r Informatik, 2015.
\newblock \href {https://doi.org/10.4230/LIPIcs.CCC.2015.244}
  {\path{doi:10.4230/LIPIcs.CCC.2015.244}}.

\bibitem[HLR23]{HiraharaLR23}
Shuichi Hirahara, Zhenjian Lu, and Hanlin Ren.
\newblock Bounded relativization.
\newblock In {\em {CCC}}, volume 264 of {\em LIPIcs}, pages 6:1--6:45. Schloss
  Dagstuhl - Leibniz-Zentrum f{\"{u}}r Informatik, 2023.
\newblock \href {https://doi.org/10.4230/LIPIcs.CCC.2023.6}
  {\path{doi:10.4230/LIPIcs.CCC.2023.6}}.

\bibitem[HNOS96]{HNOS96}
Lane~A. Hemaspaandra, Ashish~V. Naik, Mitsunori Ogihara, and Alan~L. Selman.
\newblock Computing solutions uniquely collapses the polynomial hierarchy.
\newblock {\em {SIAM} J. Comput.}, 25(4):697--708, 1996.
\newblock \href {https://doi.org/10.1137/S0097539794268315}
  {\path{doi:10.1137/S0097539794268315}}.

\bibitem[IKV18]{ImpagliazzoKV18}
Russell Impagliazzo, Valentine Kabanets, and Ilya Volkovich.
\newblock The power of natural properties as oracles.
\newblock In {\em CCC}, volume 102 of {\em LIPIcs}, pages 7:1--7:20, 2018.
\newblock \href {https://doi.org/10.4230/LIPIcs.CCC.2018.7}
  {\path{doi:10.4230/LIPIcs.CCC.2018.7}}.

\bibitem[IKW02]{ImpagliazzoKW02}
Russell Impagliazzo, Valentine Kabanets, and Avi Wigderson.
\newblock In search of an easy witness: exponential time vs. probabilistic
  polynomial time.
\newblock {\em J. Comput. Syst. Sci.}, 65(4):672--694, 2002.
\newblock \href {https://doi.org/10.1016/S0022-0000(02)00024-7}
  {\path{doi:10.1016/S0022-0000(02)00024-7}}.

\bibitem[IW97]{ImpagliazzoW97}
Russell Impagliazzo and Avi Wigderson.
\newblock {$\P = \BPP$} if {$\E$} requires exponential circuits: Derandomizing
  the {XOR} lemma.
\newblock In {\em STOC}, pages 220--229. {ACM}, 1997.
\newblock \href {https://doi.org/10.1145/258533.258590}
  {\path{doi:10.1145/258533.258590}}.

\bibitem[Je{\v{r}}04]{Jerabek04}
Emil Je{\v{r}}{\'{a}}bek.
\newblock Dual weak pigeonhole principle, {Boolean} complexity, and
  derandomization.
\newblock {\em Ann. Pure Appl. Log.}, 129(1-3):1--37, 2004.
\newblock \href {https://doi.org/10.1016/j.apal.2003.12.003}
  {\path{doi:10.1016/j.apal.2003.12.003}}.

\bibitem[Kan82]{Kannan82}
Ravi Kannan.
\newblock Circuit-size lower bounds and non-reducibility to sparse sets.
\newblock {\em Inf. Control.}, 55(1-3):40--56, 1982.
\newblock \href {https://doi.org/10.1016/S0019-9958(82)90382-5}
  {\path{doi:10.1016/S0019-9958(82)90382-5}}.

\bibitem[KC00]{KabanetsC00}
Valentine Kabanets and Jin{-}Yi Cai.
\newblock Circuit minimization problem.
\newblock In {\em STOC}, pages 73--79, 2000.
\newblock \href {https://doi.org/10.1145/335305.335314}
  {\path{doi:10.1145/335305.335314}}.

\bibitem[KKMP21]{KKMP21}
Robert Kleinberg, Oliver Korten, Daniel Mitropolsky, and Christos
  Papadimitriou.
\newblock Total functions in the polynomial hierarchy.
\newblock In {\em ITCS}, volume 185 of {\em LIPIcs}, pages 44:1--44:18, 2021.
\newblock \href {https://doi.org/10.4230/LIPIcs.ITCS.2021.44}
  {\path{doi:10.4230/LIPIcs.ITCS.2021.44}}.

\bibitem[KL80]{KarpL80}
Richard~M. Karp and Richard~J. Lipton.
\newblock Some connections between nonuniform and uniform complexity classes.
\newblock In {\em STOC}, pages 302--309, 1980.
\newblock \href {https://doi.org/10.1145/800141.804678}
  {\path{doi:10.1145/800141.804678}}.

\bibitem[Kor21]{Korten21}
Oliver Korten.
\newblock The hardest explicit construction.
\newblock In {\em FOCS}, pages 433--444. {IEEE}, 2021.
\newblock \href {https://doi.org/10.1109/FOCS52979.2021.00051}
  {\path{doi:10.1109/FOCS52979.2021.00051}}.

\bibitem[Kra01]{Krajicek01}
Jan Kraj\'{\i}\v{c}ek.
\newblock Tautologies from pseudo-random generators.
\newblock {\em Bull. Symb. Log.}, 7(2):197--212, 2001.
\newblock \href {https://doi.org/10.2307/2687774} {\path{doi:10.2307/2687774}}.

\bibitem[Kre88]{Krentel88}
Mark~W. Krentel.
\newblock The complexity of optimization problems.
\newblock {\em J. Comput. Syst. Sci.}, 36(3):490--509, 1988.
\newblock \href {https://doi.org/10.1016/0022-0000(88)90039-6}
  {\path{doi:10.1016/0022-0000(88)90039-6}}.

\bibitem[KvM02]{kvm98}
Adam~R. Klivans and Dieter van Melkebeek.
\newblock Graph nonisomorphism has subexponential size proofs unless the
  polynomial-time hierarchy collapses.
\newblock {\em {SIAM} J. Comput.}, 31(5):1501--1526, 2002.
\newblock \href {https://doi.org/10.1137/S0097539700389652}
  {\path{doi:10.1137/S0097539700389652}}.

\bibitem[KW98]{KoblerW98}
Johannes K{\"{o}}bler and Osamu Watanabe.
\newblock New collapse consequences of {$\NP$} having small circuits.
\newblock {\em {SIAM} J. Comput.}, 28(1):311--324, 1998.
\newblock \href {https://doi.org/10.1137/S0097539795296206}
  {\path{doi:10.1137/S0097539795296206}}.

\bibitem[LFKN92]{LundFKN92}
Carsten Lund, Lance Fortnow, Howard~J. Karloff, and Noam Nisan.
\newblock Algebraic methods for interactive proof systems.
\newblock {\em \jacm}, 39(4):859--868, 1992.
\newblock \href {https://doi.org/10.1145/146585.146605}
  {\path{doi:10.1145/146585.146605}}.

\bibitem[Li23]{Li23}
Xin Li.
\newblock Two source extractors for asymptotically optimal entropy, and (many)
  more.
\newblock In {\em {FOCS}}, 2023.
\newblock To appear.
\newblock \href {https://doi.org/10.48550/arXiv.2303.06802}
  {\path{doi:10.48550/arXiv.2303.06802}}.

\bibitem[Lup58]{lupanov1958synthesis}
Oleg~B Lupanov.
\newblock On the synthesis of switching circuits.
\newblock {\em Doklady Akademii Nauk SSSR}, 119(1):23--26, 1958.

\bibitem[MVW99]{MiltersenVW99}
Peter~Bro Miltersen, N.~V. Vinodchandran, and Osamu Watanabe.
\newblock Super-polynomial versus half-exponential circuit size in the
  exponential hierarchy.
\newblock In {\em COCOON}, volume 1627 of {\em Lecture Notes in Computer
  Science}, pages 210--220. Springer, 1999.
\newblock \href {https://doi.org/10.1007/3-540-48686-0\_21}
  {\path{doi:10.1007/3-540-48686-0\_21}}.

\bibitem[NW94]{NisanW94}
Noam Nisan and Avi Wigderson.
\newblock Hardness vs randomness.
\newblock {\em \jcss}, 49(2):149--167, 1994.
\newblock \href {https://doi.org/10.1016/S0022-0000(05)80043-1}
  {\path{doi:10.1016/S0022-0000(05)80043-1}}.

\bibitem[RS98]{RussellS98}
Alexander Russell and Ravi Sundaram.
\newblock Symmetric alternation captures {$\BPP$}.
\newblock {\em Comput. Complex.}, 7(2):152--162, 1998.
\newblock \href {https://doi.org/10.1007/s000370050007}
  {\path{doi:10.1007/s000370050007}}.

\bibitem[RSW22]{RenSW22}
Hanlin Ren, Rahul Santhanam, and Zhikun Wang.
\newblock On the range avoidance problem for circuits.
\newblock In {\em {FOCS}}, pages 640--650. {IEEE}, 2022.
\newblock \href {https://doi.org/10.1109/FOCS54457.2022.00067}
  {\path{doi:10.1109/FOCS54457.2022.00067}}.

\bibitem[San09]{Santhanam09}
Rahul Santhanam.
\newblock Circuit lower bounds for {M}erlin--{A}rthur classes.
\newblock {\em {SIAM} J. Comput.}, 39(3):1038--1061, 2009.
\newblock \href {https://doi.org/10.1137/070702680}
  {\path{doi:10.1137/070702680}}.

\bibitem[Sel94]{Selman94}
Alan~L. Selman.
\newblock A taxonomy of complexity classes of functions.
\newblock {\em J. Comput. Syst. Sci.}, 48(2):357--381, 1994.
\newblock \href {https://doi.org/10.1016/S0022-0000(05)80009-1}
  {\path{doi:10.1016/S0022-0000(05)80009-1}}.

\bibitem[Sha49]{Shannon49}
Claude~E. Shannon.
\newblock The synthesis of two-terminal switching circuits.
\newblock {\em Bell System technical journal}, 28(1):59--98, 1949.
\newblock \href {https://doi.org/10.1002/j.1538-7305.1949.tb03624.x}
  {\path{doi:10.1002/j.1538-7305.1949.tb03624.x}}.

\bibitem[Sud95]{Sud95}
Madhu Sudan.
\newblock {\em Efficient Checking of Polynomials and Proofs and the Hardness of
  Approximation Problems}, volume 1001 of {\em Lecture Notes in Computer
  Science}.
\newblock Springer, 1995.
\newblock \href {https://doi.org/10.1007/3-540-60615-7}
  {\path{doi:10.1007/3-540-60615-7}}.

\bibitem[Val77]{Valiant77}
Leslie~G. Valiant.
\newblock Graph-theoretic arguments in low-level complexity.
\newblock In {\em MFCS}, volume~53 of {\em Lecture Notes in Computer Science},
  pages 162--176, 1977.
\newblock \href {https://doi.org/10.1007/3-540-08353-7\_135}
  {\path{doi:10.1007/3-540-08353-7\_135}}.

\bibitem[Vin05]{Vinodchandran05}
N.~V. Vinodchandran.
\newblock A note on the circuit complexity of {$\PP$}.
\newblock {\em Theor. Comput. Sci.}, 347(1-2):415--418, 2005.
\newblock \href {https://doi.org/10.1016/j.tcs.2005.07.032}
  {\path{doi:10.1016/j.tcs.2005.07.032}}.

\bibitem[VW23]{VyasW23}
Nikhil Vyas and Ryan Williams.
\newblock On oracles and algorithmic methods for proving lower bounds.
\newblock In {\em {ITCS}}, volume 251 of {\em LIPIcs}, pages 99:1--99:26.
  Schloss Dagstuhl - Leibniz-Zentrum f{\"{u}}r Informatik, 2023.
\newblock \href {https://doi.org/10.4230/LIPIcs.ITCS.2023.99}
  {\path{doi:10.4230/LIPIcs.ITCS.2023.99}}.

\end{thebibliography}
